\documentclass[letterpaper, 11pt]{article}

\usepackage{amsthm,amssymb,stmaryrd}
\usepackage{graphics}
\usepackage{graphicx}
\usepackage{amsfonts}
\usepackage{bbm}
\usepackage{tikz}
\usepackage[plainpages=false,pdfpagelabels=false,colorlinks=true,citecolor=blue,hypertexnames=false]{hyperref}
\usepackage{color}

\usepackage[switch]{lineno}

\addtocontents{toc}{\protect\setcounter{tocdepth}{3}}

\usepackage{eucal}

\usepackage[left=1in,right=1in,top=1in,bottom=1in]{geometry}

\usepackage{dsfont}
\usepackage{mathrsfs}
\usepackage{amsmath}
\usepackage{mathtools}
\usepackage{cases}

\usepackage[all]{xy}
\usepackage{cite}
\usepackage[inline]{enumitem}

\usepackage{bm}

\usepackage{arydshln}
\usepackage{booktabs}
\usepackage{multirow}
\usepackage{listings}
\usepackage{maplestd2e}
\usepackage{enumitem}

\usepackage{graphicx}
 \usepackage{setspace}

\newcommand{\vv} { {\bf v}}
\newcommand{\va} { {\bf a}}
\newcommand{\vb} { {\bf b}}

\newcommand{\vd} { {\bf d}}
\newcommand{\ve} { {\bf e}}
\newcommand{\vr} { {\bf r}}
\newcommand{\vs} { {\bf s}}
\newcommand{\vx} { {\bf x}}
\newcommand{\vy} { {\bf y}}
\newcommand{\val} { \bm{\alpha}}
\newcommand{\vbe} { \bm{\beta}}
\newcommand{\vga} { \bm{\gamma}}
\newcommand{\vpa} { \bm{\partial}}

\newcommand{\hx} { {\hat{\bf x}}}

\newcommand*{\rom}[1]{\expandafter\@slowromancap\romannumeral #1@}

\def\F{\mathbb F}
\def\Z{\mathbb Z}
\def\N{\mathbb N}
\def\Q{\mathbb Q}

\def\a{\alpha}
\def\si{\sigma}
\def\lam{\lambda}
\def\la{\langle}
\def\ra{\rangle}
\def\<#1>{\langle#1\rangle}

\newcommand{\cR} { {\mathcal{R}}}
\newcommand{\bN} { {\mathbb{N}}}
\newcommand{\bC} { {\mathbb{C}}}
\newcommand{\bQ} { {\mathbb{Q}}}
\newcommand{\bZ} { {\mathbb{Z}}}

\newcommand{\bF} { {\mathbb{F}}}
\newcommand{\bK} { {\mathbb{K}}}
\newcommand{\bV} { {\mathbb{V}}}

\newcommand{\bE} { {\mathbb{E}}}

\DeclareMathOperator{\rank}{rank}
\DeclareMathOperator{\Span}{Span}

\def\supp{{\rm Supp}}

\def\pa{{\partial}}

\let\set\mathbbm

\def\im{\operatorname{im}}

\def\eatspace#1{#1}%--------------------------------------------------------------------
\def\step#1#2{\par\kern1pt\hangindent#2em\hangafter=1\noindent\rlap{\small#1}\kern#2em\relax\eatspace}

\newtheorem{thm}{Theorem}[section]
\def\eatspace#1{#1}
\newtheorem{cor}[thm]{Corollary}
\newtheorem{lem}[thm]{Lemma}
\newtheorem{prop}[thm]{Proposition}
\newtheorem{defn}[thm]{Definition}
\newtheorem{fact}[thm]{Fact}
\newtheorem{rem}[thm]{Remark}

\newtheorem{example}[thm]{Example}
\newtheorem{algorithm}[thm]{Algorithm}
\newtheorem{problem}[thm]{Problem}

\numberwithin{figure}{section}

\makeatletter \@addtoreset{equation}{section}

\begin{document}
\let\intro=f

	\title{Symbolic Summation of Multivariate Rational Functions
	\thanks{S. Chen and H. Fang were partially supported by the National Key R\&D Program of China (No. 2023YFA1009401), the NSFC grants (No. 12271511 and No. 11688101), CAS Project for Young Scientists in Basic Research (Grant No. YSBR-034), and the CAS Fund of the Youth Innovation Promotion Association (No. Y2022001). L.\ Du
		was supported by the Austrian FWF grant 10.55776/P31571.
}}

\author{  \bigskip
	Shaoshi Chen$^{1,2}$,\,  Lixin Du$^{3}$ ,\, Hanqian Fang$^{1,2}$\\
	$^1$KLMM,\, Academy of Mathematics and Systems Science, \\ Chinese Academy of Sciences,\\ Beijing, 100190, China\\
	$^2$School of Mathematical Sciences, \\University of Chinese Academy of Sciences,\\ Beijing, 100049, China\\
	$^3$ Institute for Algebra, Johannes Kepler University,\\ Linz, A4040,  Austria\\
\vspace{0.5cm}
	{\sf schen@amss.ac.cn,  lx.du@hotmail.com,  hqfang\_math@163.com}\\
{\bf Communicated by Evelyne Hubert}
}

\date{\today}

\maketitle

{ \large
	\begin{abstract}

Symbolic summation as an active research topic of symbolic computation provides efficient algorithmic tools
for evaluating and simplifying different types of sums arising from mathematics, computer science, physics and other areas.
Most of existing  algorithms in symbolic summation are mainly applicable to the problem with univariate inputs.
A long-term project in symbolic computation is to develop theories, algorithms and software
for the symbolic summation of multivariate functions. This paper will give complete solutions to
two challenging problems in symbolic summation of multivariate rational functions, namely the rational summability problem and the existence problem of telescopers for multivariate rational functions.
Our approach is based on the structure of Sato's isotropy groups of polynomials, which enables us to reduce the problems to testing the shift equivalence of polynomials.
Our results provide a complete solution to the discrete analogue of Picard's problem on differential forms and
can be used to detect the applicability of the Wilf-Zeilberger method to multivariate rational functions.

		
\bigskip
		
		{\bf MSC classes: }68W30, 12H05, 12H10\\
		
	\end{abstract}
}
\newpage
\tableofcontents
\newpage

\section{Introduction}\label{SECT:intro}
Symbolic summation is a classical and active research topic in symbolic computation~\cite{MCA2003}, whose
central problem is evaluating and simplifying different types of sums arising from combinatorics~\cite{PWZbook1996, KauersPaule2011, Koepf2014}, computer science~\cite{GKP1994}, theoretical physics~\cite{Schneider2006, Schneider2012} and other areas.
Similar to the integration case, there are two types of summation problems: one is the indefinite summation problem, and the other is the definite summation problem.
The two summation problems are connected by the discrete Leibniz--Newton formula.
In the early days, methods for symbolic summation were mainly summarized in calculus of finite differences (for instance one can see the books~\cite{Boole2009, Jordan1965}).
Since the early 1970s, efficient symbolic algorithms have been developed for symbolic summation~\cite[Chapter 23]{MCA2003}.
Abramov's algorithm~\cite{Abramov1971, Abramov1975, Abramov1995b} solves the indefinite summation problem for univariate rational functions.
The classical Hermite reduction for symbolic integration of rational functions was extended to the summation case by Paule via greatest factorial factorizations
in~\cite{Paule1995b} with more developments in~\cite{Pirastu1995b, Matusevich2000, AshCatoiu2005}.
The indefinite summation problem for hypergeometric terms is handled by Gosper's algorithm~\cite{Gosper1978}. For sequences in a general difference field,
the corresponding problem was studied by Karr in~\cite{Karr1981, Karr1985} with significant improvements by Schneider~\cite{Schneider2004} and recent fruitful applications in quantum field theory~\cite{Schneider2012, Schneider2013QFT}.
Most of the existing complete algorithms are mainly applicable to the problem with univariate inputs.
A long-term project in symbolic computation proposed by Andrews and Paule in~\cite{Andrews1993} is to develop theories, algorithms and software for symbolic summation of multivariate functions.
Along this way, some algorithms have been developed to deal with a special class of double sums~\cite{ChenHouMu2006} and binomial multiple sums~\cite{BostanLairezSalvy2017}.
One of the long standing open problems in this project is developing the multivariate version of Gosper's algorithm (see Problem 5.1 in~\cite{ChenKauers2017}).

In the multivariate setting, a more intrinsic way for formulating the summation problem is to use the language of
difference forms, which is a discrete analogue of differential forms.
The combinatorial theory of difference forms was first developed by Zeilberger in~\cite{Zeilberger1993} with applications
in computer-generated proofs of combinatorial identities and series convergence acceleration~\cite{Zimmermann2000}.  Motivated by problems in geometric integration,
the geometric theory of difference forms was properly developed by Mansfield and Hydon in~\cite{Mansfield2008} with a discrete analogue of de Rham cohomology.
In both theories, it is a fundamental problem to decide whether a given closed difference form is exact or not.
The corresponding problem in the differential case was first proposed by Picard in 1889 (see Note \uppercase\expandafter{\romannumeral 3} in~\cite[pp. 475-479]{Picard1897}), which is deciding whether
a given rational function $F(x, y, z) \in \bC(x, y, z)$ can be written as $F(x, y, z) = \frac{\partial P}{\partial x} + \frac{\partial Q}{\partial y} + \frac{\partial R}{\partial z}$
for some rational functions $P, Q, R\in \bC(x, y, z)$. In the language of differential forms, Picard's problem is equivalent to
deciding whether a differential 3-form  $\omega = F(x, y, z)dx dy dz$ is exact or not in the field of rational functions. Picard's problem can be naturally
extended to the case of  multivariate rational functions. Picard himself solved the problem for bivariate rational functions but it is still not completely solved in the general multivariate case. Some fundamental and significant work
had been done by Griffiths in the 1960s in his papers on periods of rational integrals~\cite{Griffiths1969a, Griffiths1969b}
and later by Dimca~\cite{Dimca1991} and more recently by Lairez~\cite{Lairez2016}.
As the first main result of our paper, we will solve the discrete analogue of Picard's
problem by generalizing Gosper's algorithm to the case of multivariate rational functions.

From the indefinite summation problem to the definite one,  creative telescoping~\cite{Zeilberger1991} is a crucial tool for computing definite sums with applications in
computer-generated proofs of combinatorial identities~\cite{Wilf1992, WilfZeilberger1990b, PWZbook1996}. One of fundamental problems related to creative telescoping is the existence problem of telescopers
for special functions which is equivalent to the termination of Zeilberger's algorithm~\cite{Zeilberger1990,Zeilberger1991}
and can be used to detect the hypertranscendence and algebraic independence
of functions defined by indefinite sums or integrals~\cite{Hardouin2008, Schneider2010}.
A sufficient condition, namely holonomicity,  on the existence of telescopers was first given by
Zeilberger in 1990 using Bernstein's theory of holonomic D-modules~\cite{Bernstein1971}.
Wilf and Zeilberger in~\cite{Wilf1992} proved that telescopers exist for
proper hypergeometric terms. However, holonomicity and properness
are only sufficient conditions. Abramov and Le~\cite{AbramovLe2002} solved the existence problem of telescopers for rational functions in two discrete variables. This work was soon extended to the
hypergeometric case by Abramov~\cite{Abramov2003}, the $q$-hypergeometric case
in~\cite{ChenHouMu2005}, the mixed rational and hypergeometric
case in~\cite{ChenSinger2012, Chen2015}, and most recently the P-recursive case in~\cite{Du2025}. The criteria on the existence of telescopers for rational functions in three variables were
given in~\cite{Chen2016, Chen2019ISSAC, Chen2021}. In arbitrary number of variables, there is no available algorithm for deciding the existence of telescopers for rational functions.
As the second main result of our paper, we will solve the existence problem of telescopers for multivariate rational functions by reducing the problem to the computation of Sato's isotropy groups
and the testing of the shift equivalence of multivariate polynomials.

The rest of this paper is organized as follows. In Section~\ref{sec:pre}, we define the existence problem of telescopers and the summability problem precisely and recall some basic complexity estimates for later use. We present a general scheme for designing algorithms to solve the shift equivalence testing problem in Section~\ref{sec:SET}, and compare our new algorithms with the other known algorithms in Section~\ref{sec:appendix}.
In Section~\ref{sec:add decomp}, we first recall the notion of isotropy groups of polynomials and their basic properties, and then introduce orbital decompositions for rational functions. We apply orbital decompositions in Section~\ref{sec:summability} to reduce the rational summability problem
for general rational functions to that for simple fractions. After this, we present a criterion on
the summability of such simple fractions. We not only decide the summability of a rational function but also construct the indefinite sums explicitly if it is summable. In Section~\ref{sec:telescopers}, we again use the structure of isotropy groups and orbital decompositions to derive a criterion for the existence of telescopers for rational functions in variables $t$ and $\vx$. Moreover, we present an algorithm for computing a telescoper if it exists.

\section{Preliminaries}\label{sec:pre}
In this section, we will recall some basic terminologies of symbolic summation and creative telescoping and overview some complexity results for later use.
\subsection{Telescopers and summability of rational functions}

Through out the paper, let $\bK$ be a field of characteristic zero and $\bK(t,\vx)$ be the field of rational functions in $t$ and $\vx=\{x_1,\ldots,x_m\}$ over $\bK$. For each $v\in \vv=\{t,x_1,\ldots,x_m\}$, the shift operator $\si_{v}$ with respect to $v$ is defined as the $\bK$-automorphism of $\bK(\vv)$ such that
\[\si_{v}(v)=v+1 \text{ and } \si_{v}(w)=w \text{ for all }w\in \vv\setminus\{v\}.\]

Let $\cR:=\bK(\vv)\<S_t,S_{x_1},\ldots, S_{x_m}>$ denote the ring of linear recurrence operators over $\bK(\vv)$, in which $S_{v_i}S_{v_j}=S_{v_j}S_{v_i}$ for all $v_i,v_j\in\vv$ and $S_v f=\si_v(f) S_v$ for any $f\in \bK(\vv)$ and $v\in\vv$.
The action of an operator $L=\sum_{i_0,i_1,\ldots,i_m\geq0}a_{i_0,i_1,\ldots,i_m}S_t^{i_0}S_{x_1}^{i_1}\cdots S_{x_m}^{i_m}\in\cR$ on a rational function $f\in \bK(\vv)$ is defined as
\[L(f)=\sum_{i_0,i_1,\ldots,i_m\geq 0} a_{i_0,i_1,\ldots,i_m}\si_t^{i_0}\si_{x_1}^{i_1}\cdots\si_{x_m}^{i_m}(f).\]
For each $v\in\vv$, the difference operator $\Delta_v$ with respect to $v$ is defined by $\Delta_v=S_v-\bf1$, where $\bf 1$ stands for the identity map on $\bK(\vv)$.

We now introduce the notion of telescopers for rational functions in $\bK(t, \vx)$.

\begin{defn}[Telescoper]
	Let $n$ be a positive integer such that $1\leq n\leq m$ and let $f\in \bK(t,\vx)$ be a rational function. A nonzero linear recurrence operator $L\in\bK(t)\<S_t>$ is called a {\em telescoper} of type $(\si_t;\si_{x_1},\ldots,\si_{x_n})$ for $f$ if there exist $g_1,\ldots,g_n\in \bK(t,\vx)$ such that
	\begin{equation*}%\label{EQ:telDEF}
		L(f)=\Delta_{x_1}(g_1)+\cdots+\Delta_{x_n}(g_n).
	\end{equation*}
	The rational functions $g_1,\ldots,g_n$ are called the {\em certificates} of $L$.
\end{defn}

\begin{problem}[Existence Problem of Telescopers]\label{PROB:telescopers}
	Given a rational function $f\in\bK(t,\vx)$ and an integer $n$ with $1\leq n \leq m$, decide whether or not $f$ has a telescoper of type $(\si_t;\si_{x_1},\ldots,\si_{x_n})$. If so, find a telescoper $L$ and its certificates $g_1, \ldots, g_n$.
\end{problem}
In order to decide the existence of telescopers for $f\in \bK(t, \vx)$, one may first use the shortcut to decide whether $L={\bf 1}$ is a telescoper for $f$.  This is equivalent to the following summability problem of $f$ in $\bF(\vx)$ with $\bF=\bK(t)$.

\begin{defn}[{Summability}]
	Let $\bF$ be a field of characteristic zero and $n$ be a positive integer such that $1\leq n\leq m$. A rational function $f\in\bF(\vx)$ is called {\em $(\si_{x_1},\ldots,\si_{x_n})$-summable} in $\bF(\vx)$ if there exist $g_1,\ldots,g_n\in\bF(\vx)$ such that
	\begin{equation*}%\label{EQ:sumDEF}
		f=\Delta_{x_1}(g_1)+\cdots+\Delta_{x_n}(g_n).
	\end{equation*}
    The rational functions $g_1,\ldots,g_n$, if they exist, are called the {\em certificates} of $f$.
\end{defn}

\begin{problem}[{Rational Summability Problem}]\label{PROB:summability problem}
	Given a rational function $f\in \bF(\vx)$ and an integer~$n$ with $1\leq n \leq m$, decide whether or not $f$ is $(\si_{x_1},\ldots,\si_{x_n})$-summable in $\bF(\vx)$. If so, find a tuple $(g_1,\ldots, g_n)$ such that the $g_i$'s are the certificates of $f$.
\end{problem}

In terms of arithmetic size, the certificate tends to be much larger than the telescoper but certificates are often not needed in applications.
Similar to the case of univariate rational summation~\cite{Gerhard2003}, we also output the certificate as a sum of the products of several rational functions applied by shift operations and difference isomorphisms (see the definitions in the specific algorithms).
Such a certificate is called an \textit{unnormalised certificate}.

The main idea of solving the summability problem is using mathematical induction to reduce the number of difference operators in this problem. To say explicitly, we shall reduce the $(\si_{x_1}, \ldots, \si_{x_n})$-summability problem for $f\in\bF(\vx)$ to the $(\si_{x_1}, \ldots, \si_{x_r})$-summability problem for   another rational function $a\in\bF(\vx)$, where $r$ is smaller than $n$ and the base field $\bF(\vx)$ in the summability problem is unchanged. Similarly, we shall reduce the existence problem of telescopers of type $(\si_t; \si_{x_1}, \ldots, \si_{x_n})$ for $f\in\bK(t, \vx)$ to the existence problem of telescopers of type $(\si_t; \si_{x_1}, \ldots, \si_{x_r})$ for some rational function $a\in\bK(t, \vx)$.

We introduce below a general definition of the summability problem and existence problem of telescopers, which plays a central role to set up the reduction process for solving Problems~\ref{PROB:summability problem} and~\ref{PROB:telescopers}. Let $G_t=\langle \si_t, \si_{x_1}, \ldots, \si_{x_n}\rangle$ be the group generated by the shift operators $\si_t,\si_{x_1},\ldots,\si_{x_n}$ under the operation of composition of functions. Then $G_t$ is a free abelian group. For any $\tau\in G_t$, the difference operator $\Delta_\tau$ is defined by
\[\Delta_\tau=S_t^{i_0}S_{x_1}^{i_1}\cdots S_{x_n}^{i_n}-{\bf 1}\quad \text{if } \tau=\si_t^{i_0}\si_{x_1}^{i_1}\cdots \si_{x_n}^{i_n}.\]
For short, we use $\Delta_{v}$ to denote $\Delta_{\si_{v}}$ for each $v\in\vv$. A finite subset $\{\tau_1,\ldots,\tau_r\}$ of $G_t$ is said to be $\Z$-{\em linearly independent} if for all $a_1,\ldots,a_r\in \Z$, we have
\[\tau_1^{a_1}\cdots\tau_r^{a_r}={\bf1} \quad \text{if and only if} \quad a_1=a_2=\cdots =a_r=0.\]

Let $G=\<\sigma_{x_1},\ldots,\sigma_{x_n}>$ be the subgroup of $G_t$ generated by shift operators $\si_{x_1},\ldots,\si_{x_n}$.  Let $\{\tau_1,\ldots,\tau_r\}(1\leq r\leq n)$ be a family of $\Z$-linearly independent elements in $G$. In general, a rational function $f\in\bF(\vx)$ is called {\em $(\tau_1,\ldots,\tau_r)$-summable} in $\bF(\vx)$ if \begin{equation*}%\label{EQ:tsumDEF}
	f=\Delta_{\tau_1}(g_1)+\cdots+\Delta_{\tau_r}(g_r)
\end{equation*}
for some $g_1,\ldots,g_r\in \bF(\vx)$.
Choose an element $\tau_0=\si_t^{k_0}\si_{x_1}^{k_1}\cdots\si_{x_n}^{k_n}\in G_t$ such that $k_0$ is nonzero. Then $\tau_0,\tau_1,\ldots,\tau_r$ are $\Z$-linearly independent in $G_t$. Let $T_0=S_t^{k_0}S_{x_1}^{k_1}\cdots S_{x_n}^{k_n}\in \cR$ be the operator corresponding to $\tau_0$. We say that a nonzero operator $L\in \bK(t)\<T_0>$ is a {\em telescoper of type} $(\tau_0;\tau_1,\ldots,\tau_r)$ for $f\in\bK(t,\vx)$ if $L(f)$ is $(\tau_1,\ldots,\tau_r)$-summable in $\bK(t,\vx)$.

Let $R$ be a ring and $\sigma\colon R\to R$ be a ring automorphism of $R$. The pair $(R,\sigma)$ is called a {\em difference ring}. If $R$ is a field, we call the pair $(R,\sigma)$ a {\em difference field}. Let $(R_1,\sigma_1)$ and $(R_2,\sigma_2)$ be two difference rings and $\phi\colon R_1\to R_2$ be a ring homomorphism. If $\phi$ satisfies the property that  $\phi\circ\sigma_1=\sigma_2\circ\phi$, that means the following diagram
%\begin{figure}
\[
\xymatrix{
	R_1\ar[d]_{\si_1} \ar[r]^{\phi}
	&R_2 \ar[d]^{\si_2}
	\\
	R_1
	\ar[r]^{\phi}
	&R_2}
\]	

commutes, then $\phi$ is called a {\em difference homomorphism}. If in addition $\phi$ is a bijection, then its inverse $\phi^{-1}$ is also a difference homomorphism. In this case, we call $\phi$ a {\em difference isomorphism}. The notion of difference isomorphisms will be used to state our summability criteria and the existence criteria of telescopers.

An operator $L\in \bK(t)\<S_t>$ is called a {\em common left multiple} of operators $L_1,\ldots,L_r\in\bK(t)\<S_t>$ if there exist $R_1,\ldots,R_r\in\bK(t)\<S_t>$ such that
\[L=R_1L_1=\cdots=R_rL_r.\]
Since $\bK(t)\<S_t>$ is a left Euclidean domain, such an operator $L$ always exists. Among all of such multiples, the monic one of smallest degree in $S_t$ is called the {\em least common left multiple} (LCLM). Efficient algorithms for computing LCLM have been developed in \cite{AbramovLeLi2005,BronsteinPetkovsek1996,BostanChyzakSalvyLi2012}.
\begin{rem}\label{REM: LCLM}
	Let $f=f_1+\cdots+f_r$ with $f_i\in \bK(t,\vx)$. If each $f_i$ has a telescoper $L_i$ of type $(\si_t;\si_{x_1},\ldots,\si_{x_n})$ for $i=1,\ldots,r$, then the LCLM of $L_i$'s is a telescoper of the same type for $f$. This fact follows from the commutativity between operators in $\bK(t)\<S_t>$ and the difference operators~$\Delta_{x_i}$'s.
\end{rem}

\subsection{Complexity estimates}

All complexity estimates of the algorithms in this paper are in terms of arithmetical operations in $\bK$, denoted by ``ops". The notation $\tilde{O}$ indicates the complexity estimates with hidden polylogarithmic factors.

Let $\vy=\{y_1,y_2,\ldots,y_r\}$ be a subset of $\vv=\{t,x_1,\ldots,x_m\}$ and $\vd=(d_1,d_2,\ldots,d_r)$ be a vector in $\bN^r$. Let $\bK[\vy]_{\vd}$ denote the set of polynomials in $\bK[\vy]$ whose degrees in $y_i$ are no more than $d_i$ for $i=1,2,\ldots,r$. Let $\bK(\vy)_\vd$ denote the set of rational functions in $\bK(\vy)$ with numerators and denominators in $\bK[\vy]_\vd$. In particular, we denote $\bK[\vy]_{\vd}$ (resp. $\bK(\vy)_{\vd}$) by $\bK[\vy]_{d}$ (resp. $\bK(\vy)_d$) for simplicity if $d_1=d_2=\cdots=d_r=d$. For a rational function $f(\vy)=p(\vy)/q(\vy)\in\bK(\vy)$, where $p(\vy)$ and $q(\vy)$ are coprime polynomials, the degree of $f(\vy)$ in $y_i$ is defined as $\max\{\deg_{y_i}(p(\vy)),\deg_{y_i}(q(\vy))\}$.  In particular, for $f(\vy)\in\bQ(\vy)$, let $\| f\|$ denote the \textit{max-norm} of $f$, i.e., the maximal absolute value of the integer coefficients appearing in the numerator and denominator of $f$ with respect to $\vy$. Similarly, for a matrix or vector $\mathcal{M}$ over $\bZ$, the \textit{max-norm} of $\mathcal{M}$, denoted by $\|\mathcal{M}\|$, is defined as the maximal absolute value of its entries.

We first recall some complexity estimates of the basic operations on univatiate polynomials and rational functions (see the books~\cite{ACT1997,MCA2003} for their proofs).
\begin{fact}\label{FACT:com-uni}
	Let $d$ be an integer in $\bN$. The following operations can be performed in $\tilde{O}(d)$ ops in $\bK$:
	\begin{enumerate}
		\item addition, multiplication and differentiation of elements in $\bK[x]_d$ and $\bK(x)_d$;
		\item computing the greatest common divisor of two elements in $\bK[x]_d$;
        \item partial fraction decomposition of an element in $\bK(x)_d$ with a given factorization of its denominator.
	\end{enumerate}
\end{fact}

Efficient algorithms for basic operations on multivariate polynomials have been developed in~\cite{Pan1994simple,VanSchost2013}. We summarize the needed results as follows.

\begin{fact}\label{FACT:com-multi}
    For a vector $\vd=(d_1,\ldots,d_m)\in\bN^m$, it takes $\tilde{O}(md_1\cdots d_m)$ ops in $\bK$ for multipoint evaluation or interpolation in $\bK[\vx]_\vd$ from the given values on $O(d_1\cdots d_m)$ points which form an $m$-dimensional tensor product grid.
\end{fact}

We recall the complexity estimates of factoring multivariate polynomials from~\cite[Theorem 3.26]{lenstra1987factoring} and~\cite[Equation~(\uppercase\expandafter{\romannumeral2})]{Mahler1962}.
\begin{fact}\label{FACT:com-factor}
    For a vector $\vd=(d_1,\ldots,d_m)\in\bN^m$ and an integer $M\in \bN^+$, let $f$ be a polynomial in $\bZ[\vx]_\vd$ with $\|f\|=M$. Then it takes $\tilde{O}((\min(\vd))^{m-1}(d_1\cdots d_m)^5(d_1\cdots d_m+\log(M)))$ ops in $\bQ$ to factor $f$ into the product of irreducible factors and the max-norms of these factors are $2^{d_1+\cdots+d_n}\sqrt{(d_1+1)\cdots(d_n+1)}M$.
\end{fact}

Let $\omega\in(2,3]$ be a feasible exponent of
matrix multiplication in $\bK$, i.e., two square matrices of order $r$ can be multiplied using $O(r^\omega)$ ops. Solving a system of linear equations is almost as hard as multiplying two matrices, see~\cite{ACT1997} for more details. In particular, one approach to solve a linear system over $\bQ$ for integer solutions is computing the Hermite normal form of its corresponding matrix, see~\cite{Chou1982}. Combining the algorithms for computing Hermite normal forms developed in~\cite{Kannan1979,Chou1982,Domich1987,Iliopoulos1989,Hafner1989,Hafner1991,Stor1996,Storjohann2000}, we obtain the following result.

\begin{fact}(See~\cite[Theorem (16.18)]{ACT1997}, \cite[Proposition 2.3 and its proof]{Hafner1991})\label{FACT:com-linear}
    A $\bK$-linear system $L$ of size $r$ can be solved in $O(r^\omega)$ ops in $\bK$. Furthermore, if $K=\bQ$ and the absolute values of the coefficients of $L$ are integers bounded by $M$, then one can find a special solution of $L$ with a $\bQ$-basis for the corresponding homogeneous linear system over $\bZ^r$ using $\tilde{O}(r^{\omega+1}\log(M))$ bit ops and the max-norms of these vectors are $O(r^{r/2}M^{r})$.
\end{fact}

\section{Shift equivalence testing of polynomials}\label{sec:SET}

In this section, we first state the problem of Shift Equivalence Testing (SET) and give an overview of our algorithm for solving the SET problem in Section~\ref{SUBSEC: SET overview}. The idea of our algorithm is inspired by the DOS algorithm~\cite{Dvir2014, Dvir2014complete}. Then we develop a general scheme for designing algorithms to solve the SET problem, whose proof is given in Section~\ref{SUBSEC: SET main proof}. More precisely, we introduce admissible covers
of the associated polynomial system with the SET problem and prove that every admissible cover corresponds to an algorithm for solving the SET problem. In Section~\ref{SUBSEC:AP}, we give two special admissible covers in practice, one of which corresponds to the DOS algorithm.

\subsection{Overview of the general algorithm}\label{SUBSEC: SET overview}
Let $\bF$ be a field of characteristic zero and let $\bF[\vx]$ be the ring of polynomials in $\vx = \{ x_1, \ldots, x_n\}$
over $\bF$. Two polynomials $p, q \in \bF[\vx]$ are said to be {\em shift equivalent} if there exist $s_1,\ldots,s_n \in \bF$ such that
\[ p(x_1 + s_1, \ldots, x_n + s_n) = q(x_1, \ldots, x_n).\]
The set $\{\vs\in \bF^n\mid p(\vx + \vs) = q(\vx)\}$ is called the \emph{dispersion set}
of $p$ and $q$ over $\bF$, denoted by $F_{p, q}$. Recall basic properties of the dispersion set in~\cite{Dvir2014complete}.

\begin{lem}(See~\cite[Observation 4.2 and Lemma 4.4]{Dvir2014complete})\label{LEM:Z_p}
	Let $p, q \in \F[\vx]$. Then
	\begin{enumerate}
		\item $F_{p,p}$ is a linear subspace of $\F^n$ over $\F$.
		\item $F_{p, q} = \vs + F_{p,p}$ for any $\vs\in F_{p, q}$ if $F_{p, q}\neq\varnothing$.
	\end{enumerate}	
\end{lem}

The problem of Shift Equivalence Testing can be stated as follows.
\begin{problem}[Shift Equivalence Testing Problem]\label{PROB:SET}
	Given $p, q \in \bF[x_1, \ldots, x_n]$, decide whether there exists $\vs = (s_1, \ldots, s_n) \in \bF^n$ such that
	\[p(\vx + \vs) = q(\vx).\]
	If such a vector $\vs$ exists, compute the dispersion set $F_{p,q}$ of $p$ and $q$ over $\bF$. In this case, by Lemma~\ref{LEM:Z_p}, it suffices to find a special solution $\vs$ in $F_{p,q}$ and a basis of $F_{p,p}$ over $\bF$.
 \end{problem}

A related problem is testing the shift equivalence over integers, i.e. deciding whether there exists a vector $\vs \in \bZ^n$
such that $p(\vx + \vs) = q(\vx)$. We denote the set $\{\vs\in \bZ^n\mid p(\vx + \vs) = q(\vx)\}$ by~$Z_{p,q}$.
The computation of $Z_{p, q}$ will play an important role in the next sections where we study
the rational summability problem and the existence problem of telescopers. By Lemma~\ref{LEM:Z_p}, we know $F_{p, q}$ is a linear variety over $\F$. Once the computation of $F_{p, q}$ boils down to
solving linear systems, we can also compute $Z_{p, q}$ by combining the same methods for the SET problem over $\bF$
and any algorithm for computing integer solutions of linear systems.

In the univariate case, the SET problem was solved by computing the resultant of two polynomials~\cite{Abramov1971}. In the multivariate case, there are three different methods for solving the SET problem in the literature. In 1996, Grigoriev  first gave a recursive algorithm (G) for the SET problem in~\cite{Grigoriev1996,Grigoriev1997}. In 2010, motivated by solving linear partial difference equations,
another algorithm (KS) for computing $Z_{p, q}$ via the Gr{\"o}bner basis method was given by Kauers and Schneider in~\cite{KauersSchneider2010}. In 2014, a new algorithm with better complexity was given by Dvir, Oliveira and Shpilka (DOS) in~\cite{Dvir2014,Dvir2014complete}.
We have implemented all of the three algorithms in Maple and the experimental comparison is tabulated in Section~\ref{sec:appendix}. The timings indicate that the DOS algorithm is the most efficient one among the three methods in practice.

In this section, we introduce $n$ new variables $\va=\{a_1,\ldots,a_n\}$. The SET problem is equivalent to finding the zeros of the polynomial $p(\vx+\va)-q(\vx)\in\bF[\va,\vx]$ with respect to $\va$. Collecting its coefficients in $\vx$, we obtain a polynomial system. A direct approach to the SET problem is solving this polynomial system. Without exploring the hidden structure of the polynomial system, this naive approach could be very in-efficient. The common idea of the above three methods is to find the defining linear system of $F_{p, q}$, which avoids solving the polynomial system directly. To do this, the DOS algorithm finds an appropriate finite cover of the polynomial system. Then it reduces the SET problem to solving several linear systems successively by evaluating the non-linear part of polynomials. This kind of evaluation is called the \textit{linearization} of polynomials, whose definition will be strictly stated below.

We first introduce some notations for later use. For any two vectors $\val=(\alpha_1,\alpha_2,\ldots,\alpha_n),\vbe=(\beta_1,\beta_2,\ldots,\beta_n)\in \bN^n$, we say $\val\geq\vbe$ if and only if $\a_i\geq\beta_i$ for all $1\leq i\leq n$. This defines a partial order on $\bN^n$. For a subset $\vy=\{y_1,y_2,\ldots,y_m\}$ of $\vx=\{x_1,x_2,\ldots,x_n\}$ with $\hat{\vy}:=\vx\setminus\vy$, let $f\left(\vx\right)=\sum_{\a}c_{\val}\left(\hat{\vy}\right)\vy^{\val}\in\F[\hat{\vy}][\vy]$. Let $H^{d}_{\vy}\left(f(\vx)\right)$ denote the homogeneous component of $f(\vx)$ of degree $d$ in $\vy$ and let $\supp_{\vy}(f)$ denote the set $\{\vy^{\val}\mid c_{\val}(\hat\vy)\neq 0\}$. For simplicity, when $\vy=\vx$, we write $H^{\ell}_{\vy}(f(\vx))$ as $H^{\ell}(f(\vx))$ and $\supp_{\vy}(f)$ as $\supp(f)$. For a subset $S\subseteq\F[\vx]$, let $\bV_{\bF}(S)$ be the zero set $\{\vs\in\F^n \mid f(\vs) = 0,\,\forall f\in S \}$.

\begin{defn}[Linearization]\label{DEF_L}
	Let $f(\vx) = H^0_{\vy}(f)(\vy) + H^1_{\vy}(f)(\vy) + \cdots + H^d_{\vy}(f)(\vy)$ be the homogeneous decomposition of $f\in \bF[\vx]=\bF[\hat{\vy}][\vy]$.  For a vector $\vs \in \bF^m$, we call the linear polynomial $H^0_{\vy}(f)(\vy) + H^1_{\vy}(f)(\vy) + \sum_{i=2}^d H^i_{\vy}(f)(\vs)$ the \emph{linearization} of $f$ at $\vs$ with respect to $\vy$, denoted by $L_{\vy = \vs}(f)$. Note that $L_{\vy = \vs}(f) = f$ if $d \leq 1$. For a polynomial set $S \subseteq \bF[\vx]$, let $L_{\vy= \vs}(S) := \{L_{\vy = \vs}(f)\mid f \in S\}$ be the \emph{linearization} of $S$.
\end{defn}	
In the following we present the main idea of algorithm DOS~\cite{Dvir2014, Dvir2014complete} and our new algorithm ADC in a general framework.

In order to compute $F_{p,q}$, we first write
	\begin{align*}%\label{EQ:p-q}
		p(\vx+\va)-q(\vx)=\sum_{\val\in\Lambda}c_{\val}(\va)\vx^{\val},
	\end{align*}
	where $c_{\val}(\va)\in \F[\va]$ and $\Lambda$ is a finite subset of $\bN^n$.
	Let
	\begin{equation}\label{EQ: S def}
		S:=\left\{c_{\val}(\va)\in\bF[\va]\mid c_{\val}(\va)\text{ is a nonzero coefficient of }\vx^{\val}\text{ in }p(\vx+\va)-q(\vx)\right\}.
	\end{equation}
	Then $F_{p,q}= \bV_{\bF}(S)$ is the zero set of $S$ in $\bF^n$. First, we classify all polynomials in $S$ according to their total degrees in $\va$ and write $S=S_0^D\cup\cdots\cup S_{d'}^D$, where $d'=\deg_\va(p(\vx + \va)-q(\vx))$ and
	\[S_i^D =\left\{c_\alpha(\va)\in S \mid \deg_\va(c_\alpha(\va))=i\right\}\]
	for $i=0,\ldots, d'$. Then $\bV_\bF(S)=\bV_\bF(\cup_{i=0}^{d'} S_i^D)$. We may assume that $S_0^D=\varnothing$, otherwise $p,q$ are not shift equivalent and return $F_{p,q} =\varnothing$. If $S_0^D\cup S_1^D$ has no solution in $\bF^n$, return $F_{p,q}=\varnothing$. Otherwise take an arbitrary solution $\vs^{(0)}\in \bV_\bF(S_0^D\cup S_1^D)$. Note that all polynomials in $S_0^D\cup S_1^D$ are linear and thus such an element $\vs^{(0)}$ can be computed straightforwardly. We shall prove that the nonlinear system $S_0^D \cup S_1^D \cup S_2^D$  has the same solutions as its linearization $S_0^D\cup S_1^D\cup L_{\va=\vs^{(0)}}(S_2^D)$ at the point $\vs^{(0)}$. If the latter linear system has no solution, return $F_{p,q}=\varnothing$. Otherwise, take an arbitrary solution $\vs^{(1)}\in \bV_\bF(S_0^D\cup S_1^D \cup L_{\va=\vs^{(0)}}(S_2^D))$ by solving the linear system. Then consider the linearization of $\cup_{i=0}^3 S_i^D$ at $\vs^{(1)}$ and we shall prove that $\bV_\bF(\cup_{i=0}^3\,S_i^D) =\bV_\bF(\cup_{i=0}^3\, L_{\va=\vs^{(1)}}(S_i^D))$. Continuing the above process, we will finally find an equivalent linear system of the polynomial system $S= \cup_{i=0}^{d'} S_i^D$ by linearization.
	
	\begin{example}\label{EX:SET}
		Let $p = x^2 + 2xy + y^2 + 2x + 6y$ and $q = x^2 + 2xy + y^2 + 4x + 8y + 11$ be two polynomials in $\bQ[x,y]$. Decide whether $p,q$ are shift equivalent with respect to $x,y$. Since
		\[p(x + a, y + b) - q(x, y) = (2a + 2b - 2)\cdot x + (2a + 2b - 2)\cdot y + (a^2 + 2ab+ b^2 + 2a + 6b - 11),\]
		we have $S = S_1^D \cup S_2^D$, where $S_1^D = \{2a + 2b - 2\}$ and $S_2^D = \{a^2 + 2ab+ b^2 + 2a + 6b - 11\}$. Take an arbitrary solution $(a,b) = (1, 0)$ of $S_1^D$. The linearization of $S_2^D$ at $(1,0)$ is
		\[L_{(a,b)=(1,0)}(S_2^D) =\{1^2 + 2\cdot 1\cdot 0 + 0^2 +2a +6b -11\} = \{ 2a+6b -10\}.\]
		In this example, the linear system $S_1^D \cup L_{(a,b)=(1,0)}(S_2^D)$ is indeed equivalent to the polynomial system $S_1^D\cup S_2^D$. So $F_{p, q} = \bV_\bF(S_1^D\cup L_{(a,b)=(1,0)} (S_2^D))= \{(-1,2)\}$.
	\end{example}
	
	Since shift operations do not change the total degree in $\vx$, the homogeneous components of both sides of $p(\vx + \va) = q(\vx)$ with respect to $\vx$ must be equal. The homogeneous decomposition of $p(\vx + \va) - q(\vx)$ yields another cover $\{S_0^H, S_1^H, \ldots, S_d^H\}$ of $S$, where $d=\max\{\deg_\vx(p(\vx)), \deg_\vx(q(\vx))\}$ and
	$$S_i^H:=\{c_{\val}(\va)\in S\mid c_{\val}(\va)\text{ is the coefficient of }\vx^{\val}\text{ in }H_{\vx}^{d-i}(p(\vx+\va))-H_{\vx}^{d-i}(q(\vx))\}$$
	for $i=0,1,\ldots, d$. In the DOS algorithm, they first introduced the above method of linearization to solve the polynomial system $S = S_0^H\cup S_1^H \cup \cdots \cup S_d^H$ and proved the correctness of their algorithm by using formal partial derivatives. In Example~\ref{EX:SET}, $S =S_1^D\cup S_2^D= S_1^H \cup S_2^H$, where $S_i^H=S_i^D$ for $i=1,2$. In general, these two covers are different. A natural question is for which cover, we can use the method of linearization to compute the dispersion set. One answer is the admissible cover defined below. In fact, the above two covers $\{S^D_0, S^D_1, \ldots, S^D_{d'}\}$ and $\{S_0^H, S_1^H, \ldots, S_d^H\}$, called by \textit{$\va$-degree cover} and \textit{$\vx$-homogeneous cover} respectively, are both admissible, which will be proved in Section~\ref{SUBSEC:AP}.

\begin{defn}[Admissible cover] \label{DEF_AP}
	Let $S\subseteq \bF[\va]$ be as in~\eqref{EQ: S def}. A collection $\{S_0, S_1, \ldots, S_m\}$ of subsets is called a \emph{cover} of $S$ if $S$ is the union of $S_0, S_1,\ldots, S_m$. Such a cover $\{S_0,S_1,\ldots,S_m\}$ is called an \emph{admissible cover} of $S$ if it satisfies the following two conditions:
	\begin{enumerate}
		\item\label{it:ac-1} The degree of any polynomial in $S_0$ is at most one.
		\item\label{it:ac-2} For every $\ell=1,2,\ldots,m$, if $c_{\val}(\va)\in S_{\ell}$, then $c_{\vbe}(\va)\in\cup_{i=0}^{\ell-1}S_i$ for all $ \vbe\in\bN^n$ with $\vbe>\val$ and $\vx^{\vbe}\in\supp_{\vx}(p(\vx+\va)-q(\vx))$.
	\end{enumerate}
\end{defn}

A general algorithm for solving the SET problem via the method of linearization is as follows. This algorithm inherits one feature of the DOS algorithm: it could be early terminated when $p,q$ are not shift equivalent. If two nonzero polynomials $p(\vx)$ and $q(\vx)$ are shift equivalent, then they have the same degree $d$ in $x$ and $H^{d}(p(\vx))=H^{d}(q(\vx))$, which means $\deg(p(\vx)-q(\vx))<\deg(p(\vx))$. Therefore, we can check the degree condition at the beginning of the algorithm for better efficiency.

\begin{algorithm}[Shift Equivalence Testing]\label{Alg:SET}{\em\bf ShiftEquivalent($p$, $q$, $[x_1,\ldots, x_n]$)}.\\
	INPUT: two multivariate polynomials $p, q\in \bF[\vx]$;\\
	OUTPUT: a special solution of $F_{p, q}$ and an $\bF$-basis of $F_{p,p}$ if $p$ and $q$ are shift equivalent; $\{\}$ otherwise.
	\step 12 if $p(\vx)=q(\vx)=0$,  {\em\bf return} ${\bf 0}$ {\em and the standard basis of} $\bF^n$.
	\step 22 if $\deg(p(\vx)-q(\vx))\geq \deg(p(\vx))$,  {\em\bf return} $\{\}$.
	\step 32 set $S:= \text{\em Coefficients}(p(\vx + \va) - q(\vx), {\vx}) \subseteq \F[\va]$.
	\step 42 let $\{S_0, S_1, \ldots, S_m\}$ be an admissible cover of $S$.
	\step 52 set $\vs^{(0)} := \bf0$.
	\step 62 {{\em\bf for}} $\ell = 0, \ldots, m$ {{\em \bf do}}
	\step 73 set $L^{(\ell)} := \cup_{i=0}^{\ell}\; L_{\va=\vs^{(\ell)}}\left(S_i\right)$.
	\step 83 solve the linear system in $\va$ defined by $L^{(\ell)}$.
	\step {9}3 if the linear system $L^{(\ell)}$ has no solution, {{\em \bf return}} $\{\}$.
	\step {10}3 else there is a special solution $\vs'\in \bF^n$ by evaluating each free variable at $0$, set $\vs^{(\ell + 1)}:=\vs'$ .
	\step {11}2 {{\em \bf return}} $\vs^{(m+1)}$ and an $\bF$-basis of the solution space of the homogeneous linear equations induced by $L^{(m)}$.
\end{algorithm}

The correctness of Algorithm~\ref{Alg:SET} is guaranteed by the following theorem.

\begin{thm}\label{THM:general_cover}
	If the cover $\{S_0, S_1, \ldots, S_m\}$ of $S$ is admissible, then
	for all $\ell=0,1,\ldots, m$,  we have either $\bV_\bF\left(\bigcup_{i=0}^{{ \ell-1}} S_i\right) = \varnothing$ or
	\[\bV_\bF\left(\bigcup_{i=0}^{\ell} S_i\right) = \bV_\bF\left(\bigcup_{i=0}^{\ell} L_{\va = \vs^{(\ell)}}(S_i)\right) \quad \text{for any $\vs^{(\ell)} \in \bV_\bF\left(\bigcup_{i=0}^{{ \ell-1}} S_i\right)$}.\]
\end{thm}

The proof of Theorem~\ref{THM:general_cover} will be given in the next subsection.

\begin{thm}\label{THM:com-SET}
    For a vector $\vd\in\bN^n$, let $p(\vx)$ and $q(\vx)$ be two multivariate polynomials in $\bF[\vx]_{\vd}$. Then Algorithm~\ref{Alg:SET} can test whether $p$ and $q$ are shift equivalent and output a special solution of $F_{p,q}$ with an $\bF$-basis of $F_{p,p}$ if $F_{p,q}\neq\varnothing$ using $\tilde{O}(d_1^{\,\omega} \cdots d_n^{\,\omega})$ ops in $\bF$. Furthermore, combining the algorithms for computing the integer solutions of linear systems, Algorithm~\ref{Alg:SET} can test whether $p$ and $q$ are shift equivalent over $\bZ$ and output a special solution of $Z_{p,q}$ with an $\bQ$-basis of $Z_{p,p}$ in the case of $Z_{p,q}\neq\varnothing$; if $n$ is fixed and $p(\vx),\,q(\vx)\in\bZ[\vx]_{\delta}$ with max-norms bounded by the positive integer $M$, then the max-norms of the output vectors are $O(M^{\delta^{O(1)}})$ and the algorithm costs $\tilde{O}(\delta^{n(\omega+1)}\log(M))$ ops.
\end{thm}
\begin{proof}
    The first three steps take $\tilde{O}(2nd_1^2\cdots d_n^2)$ ops by Fact~\ref{FACT:com-multi}. Since both $\va$-degree cover and $\vx$-homogeneous cover can be obtained by traversing elements in $S$, Step~4 can be performed in $\tilde{O}(d_1^2\cdots d_n^2)$ ops. By the definitions of $\va$-degree covers and $\vx$-homogeneous covers, we have that $m$ is of  the size $\tilde{O}(d_1+\cdots +d_n)$. Note that by Lemma~\ref{LEM:equal_sub} below, Step~7 in the loop can be replaced by setting $L^{(\ell)}$ to be the union of $L^{(\ell-1)}$ and $L_{\va=\vs^{(\ell)}}(S_\ell)$ if $\ell\geq1$, and hence the size of the linear system in Step~8 is no more than $|S_\ell|+n$. As a result, the cost of Step~7 is $\tilde{O}(|S_\ell|nd_1\cdots d_n)$ ops and that of Step~8 is $\tilde{O}((|S_\ell|+n)^\omega)$ ops in each iteration. This implies that the loop takes no more than $\tilde{O}(C)$ ops, where
    \begin{align*}
        C=\sum_{\ell=0}^m\left((|S_\ell|+n)^\omega+|S_\ell|nd_1\cdots d_n\right)
    &\leq\left(\sum_{\ell=0}^m(|S_\ell|+n)\right)^\omega+\left(\sum_{\ell=0}^m|S_\ell|\right)nd_1\cdots d_n\\
    &=(|S|+mn)^\omega+n|S|d_1\cdots d_n.
    \end{align*}
    Since $|S|$ is no more than $d_1d_2\cdots d_n$, the loop needs $\tilde{O}(d_1^{\,\omega}\cdots d_n^{\,\omega})$ ops that dominates the whole costs. This completes the proof of the first claim.

    Now we turn to the case of computing $Z_{p,q}$ by Algorithm~\ref{Alg:SET}. The cost of the algorithm follows from Fact~\ref{FACT:com-linear} since the timing is dominated by Step~8 in the loop. For the max-norms of the output vectors, note that the solution sets of $L^{(\ell)}$ would be updated at most $n+1$ times, so the unavoidable coefficient explosions would occur in at most $n+1$ iterations of the loop. Then we obtain the desired estimate of these max-norms by Fact~\ref{FACT:com-linear}.
\end{proof}

From the above complexity analysis, we can not distinguish the algorithms with $\va$-degree cover and $\vx$-homogeneous cover. In Section~\ref{sec:appendix}, we have implemented both algorithms to compare the practical
efficiency. The experiments show that our ADC algorithm is more efficient than the DOS algorithm for sparse polynomials.

\subsection{Proof of correctness of Theorem~\ref{THM:general_cover}}\label{SUBSEC: SET main proof}

Before proving Theorem~\ref{THM:general_cover}, we need several lemmas to explore the inner structure of polynomials $c_{\val}(\va)$ in $S$. First we give an explicit expression of the non-constant homogeneous components of $c_{\val}(\va)$ and find a recurrence relation among the homogeneous components. Then we explain the role of the admissible cover and the magic of linearization in Algorithm~\ref{Alg:SET}. Finally, we prove Theorem~\ref{THM:general_cover} by induction on $\ell$.

For a vector $\val=(\alpha_1,\alpha_2,\ldots,\alpha_n)\in \bN^n$, let $|\val|:=\sum_{i=1}^n\a_i$ and $\tbinom{|\val|}{\val}:=\frac{|\val|!}{\a_1!\a_2!\cdots\a_n!}$. Let $\partial_{x_i}$ denote the partial derivative with respect to $x_i$ and $\vpa^{\val}$ denote $(\pa_{x_1})^{\a_1}(\pa_{x_2})^{\a_2}\cdots(\pa_{x_n})^{\a_n}$. For $n$ variables $\va=\{a_1,a_2,\ldots,a_n\}$, we use $D_{\va}$ to denote the directional derivative in the direction of~$\va$, i.e., $D_{\va}:=\sum_{i=1}^n a_i\partial_{x_i}$. For $\vs=(s_1,\ldots,s_n)\in\bF^n$, the notation $D_\vs$ means $D_\va|_{\va=\vs}$. Then for any $k\in\bN^+$, $$D_{\va}^k:=(D_{\va})^k=\sum_{|\val|=k}\tbinom{k}{\val}\va^{\val}\vpa^{\val}$$ by the multinomial theorem since $\pa_{x_i}$ and $\pa_{x_j}$ commute.

By the directional derivative and Taylor's expansion, the homogeneous components of polynomials in $S_{\ell}^H$ can be expressed as follows.

\begin{lem}\label{LEM:DOS_hc}(See~\cite[Lemma 3.5]{Dvir2014complete})
	Let $d:=\max\{\deg_{\vx}(p(\vx)),\deg_{\vx}(q(\vx))\}$. For any $k\in\bN$ and $\ell\in\{0,1,\ldots,d\}$, we have
	\begin{align}\label{EQ:DOS_c}
		H_{\vx}^{d-\ell}(p(\vx+\va))-H_{\vx}^{d-\ell}(q(\vx))=\sum_{i=0}^{\ell}\frac{1}{i!}D^i_{\va}\left(H_{\vx}^{d-\ell+i}(p(\vx))\right)-H_{\vx}^{d-\ell}\left(q(\vx)\right)
	\end{align}
	and
	\begin{equation}\label{EQ:DOS_hc}
		H^k_{\va}\left(H_{\vx}^{d-\ell}\left(p\left(\vx+\va\right)\right)-H_{\vx}^{d-\ell}\left(q\left(\vx\right)\right)\right)=\left\{ \begin{aligned}
			&\frac{1}{k!}D^{k}_{\va}\left(H_{\vx}^{d-\ell+k}\left(p\left(\vx\right)\right)\right), &\text{~if~} k\geq 1,\\
			&H_{\vx}^{d-\ell}\left(p\left(\vx\right)\right)-H_{\vx}^{d-\ell}\left(q\left(\vx\right)\right), & \text{~if~} k=0. \\
		\end{aligned}\right.
	\end{equation}
	Moreover, for any $c_{\val}(\va)\in S$ and $k\geq1$, $H_{\va}^k(c_{\val}(\va))$ is the coefficient of $\vx^{\val}$ in $\frac{1}{k!}D^{k}_{\va}\left(H_{\vx}^{|\val|+k}\left(p\right)\right)$.
\end{lem}

\begin{proof}
	Note that $c_{\val}(\va)$ is exactly the coefficient of $\vx^{\val}$ in $H_{\vx}^{|\val|}\left(p\left(\vx+\va\right)\right)-H_{\vx}^{|\val|}\left(q\left(\vx\right)\right)$, so it is sufficient to prove Equations~\eqref{EQ:DOS_c}~and~\eqref{EQ:DOS_hc}. By Taylor's expansion, we have $$p(\vx+\va)=\sum_{i=0}^d\frac{1}{i!}D^i_{\va}(p)(\vx)=\sum_{i=0}^d\sum_{j=0}^d\frac{1}{i!}D^i_{\va}(H_{\vx}^j(p))(\vx).$$ Note that if $D^i_{\va}(H_{\vx}^j(p))$ is not equal to zero, then it is homogeneous of degree $j-i$ in $\vx$. Consequently, we obtain Equation~\eqref{EQ:DOS_c}. Moreover, note that $D^i_{\va}(H_{\vx}^{d-\ell+i}(p))(\vx)$ is homogeneous of degree $i$ with respect to $\va$, so we get Equation~\eqref{EQ:DOS_hc}, which completes the proof.
\end{proof}

Since $$\frac{1}{k!}D^{k}_{\va}\left(H_{\vx}^{|\val|+k}\left(p\left(\vx\right)\right)\right)=\frac{1}{k!}\sum_{|\vbe|=k}\tbinom{k}{\vbe}\va^{\vbe}\vpa^{\vbe}\left(H_{\vx}^{|\val|+k}(p(\vx))\right),$$
extracting the terms $H^k(c_{\val}(\va))\cdot\vx^{\val}$ in the both sides of the equality, we have
\begin{align}\label{EQ:hc}
	H^k(c_{\val}(\va))\cdot\vx^{\val}=\frac{1}{k!}\sum_{|\vbe|=k}\tbinom{k}{\vbe}\va^{\vbe}\vpa^{\vbe}\left([\vx^{\val+\vbe}](p(\vx))\cdot\vx^{\val+\vbe}\right),
\end{align}
where $[\vx^{\val+\vbe}](p(\vx))$ denotes the coefficient of $\vx^{\val+\vbe}$ in $p(\vx)$. Therefore, for any $f(\vx)\in \bF[\vx]$, we can write $D_{\va,\val}^k(f(\vx)):=\sum_{|\vbe|=k}\tbinom{k}{\vbe}\va^{\vbe}\vpa^{\vbe}\left([\vx^{\val+\vbe}](f(\vx))\cdot\vx^{\val+\vbe}\right)$ and use $D_{\vs,\val}^k$ to denote $D_{\va,\val}^k\big{|}_{\va=\vs}$ for $\vs\in\bF^n$. The following lemma is derived straightforwardly.

\begin{lem}\label{LEM:hc}
	Let $k\in\bN^+$ and $c_{\val}(\va)\in S$. Then we have
	$H^k(c_{\val}(\va))\cdot\vx^{\val}=\frac{1}{k!}D_{\va,\val}^k(p(\vx)).$
\end{lem}

For the directional derivative, we know $D^k_{\va}(f(\vx))=(D^1_{\va})^k(f(\vx))$. However $D^k_{\va,\val}(f(\vx))$ may be different from $(D^1_{\va,\val})^k(f(\vx))$, as the following example shows.

\begin{example}
	Let $\bF = \bQ$, $p(x,y),q(x,y)\in\bQ[x,y]$ with $p(x,y)=x^3  + y^3$ and $q(x,y)=p(x, y)+1$. Expanding $p(x+a,y+b)-q(x,y)$, we have $p(x+a,y+b)-q(x,y)=3a\cdot x^2+3b\cdot y^2+3a^2\cdot x +3b^2\cdot y+(a^3+b^3-1).$ Then we have $c_{(1,0)}(a,b)=3a^2$,
	\begin{align*}
		D_{(a,b),(1,0)}^1(p(\vx))=&\sum_{i+j=1}\tbinom{1}{(i,j)}a^{i}b^j\cdot\pa_{x}^i\pa_{y}^j\left([x^{1+i}y^{0+j}](p(x,y))\cdot x^{1+i}y^{0+j}\right)\\
		=&\tbinom{1}{(1,0)}a\cdot\pa_{x}\left([x^2](p(x,y))\cdot x^{2}\right)+\tbinom{1}{(0,1)}b\cdot\pa_{y}\left([xy](p(x,y))\cdot xy\right)=0,\\
		D_{(a,b),(1,0)}^2(p(\vx))=&\sum_{i+j=2}\tbinom{2}{(i,j)}a^{i}b^j\cdot\pa_{x}^i\pa_{y}^j\left([x^{1+i}y^{0+j}](p(x,y))\cdot x^{1+i}y^{0+j}\right)\\
		=&\tbinom{2}{(2,0)}a^2\cdot\pa^2_{x}\left([x^3](p(x,y))\cdot x^{3}\right)+\tbinom{2}{(1,1)}ab\cdot\pa_x\pa_{y}\left([x^2y](p(x,y))\cdot x^2y\right)\\
		&+\tbinom{2}{(0,2)}b^2\cdot\pa^2_{y}\left([xy^2](p(x,y))\cdot xy^{2}\right)\\
		=&\frac{2!}{0!2!}a^2\cdot\pa^2_{x}(x^{3})=6a^2x
	\end{align*}
	and $\left(D^1_{(a,b),(1,0)}\right)^2(p(\vx))=D^1_{(a,b),(1,0)}(0)=0$. Therefore, we can check that $H^k\left(c_{(1,0)}(a,b)\right)\cdot x$ is equal to $\frac{1}{k!}D_{(a,b),(1,0)}^k(p(x,y))$ for $k=1,2$, but $D^2_{(a,b),(1,0)}(p(\vx))$ is not equal to $\left(D^1_{(a,b),(1,0)}\right)^2(p(\vx))$.
\end{example}

Now we rewrite the expression of $D_{\va,\val}^k(f)$ and derive a recurrence relation for $D_{\va,\val}^k(f)$.

\begin{lem}\label{LEM:Drr}
	Let $\val\in\bN^n$, $k,\ell\in\bN^+$ and $f\in\bF[\vx]$. Let $\ve_i\in\bN^n$ denote a unit vector with the i-th component being one and others being zero. Then we have:
	\begin{enumerate}
		\item\label{it:exp} $D_{\va,\val}^k(f(\vx))=\sum_{j_1=1}^n\cdots\sum_{j_k=1}^n\va^{\sum_{i=1}^k\ve_{j_i}}\vpa^{\sum_{i=1}^k\ve_{j_i}}\left(\left[\vx^{\val+\sum_{i=1}^k\ve_{j_i}}\right](f(\vx))\cdot\vx^{\val+\sum_{i=1}^k\ve_{j_i}}\right).$
		\item\label{it:rec} $D_{\va,\val}^{k+\ell}(f(\vx))=\sum_{|\vbe|=\ell}\tbinom{\ell}{\vbe}\va^{\vbe}\vpa^{\vbe}\left(D_{\va,\val+\vbe}^k(f(\vx))\right).$
	\end{enumerate}
\end{lem}
\begin{proof}
	\begin{enumerate}
		\item Note that for any $\vbe\in\bN^n$ with $|\vbe|=k$, $\vbe$ can be expressed as a sum of $k$ unit vectors. Moreover, there are $\tbinom{k}{\vbe}$ different k-tuples $(j_1,j_2,\ldots,j_k)$ such that $\vbe=\sum_{i=1}^k\ve_{j_i}$. Then the claim follows from the definition of $D_{\va,\val}^k(f(\vx))$.
		\item Applying~\eqref{it:exp} twice, we have
	\begin{align*}
		&D_{\va,\val}^{k+\ell}(f(\vx))\\
		=&\begin{multlined}[t]
			\sum_{j_1=1}^n\cdots\sum_{j_{k}=1}^n\sum_{j_{k+1}=1}^n\cdots\sum_{j_{k+\ell}=1}^n\va^{\sum_{i=1}^{k}\ve_{j_i}+\sum_{i=k+1}^{k+\ell}\ve_{j_i}}\vpa^{\sum_{i=1}^{k}\ve_{j_i}+\sum_{i=k+1}^{k+\ell}\ve_{j_i}}\\
			\left(\left[\vx^{\val+\sum_{i=1}^{k}\ve_{j_i}+\sum_{i=k+1}^{k+\ell}\ve_{j_i}}\right](f(\vx))\cdot\vx^{\val+\sum_{i=1}^{k}\ve_{j_i}+\sum_{i=k+1}^{k+\ell}\ve_{j_i}}\right)
		\end{multlined}\\
		=&\begin{multlined}[t]       \sum_{j_{k+1}=1}^n\cdots\sum_{j_{k+\ell}=1}^n\va^{\sum_{i=k+1}^{k+\ell}\ve_{j_i}}\vpa^{\sum_{i=k+1}^{k+\ell}\ve_{j_i}}\\
			\left(\sum_{j_1=1}^n\cdots\sum_{j_{k}=1}^n\va^{\sum_{i=1}^{k}\ve_{j_i}}\vpa^{\sum_{i=1}^{k}\ve_{j_i}}\left(\left[\vx^{\val+\sum_{i=1}^{k+\ell}\ve_{j_i}}\right](f(\vx))\cdot\vx^{\val+\sum_{i=1}^{k+\ell}\ve_{j_i}}\right)\right)
		\end{multlined}\\
		=&\sum_{j_{k+1}=1}^n\cdots\sum_{j_{k+\ell}=1}^n\va^{\sum_{i=k+1}^{k+\ell}\ve_{j_i}}\vpa^{\sum_{i=k+1}^{k+\ell}\ve_{j_i}}\left(D_{\va,\val+\sum_{i=k+1}^{k+\ell}\ve_{j_i}}^k(f(\vx))\right).
	\end{align*}
	Then as the proof of~\eqref{it:exp}, we can finally obtain~\eqref{it:rec} by setting $\vbe=\sum_{i=k+1}^{k+\ell}\ve_{j_i}$.
\end{enumerate}
\end{proof}

\begin{example}
	Let $\bF=\bQ$ and $p=x^4+x^2y + y^3\in \bQ[x,y]$. After expanding, we get $p(x+a,y+b)=x^4 + 4a\cdot x^3 + x^2y+y^3 + (6a^2+b)\cdot x^2 + 2a\cdot xy + 3b\cdot y^2+(4a^3+2ab)\cdot x+ (a^2+3b^2)\cdot y +(a^4+a^2b+b^3)$. All terms of $p(x+a,y+b)$ are listed in Figure~\ref{fig:tree}.
	
	\begin{figure}[htpb]
		\begin{spacing}{1.0}
				\begin{center}
					\begin{tikzpicture}[scale=0.65,every node/.style = {
							align=center}
						]
						\node[draw=white] at (-12.5,7.2) {$H_{(x,y)}^4:$};
						\node[draw=white] at (-12.5,5.4) {$H_{(x,y)}^3:$};
						\node[draw=white] at (-12.5,3.6) {$H_{(x,y)}^2:$};
						\node[draw=white] at (-12.5,1.8) {$H_{(x,y)}^1:$};
						\node[draw=white] at (-12.5,0) {$H_{(x,y)}^0:$};
						\node[draw=white] at (10.0,7.2) {\quad};
						
						\node [draw=cyan,fill=cyan!40,rounded corners, rectangle](1) at (0,0){$(a^4+a^2b+b^3)\cdot1$};
						\node [draw=cyan,fill=cyan!40,rounded corners, rectangle](x) at (-2.5,1.8){$(4a^3+2ab)\cdot x$};
						\node [draw=cyan,fill=cyan!40,rounded corners, rectangle](x2) at (-5.0,3.6){$(6a^2+b)\cdot x^2$};
						\node [draw=cyan,fill=cyan!40,rounded corners, rectangle](x3) at (-7.5,5.4){$4a\cdot x^3$};
						\node [draw=cyan,fill=cyan!40,rounded corners, rectangle](x4) at (-10.0,7.2){$1\cdot x^4$};
						
						\node [draw=green,fill=green!40,rounded corners, rectangle] (y) at (2.5,1.8){$(a^2+3b^2)\cdot y$};
						\node [draw=green,fill=green!40,rounded corners, rectangle] (xy) at (0,3.6){$2a\cdot xy$};		
						\node [draw=green,fill=green!40,rounded corners, rectangle] (y2) at (5,3.6){$3b\cdot y^2$};		
						\node [draw=green,fill=green!40,rounded corners, rectangle] (y3) at (7.5,5.4){$1\cdot y^3$};		
						\node [draw=green,fill=green!40,rounded corners, rectangle] (xy2) at (2.5,5.4){$0\cdot xy^2$};		
						\node [draw=green,fill=green!40,rounded corners, rectangle] (x2y) at (-2.5,5.4){$1\cdot x^2y$};		
						
						%\draw [draw=black!50!green,rounded corners=14mm] (-5.3,6.3)--(10.0,6.3)--(2.5,-0.2)--cycle;
						
						\draw [draw=blue,-](x)->(1);
						\draw [draw=blue,-](x2)->(x);
						\draw [draw=blue,-](x3)->(x2);
						\draw [draw=blue,-](x4)->(x3);
						
						\draw [draw=black!50!green,-](y)->(1);
						\draw [draw=black!50!green,-](xy)->(x);
						\draw [draw=black!50!green,-](y2)->(y) node [right,xshift=10mm,yshift=5mm]{$b\partial_y$};
						\draw [draw=black!50!green,-](x2y)->(x2);
						\draw [draw=black!50!green,-](y3)->(y2);
						\draw [draw=black!50!green,-](x2y)->(xy);
						\draw [draw=black!50!green,-](xy2)->(y2);
						\draw [draw=black!50!green,-](xy2)->(xy);
						\draw [draw=black!50!green,-](xy)->(y) node [left, xshift=-9mm,yshift=6mm]{$a\partial_x$};;
						
					\end{tikzpicture}\vspace{-0.3cm}
				\end{center}
		\end{spacing}
		\caption{Terms of the polynomial $p(x+a,y+b)$}
		\label{fig:tree}
	\end{figure}
    Taking $q(x,y)=0$ in Lemma~\ref{LEM:hc}, we get $D_{(a,b),(i,j)}^k(p(x,y))=k!H^k_{(a,b)}([x^iy^j](p(x+a,y+b)))\cdot x^iy^j$ for all $k\geq 1$. So we can read off $D_{(a,b),(i,j)}^k(p(x,y))$ from Figure~\ref{fig:tree}. For instance,
   $$D_{(a,b),(0,1)}^2 (p(x,y))= 2! \cdot (a^2+3b^2)\cdot y, \, D_{(a,b),(1,1)}^1(p(x,y)) = 2a\cdot xy \text{ and } D_{(a,b),(0,2)}^1(p(x,y))= 3b\cdot y^2.$$
   Taking $k=\ell=1$ in Lemma~\ref{LEM:Drr}~\eqref{it:rec}, we obtain a recurrence relation among these three terms:
   \begin{align*}
   	D_{(a,b),(0,1)}^2(p(x,y)) &= \sum_{i+j=1}\tbinom{1}{(i,j)}a^ib^j\partial_x^i\partial_y^j\left(D_{(a,b),(i,1+j)}^1(p(x,y))\right)\\
   	&=\tbinom{1}{(1,0)}a\partial_x\left(D_{(a,b),(1,1)}^1(p(x,y))\right) + \tbinom{1}{(0,1)}b\partial_y\left(D_{(a,b),(0,2)}^1(p(x,y))\right).
   	\end{align*}
  This implies
   \begin{equation}\label{EQ:rec1}
    	2(a^2 + 3b^2)y=a\partial_x\left(2axy\right) + b\partial_y\left(3by^2\right).
   \end{equation}
   %Similarly, taking $k=0$ and $\ell=2$ in Lemma~\ref{LEM:hc},
   By the definition of $D_{\va,\val}^k(p)$ and Lemma~\ref{LEM:hc}, we get
   \begin{equation}\label{EQ:rec2}
   	2(a^2+3b^2)y=a^2\partial_x^2(x^2y) + 2ab\partial_x\partial_y(0\cdot xy^2) + b^2\partial_y^2(y^3).
   	\end{equation}
   Note that the term $x^4$ does not involve in the above two equations \eqref{EQ:rec1} and~\eqref{EQ:rec2} because $y\nmid x^4$. In this example, the term $x^4$ only affects all terms in the blue branch, such as $3!\cdot 4a^3x=a^3\partial_x^3(x^4)$.

	Without introducing the notation $D_{\va,\val}^k$, by Lemma~\ref{LEM:DOS_hc} (or Lemma 3.5 in~\cite{Dvir2014complete}) we only get ``global'' relations, such as
	\[2!\cdot H^2_{(a,b)}\left(H^1_{(x,y)}(p(x+a,y+b))\right)=D_{(a,b)}^2\left(H_{(x,y)}^{1+2}(p(x,y))\right).\]
	This implies two relations among the rows (instead of the points) in the figure:
	\begin{align*}
		2(2abx + a^2y+3b^2y) &= (a\partial_x + b\partial_y)^2 (x^2y+0\cdot xy^2+y^3)\\
		&=(a\partial_x + b\partial_y)(bx^2 + 2axy+3by^2).
	\end{align*}
\end{example}

From Observation~3.4 in~\cite{Dvir2014complete}, we know if $D_\va^1(f(\vx)) = D_\vb^1(f(\vx))$, then $D_\va^k(f(\vx))= D_\vb^k(f(\vx))$ for all $k\geq 1$. Now we show that for any $\vs\in\bF^n$, $D_{\vs,\val}^k(f(\vx))$ can be determined by $D_{\vs,\vbe}^1(f(\vx))$ for all $\vbe \in \bN^n$ with $\vbe \geq \val$ and $|\vbe|=|\val|+k-1$. This is why we introduce the second condition in the definition of an admissible cover.

\begin{lem}\label{LEM:Da,al}
	Let $\vr,\vs\in\bF^n$, $\val\in\bN^n$, $k\in\bN^+$ and $f(\vx)\in\bF[\vx]$. If $D^1_{\vr,\vbe}(f(\vx))=D^1_{\vs,\vbe}(f(\vx))$ for all $\vbe\in\bN^n$ with $\vbe\geq \val$ and $|\vbe|=|\val|+k-1$, then we have $D^k_{\vr,\val}(f(\vx))=D_{\vs,\val}^k(f(\vx))$.
\end{lem}
\begin{proof}
	The proof is by induction on $k$. It is clear to see that the lemma is true for $k=1$. Now assume the equality holds for $k$. For $k + 1$, assume that $D^1_{\vr,\vbe}(f(\vx))=D^1_{\vs,\vbe}(f(\vx))$ for all $\vbe\in\bN^n$ with $\vbe\geq \val$ and $|\vbe|=|\val|+(k+1)-1$. We have $D_{\vr,\val}^{k+1}(f(\vx))=\sum_{i=1}^n\vr^{\ve_i}\vpa^{\ve_i}(D_{\vr,\val+\ve_i}^k(f(\vx)))$ by Lemma~\ref{LEM:Drr}~\eqref{it:rec}. Note that for all $\vga\in\bN^n$ with $\vga\geq \val+\ve_i$ and $|\vga|=|\val+\ve_i|+k-1$, we have $\vga\geq\val$ and $|\vga|=|\val|+(k+1)-1$. Thus by assumption we have $D^1_{\vr,\vga}(f(\vx))=D^1_{\vs,\vga}(f(\vx))$. It follows from the inductive hypothesis that $D_{\vr,\val+\ve_i}^k(f(\vx))=D_{\vs,\val+\ve_i}^k(f(\vx))$. So
	\begin{align*}	        D_{\vr,\val}^{k+1}(f(\vx))=&\sum_{i=1}^n\vr^{\ve_i}\vpa^{\ve_i}\left(D_{\vs,\val+\ve_i}^k(f(\vx))\right)\\
		=&\sum_{i=1}^n\vr^{\ve_i}\vpa^{\ve_i}\left(\sum_{|\vga|=k}\tbinom{k}{\vga}\vs^{\vga}\vpa^{\vga}\left([\vx^{\val+\ve_i+\vga}](f(\vx))\cdot\vx^{\val+\ve_i+\vga}\right)\right)\\
		=&\sum_{|\vga|=k}\tbinom{k}{\vga}\vs^{\vga}\vpa^{\vga}\left(\sum_{i=1}^n\vr^{\ve_i}\vpa^{\ve_i}([\vx^{\val+\vga+\ve_i}](f(\vx))\cdot\vx^{\val+\vga+\ve_i})\right)\\
		=&\sum_{|\vga|=k}\tbinom{k}{\vga}\vs^{\vga}\vpa^{\vga}D^1_{\vr,\val+\vga}(f(\vx)).
	\end{align*}
	Because $\val+\vga\geq \val$ and $|\val+\vga|=|\val|+|\vga|=|\val|+(k+1)-1$, we have $D^1_{\vr,\val+\vga}(f(\vx))=D^1_{\vs,\val+\vga}(f(\vx))$ by assumption. Applying Lemma~\ref{LEM:Drr}~\eqref{it:rec} again, the proof is completed.
\end{proof}

Combining Lemma~\ref{LEM:hc} and the above lemma, we get the following lemma immediately.

\begin{lem}\label{LEM:sub_special_solution}
	Let $c_{\val}(\va)\in S$, $k\in\bN$ with $k\geq 2$ and $\vr,\vs\in\F^n$. If \[H^1(c_{\vbe})(\vr)=H^1(c_{\vbe})(\vs)\text{\quad for all $\vbe\in\bN^n$ with $\vbe\geq\val$ and $|\vbe|=|\val|+k-1$,}\] then we have \[H^k(c_{\val})(\vr)=H^k(c_{\val})(\vs).\]
\end{lem}

Now we are ready to show that for an admissible cover, the linearization does not change the zero set of the polynomial system $S$.

\begin{lem}\label{LEM:equal_sub}
	Let $\{S_0, S_1 ,\ldots , S_m\}$ be an admissible cover of $S$ and $\ell\in\{0,1,\ldots, m\}$. If there exist $\vr,\vs\in \bV_{\bF}(\cup_{i=0}^{\ell-1}S_i)$, then for all $c_{\val}(\va)\in\cup_{i=0}^{\ell}S_i$ we have
	\begin{align}\label{EQ:equal_sub}
		L_{\va=\vr}(c_{\val})(\va)=L_{\va=\vs}(c_{\val})(\va).
	\end{align}
	Furthermore, we have $c_{\val}(\vr)=L_{\va=\vs}(c_{\val})(\vr)$.
\end{lem}
\begin{proof}
	Since $c_{\val}(\vr)=L_{\va=\vr}(c_{\val})(\vr)$, it is sufficient to prove Equation~\eqref{EQ:equal_sub} by induction on $\ell$.
	
	For $\ell=0$, we have $\deg(c_{\val}(\va))\leq1$ by Definition~\ref{DEF_AP}, so $L_{\va=\vr}(c_{\val})(\va)=c_{\val}(\va)=L_{\va=\vs}(c_{\val})(\va)$.
	
	For $\ell>0$, suppose the lemma holds for smaller $\ell$. Then it is sufficient to show that $H^k(c_{\val})(\vr)=H^k(c_{\val})(\vs)$ for all $k\in\bN$ with $k\geq 2$. By Lemma~\ref{LEM:sub_special_solution}, we know the proof is completed by showing that $H^1(c_{\vbe})(\vr)=H^1(c_{\vbe})(\vs)$ for $\vbe\in\bN^n$ with $\vbe\geq\val$ and $|\vbe|=|\val|+k-1$. Because $|\vbe|\geq|\val|+2-1>|\val|$, we have $\vbe>\val$. Then we get either $c_{\vbe}(\va)\in \cup_{i=0}^{\ell-1}S_i$ or $c_{\vbe}(\va)=0$ by Definition~\ref{DEF_AP}. For $\ell-k+1\leq\ell-2+1\leq\ell-1$, we have $\vr,\vs\in \bV_{\bF}(\cup_{i=0}^{\ell-1}S_i)\subseteq\bV_{\bF}(\cup_{i=0}^{\ell-k+1}S_i)$. This means $\vr$ and $\vs$ are zeros of all polynomials $c_{\vbe}(\va)\in S_{\ell-k+1}$. Then we have
	\begin{align}\label{EQ:linear-1}
		L_{\va=\vr}(c_{\vbe})(\vr)=c_{\vbe}(\vr)=0 \text{\quad and \quad}
		L_{\va=\vs}(c_{\vbe})(\vs)=0.
	\end{align}
	On the other hand, $\vr,\vs\in\bV_{\bF}(\cup_{i=0}^{(\ell-k+1)-1}S_i)$ because $\bV_{\bF}(\cup_{i=0}^{\ell-k+1}S_i)\subseteq\bV_{\bF}(\cup_{i=0}^{(\ell-k+1)-1}S_i)$. By the inductive hypothesis with $\ell-k+1$, we get $L_{\va=\vr}(c_{\vbe})(\va)=L_{\va=\vs}(c_{\vbe})(\va)$. So
	\begin{align}\label{EQ:linear-2}
		H^0(L_{\va=\vr}(c_{\vbe}))=H^0(L_{\va=\vs}(c_{\vbe})).
	\end{align}
	Note that $H^1(c_{\vbe})(\va)=H^1(L_{\va=\vr}(c_{\vbe}))(\va)=H^1(L_{\va=\vs}(c_{\vbe}))(\va)$. Combining the equations~\eqref{EQ:linear-1} and~\eqref{EQ:linear-2}, we have
	\begin{align*}
		H^1(c_{\vbe})(\vr)&=H^1(L_{\va=\vr}(c_{\vbe}))(\vr)=-H^0(L_{\va=\vr}(c_{\vbe}))\\
		&=-H^0(L_{\va=\vs}(c_{\vbe}))=H^1(L_{\va=\vs}(c_{\vbe}))(\vs)=H^1(c_{\vbe})(\vs),
	\end{align*}
	which completes the proof.
\end{proof}	

\begin{proof}[Proof of Theorem~\ref{THM:general_cover}]
	We shall prove the theorem by induction on $\ell$.
	
	For $\ell=0$, we know that any $c_{\val}(\va)$ in $S_0$ satisfies $\deg(c_{\val}(\va))\leq1$ by Definition~\ref{DEF_AP}. Thus we have $L_{\va=\vs^{(0)}}(S_0)=S_0$ and $\bV_{\F}(S_0)=\bV_{\F}(L_{\va=\vs^{(0)}}(S_0))$.
	
	For $\ell>0$, assume the theorem holds for $\ell-1$, i.e., $\bV_\F(\cup_{i=0}^{ \ell-1} S_i)=\bV_{\F}(\cup_{i=0}^{\ell-1}L_{\va=\vs^{(\ell-1)}}(S_i))$. Taking $\vr,\vs=\vs^{(\ell-1)},\vs^{(\ell)}\in \bV_\F(\cup_{i=0}^{ \ell-2} S_i)$ in Lemma~\ref{LEM:equal_sub}, we know that the linearizations of $c_{\vbe}(\va)$ at $\vs^{(\ell)}$ and $\vs^{(\ell-1)}$ are equal for all $c_{\vbe}(\va)\in\cup_{i=0}^{\ell-1}S_i$. This means $\cup_{i=0}^{\ell-1}L_{\va=\vs^{(\ell)}}(S_i)=\cup_{i=0}^{\ell-1}L_{\va=\vs^{(\ell-1)}}(S_i)$. Then we have
	\[
	\bV_{\F}\left(\cup_{i=0}^{\ell}L_{\va=\vs^{(\ell)}}(S_i)\right)
	\subseteq\bV_{\F}\left(\cup_{i=0}^{\ell-1}L_{\va=\vs^{(\ell)}}(S_i)\right)
	=\bV_{\F}\left(\cup_{i=0}^{\ell-1}L_{\va=\vs^{(\ell-1)}}(S_i)\right)
	=\bV_{\F}\left(\cup_{i=0}^{\ell-1}S_i\right),
	\]
	where the last equality follows from the inductive hypothesis. Note that $\bV_\F(\cup_{i=0}^{ \ell} S_i)$ is also a subset of $\bV_\F(\cup_{i=0}^{ \ell-1} S_i)$. So we only need to prove that for all $\vr\in \bV_\F(\cup_{i=0}^{ \ell-1} S_i)$,
	\begin{align}\label{EQ:leq}
		\vr\in \bV_{\F}\left(\cup_{i=0}^{\ell}S_i\right)\text{\quad if and only if \quad}\vr\in \bV_{\F}\left(\cup_{i=0}^{\ell}L_{\va=\vs^{(\ell)}}(S_i)\right).
	\end{align}
	Because $s^{(\ell)}\in\bV_\F(\cup_{i=0}^{ \ell-1} S_i)$, we have $c_{\val}(\vr)=L_{\va=\vs^{(\ell)}}(c_{\val})(\vr)$ for all $c_{\val}(\va)\in\cup_{i=0}^{ \ell} S_i$ by Lemma~\ref{LEM:equal_sub}. Then the claim~\eqref{EQ:leq} follows immediately.
\end{proof}

The main distinction between the reasoning of our general scheme and the DOS algorithm can be shown in Lemma~\ref{LEM:sub_special_solution}. In~\cite{Dvir2014complete}, Lemma~3.5 proves Equation~\eqref{EQ:DOS_c} which we extend to Lemmas~\ref{LEM:DOS_hc} and~\ref{LEM:hc}. Observation~3.4 in~\cite{Dvir2014complete} is generalized by Lemma~\ref{LEM:Da,al}. The subsequent proof for the correctness of the DOS algorithm can be summarized by the lemma below. Then the remaining steps for proving Theorem~\ref{THM:general_cover} and the correctness of the DOS algorithm are similar.

\begin{lem}\label{LEM:DOS_sub_special_solution}
	Let $d:=\max\{\deg_{\vx}(p(\vx)),\deg_{\vx}(q(\vx))\}$, $i\in\{0,1,\ldots,d\}$, $k\in\bN$ with $k\geq2$ and $\vr,\vs\in\bF^n$. If
	\begin{align*}       H^1(c_{\vbe})(\vr)=H^1(c_{\vbe})(\vs)\text{\quad for all $c_{\vbe}(\va)\in S_{i-k+1}^H$,}
	\end{align*}
	then we have
	\[H^k(c_{\val})(\vr)=H^k(c_{\val})(\vs)\text{\quad for all $c_{\val}(\va)\in S_{i}^H$.}\]
\end{lem}

Given two points $\vr$ and $\vs$ in $\bF^n$, we use diagrams in Figures~\ref{FIG:DOS_3D} and~\ref{FIG:General} to explain how Lemma~\ref{LEM:sub_special_solution} makes the statement more precise than this lemma for the case where $d=\deg_{\vx}(p)=\deg_{\vx}(q)=4$ and $n=2$. Let $p(x+a,y+b)-q(x,y)=\sum_{(\alpha,\beta)\in\Lambda}c_{(\alpha,\beta)}(a,b)x^{\alpha}y^{\beta}$. For the $k$-th homogeneous component of $c_{(\alpha,\beta)}$, if its values at $\vr$ and $\vs$ are equal, we draw a point at position $(\alpha,\beta,k)$ in the space. Furthermore, if $H^k(c_{(\alpha,\beta)})(\vr)=H^k(c_{(\alpha,\beta)})(\vs)$ for all $(\alpha,\beta)\in\bN^2$ such that the sum of $\alpha$ and $\beta$ is a fixed constant $i$, which means there are points $(\alpha,\beta,k)$ on the same line, then we draw a segment to connect them to each other. Note that the degree of any polynomial in $S_{d-i}^H$ is no more than $i$, which will be proved exactly in Lemma~\ref{LEM:separate_connection}. Lemma~\ref{LEM:DOS_sub_special_solution} implies that the dark green segment on one triangle face can conclude all the segments on this face with $k\geq2$ in Figure~\ref{FIG:DOS_3D}. More precisely, Lemma~\ref{LEM:sub_special_solution} tells us that in Figure~\ref{FIG:General}, on one triangle face, every point with $k\geq2$ can be deduced from the part of dark green segment which is cut out by two dotted line from this point. For example, Point $A$ can be inferred from Segment $\ell$.

\begin{figure}[htbp]
	\centering
	
	\begin{minipage}{6cm}
		\centering
		\includegraphics[width=1\textwidth]{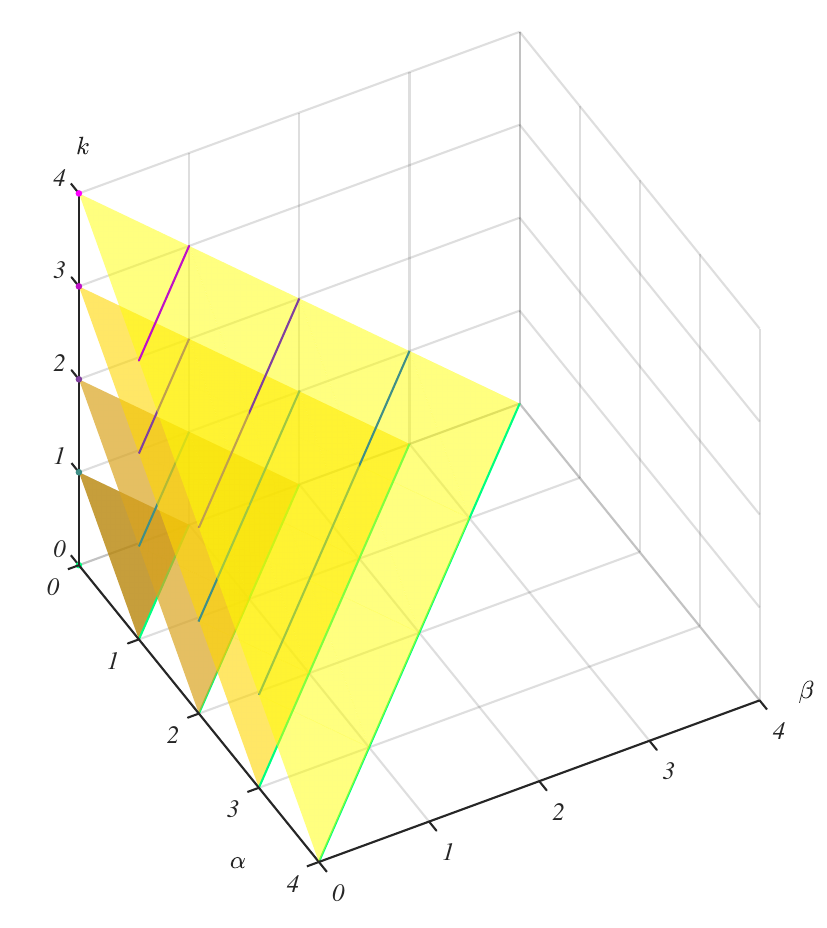}
		\caption{Graph for Lemma~\ref{LEM:DOS_sub_special_solution}} \label{FIG:DOS_3D}
	\end{minipage}
\hfill
	\begin{minipage}{6cm}
		\centering
		\includegraphics[width=1\textwidth]{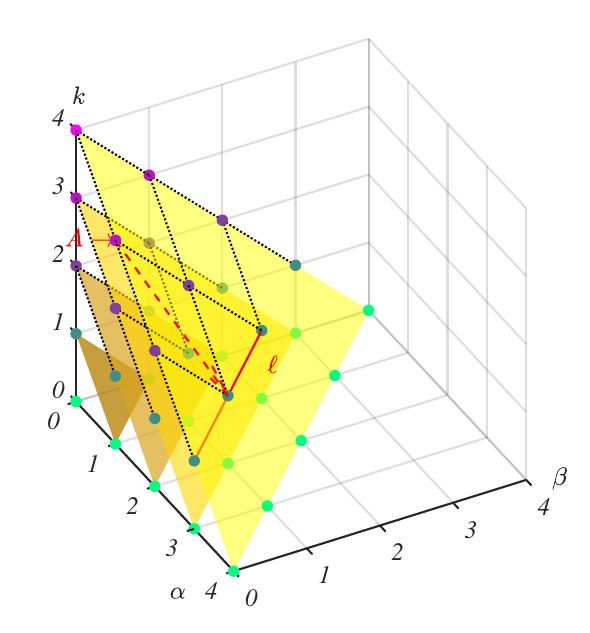}
		\caption{Graph for Lemma~\ref{LEM:sub_special_solution}}\label{FIG:General}
	\end{minipage}
\end{figure}

\subsection{Two special admissible covers}\label{SUBSEC:AP}

By Theorem~\ref{THM:general_cover}, we see that any admissible cover of the polynomial system $S$ corresponds to an algorithm for solving the SET problem via linear system solving. We now present two special admissible covers.
The first admissible cover defined below classifies the polynomials in the set $S$ according to their degrees in $\va$.

\begin{thm}[$\va$-degree cover]\label{THM:separate_tdeg}
	Let $d'$ be the degree of $p(\vx+\va)-q(\vx)$ with respect to $\va$. Let $S^D_i:=\{c_{\val}(\va)\in S\mid\deg(c_{\val}(\va))=i\}$. Then the cover $\{S_0^D,S_1^D,\ldots,S_{d'}^D\}$ of $S$ is admissible.
\end{thm}
\begin{proof}
	We will check that this cover satisfies the two conditions mentioned in Definition~\ref{DEF_AP}. The condition~\eqref{it:ac-1} can be checked directly by definition. As for~\eqref{it:ac-2}, assume that $c_{\val}(\va)\in S^D_{\ell}$ and $\vbe$ is an arbitrary vector in $\bN^n$ with $\vbe>\val$ and $\vx^{\vbe}\in\supp_{\vx}(p(\vx+\va)-q(\vx))$. We will argue by contradiction that $\deg(c_{\val}(\va))>\deg(c_{\vbe}(\va))$. By assumption, we know that $\deg(c_{\val}(\va))=\ell$. If there is a monomial $\va^{\vga}\in\supp(c_{\vbe}(\va))$ with $|\vga|\geq\ell$, then by Equation~\eqref{EQ:hc}, we have $\vx^{\vga+\vbe}\in\supp(p(\vx))$. Since $|\vga+\vbe-\val|\geq|\vbe|-|\val|>0$, we obtain $\va^{\vga+\vbe-\val}\in\supp(H^{|\vga+\vbe-\val|}(c_{\val}(\va)))\in\supp(c_{\val}(\va))$ by Equation~\eqref{EQ:hc}. However, note that $\vbe>\val$ implies $|\vbe|>|\val|$, so $|\vga+\vbe-\val|=|\vga|+|\vbe|-|\val|>\ell$, which leads to a contradiction to the fact that $\ell$ is the degree of $c_{\val}(\va)$.
\end{proof}

Note that $\vbe>\val$ implies $|\vbe|>|\val|$. This inspires the second admissible cover called the \emph{$\vx$-homogeneous cover}.
Before we prove it, we first present a useful lemma.

\begin{lem}\label{LEM:separate_connection}
	Let $d:=\deg_{\vx}(p(\vx))$. For any $\val\in\bN^n$ with $c_{\val}(\va)\in S$, we have $\deg_{\va}(c_{\val}(\va))\leq d-|\val|$.
\end{lem}

\begin{proof}
	Note that $[\vx^{\val+\vbe}](p(\vx))\neq0$ yields $|\val+\vbe|\leq d$, so $k=|\vbe|\leq d-|\val|$ in Equation~\eqref{EQ:hc}. That is to say, $H^k(c_{\val}(\va))\cdot\vx^{\val}=0$ if $k>d-|\val|$. Then the conclusion follows.
\end{proof}
	
\begin{thm}[$\vx$-homogeneous cover]\label{THM:DOS}
	Let $d$ be the maximal degree of $p(\vx)$ and $q(\vx)$. Let $S^H_i:=\{c_{\val}(\va)\in S\mid|\val|=d-i\}$ for $i=0,1,\ldots,d$. Then the cover $\{S_0^H,S_1^H,\ldots,S_d^H\}$ of $S$ is admissible.
\end{thm}
\begin{proof}
	We first show that $\{S_0^H,S_1^H,\ldots,S_d^H\}$ is exactly a cover of $S$. It is sufficient to show that $S$ is a subset of $\cup_{i=0}^dS_i^H$. This is true because we have $\vx^{\val}\in\supp_{\vx}(p(\vx+\va)-q(\vx))\subseteq\supp(p(\vx))\cup\supp(q(\vx))$ for arbitrary $c_{\val}(\va)\in S$ and then $0\leq |\val|\leq d$.
	
	Then we check this cover is admissible. If there exists $c_{\val}(\va)\in S^H_0$ with a nonlinear monomial $\va^{\vbe}$, then we have $\vx^{\val+\vbe}\in\supp(p(\vx))$ by Equation~\eqref{EQ:hc}. For $|\vbe|\geq 2$, there is a unit vector $\ve_j\in\bN^n$ such that $\ve_j<\val+\vbe$. Thus $\va^{\ve_j}\vx^{\val+\vbe-\ve_j}\in\supp_{\vx\cup\va}(D_{\va,\val+\vbe-\ve_j}^{1}(p(\vx)))$. By Lemma~\ref{LEM:hc}, $\vx^{\val+\vbe-\ve_j}\in\supp_{\vx}(p(\vx+\va)-q(\vx))$. However, $|\val+\vbe-\ve_j|=|\val|+|\vbe|-|\ve_ j|\geq d+2-1>d$, which leads to a contradiction to the fact that $\deg_{\vx}(p(\vx+\va)-q(\vx))\leq\max\{\deg_{\vx}(p(\vx))),\deg_{\vx}(q(\vx))\}=d$. So the condition~\eqref{it:ac-1} in Definition~\ref{DEF_AP} is satisfied. Finally, for any $\ell=1,2,\ldots,d$, let $c_{\val}(\va)\in S^H_{\ell}$. Then for all $ \vbe\in\bN^n$ with $\vbe>\val$ and $\vx^{\vbe}\in\supp_{\vx}(p(\vx+\va)-q(\vx))$, we have $|\vbe|>|\val|=d-\ell$, so $c_{\vbe}(\va)\in S_{d-|\vbe|}^H\subseteq\cup_{i=0}^{\ell-1}S^H_i$. Therefore, the cover $\{S^H_0,S^H_1,\ldots,S^H_d\}$ also satisfies the condition~\eqref{it:ac-2} in Definition~\ref{DEF_AP} and the proof is complete.
\end{proof}

The cover $\{S^H_0,S^H_1,\ldots,S^H_d\}$ is exactly the one defined in the DOS algorithm and we call it $\vx$-homogeneous cover. As a consequence, we reproved the correctness of the DOS algorithm in our general framework of admissible covers.

After introducing two special admissible covers, we would like to compare them and explain their connection. For simplicity, we always assume $H^{d}_{\vx}(p(\vx + \va)) = H^{d}_{\vx}(q(\vx))$ with $\deg (p(\vx)) = \deg(q(\vx))=d$ in the following discussion. For this special case, $S_0^H=\varnothing$. We get $\deg_{\va}(p(\vx + \va)-q(\vx))=d$ by Equation~\eqref{EQ:hc}. Lemma~\ref{LEM:separate_connection} yields $S_\ell^H\subseteq\cup_{i=0}^{\ell}S_i^D$. Hence the connection between two different covers is like the following figure.

\begin{center}
	\renewcommand\arraystretch{2.5}
	\begin{tabular}{|c:c:c:c|l}
		\multicolumn{1}{c}{$\,\enspace\enspace\enspace\enspace$}     & \multicolumn{1}{c}{$\,\enspace\enspace\enspace\enspace$}     &
		\multicolumn{1}{c}{$\,\enspace\enspace\enspace\enspace$}     &
		\multicolumn{1}{c}{$\,\enspace\enspace\enspace\enspace$} &
		\\
		\cline{4-4}
		\multicolumn{1}{c}{$S$}     & \multicolumn{2}{c}{}     & \multicolumn{1}{|c|}{$\,$} &
		$S_d^D$\\
		\cdashline{4-4}[0.5pt/0.5pt]
		\multicolumn{2}{c}{}     &
		\multicolumn{1}{c}{$\cdots$}     &
		\multicolumn{1}{:c|}{$\vdots$} &
		$\vdots$\\
		\cline{2-2}
		\cdashline{3-4}[0.5pt/0.5pt]
		\multicolumn{1}{c}{}     & \multicolumn{1}{|c:}{$\,$}     &
		\multicolumn{1}{c}{}   &
		\multicolumn{1}{:c|}{$\,$} &
		$S_2^D$\\
		\cline{1-1}
		\cdashline{2-4}[0.5pt/0.5pt]
		\multicolumn{1}{|c}{$\,$}     & \multicolumn{1}{:c:}{$\,$}     &
		\multicolumn{1}{c}{}   &
		\multicolumn{1}{:c|}{$\,$} &
		$S_1^D$\\
		\cdashline{1-4}[0.5pt/0.5pt]
		\multicolumn{1}{|c}{$\,$}     & \multicolumn{1}{:c:}{$\,$}     &
		\multicolumn{1}{c}{}   &
		\multicolumn{1}{:c|}{$\,$} &
		$S_0^D$\\
		\cline{1-4}
		\multicolumn{1}{c}{$S_1^H$}     & \multicolumn{1}{c}{$S_2^H$}     &
		\multicolumn{1}{c}{$\cdots$}     &
		\multicolumn{1}{c}{$S_{d}^H$} &
		\\
	\end{tabular}
\end{center}

Now we give some examples to illustrate the two algorithms induced by the two elaborated admissible covers.
\begin{example}\label{EX:SET_compare}
	Let $\bF = \bQ$, $p_i(x,y,z),q_i(x,y,z)\in\bQ[x,y,z]$ for $i=1,2$ with $p_1(x,y,z)=x^4 + x^2y + y^2$, $q_1(x,y,z)=p_1(x, y + 1, z + 2)+ z$, $p_2(x,y,z)=x^4 + x^3y + xy^2 + z^2$ and $q_2(x,y,z)=p_2(x,y+1,z+2)+xy$.
	\begin{enumerate}
		\item\label{it:p1q1} Compute $F_{p_1,q_1}$. We expand $p_1(x + a, y + b,z+c) - q_1(x, y,z)$ and get that
		\begin{align*}
			&p_1(x + a, y + b,z+c) - q_1(x, y,z)\\
			=&(4a\cdot x^3)+((6a^2+b-1)\cdot x^2+2a\cdot xy)\\
			&+((4a^3+2ab)\cdot x+(a^2+2b-2)\cdot y-z)+(a^4+ a^2b+ b^2-1).
		\end{align*}
		Then we can separate the coefficients of $p_1(x + a, y + b,z+c) - q_1(x, y,z)$ with respect to $x$, $y$ and $z$ in two different methods as following.
		
		\begin{center}
			\newcommand{\tabincell}[2]{
				\begin{tabular}{@{}#1@{}}#2\end{tabular}
			}
			\renewcommand\arraystretch{2}
			\scalebox{0.9}{
				\begin{tabular}{|c:c:c:c|l}
					\multicolumn{1}{c}{$\qquad\qquad\qquad\quad$}     & \multicolumn{1}{c}{$\qquad\qquad\qquad\quad$}     &
					\multicolumn{1}{c}{$\qquad\qquad\qquad\quad$}     &
					\multicolumn{1}{c}{$\qquad\qquad\qquad\quad$} &
					\\
					\cline{4-4}
					\multicolumn{1}{c}{$S$}     & \multicolumn{2}{c}{}     & \multicolumn{1}{|c|}{$a^4 + a^2b + b^2 - 1$} &
					$S_4^D$\\
					\cline{3-3}\cdashline{4-4}[0.5pt/0.5pt]
					\multicolumn{2}{c}{}     &
					\multicolumn{1}{|c}{$4a^3 + 2ab$}    &
					\multicolumn{1}{:c|}{} &
					$S_3^D$\\
					\cline{2-2}
					\cdashline{3-4}[0.5pt/0.5pt]
					\multicolumn{1}{c}{}     & \multicolumn{1}{|c}{$6a^2 + b - 1$}     &
					\multicolumn{1}{:c}{$a^2 + 2b - 2$} &
					\multicolumn{1}{:c|}{} &
					$S_2^D$\\
					\cline{1-1}
					\cdashline{2-4}[0.5pt/0.5pt]
					\multicolumn{1}{|c}{$4a$}     & \multicolumn{1}{:c}{$2a$}     &
					\multicolumn{1}{:c}{} &
					\multicolumn{1}{:c|}{} &
					$S_1^D$\\
					\cdashline{1-4}[0.5pt/0.5pt]
					\multicolumn{1}{|c}{}     & \multicolumn{1}{:c}{}     &
					\multicolumn{1}{:c}{$-1$} &
					\multicolumn{1}{:c|}{$\,$} &
					$S_0^D$\\
					\cline{1-4}
					\multicolumn{1}{c}{$S_1^H$}     & \multicolumn{1}{c}{$S_2^H$}     &
					\multicolumn{1}{c}{$S_3^H$}     &
					\multicolumn{1}{c}{$S_4^H$} &
					\\
				\end{tabular}
			}
		\end{center}
		So we can get $F_{p_1,q_1}=\varnothing$ at once if we use the $\va$-degree cover, while by the $\vx$-homogeneous cover, we will calculate until we get $S_3^H$.
		\item Compute $F_{p_2,q_2}$. We expand $p_2(x + a, y + b,z+c) - q_2(x, y,z)$ and get that
		\begin{align*}
			&p_2(x + a, y + b,z+c) - q_2(x, y,z)\\
			=&((4a+b-1)\cdot x^3+3a\cdot x^2y)\\
			&+((6a^2+ 3ab)\cdot x^2+(3a^2+2b-3)\cdot xy + a\cdot y^2 )\\
			&+((4a^3+ 3a^2b+b^2-1)\cdot x+(a^3+2ab)\cdot y+ (2c-4)\cdot z)\\
			&+(a^4 + a^3b+ ab^2+ c^2 - 4).
		\end{align*}
		Then we can separate the coefficients of $p_2(x + a, y + b,z+c) - q_2(x, y,z)$ with respect to $x$, $y$ and $z$ in two different methods as following.
		\begin{center}
			\newcommand{\tabincell}[2]{
				\begin{tabular}{@{}#1@{}}#2\end{tabular}
			}
			\renewcommand\arraystretch{1.7}
			\scalebox{0.8}{
				\begin{tabular}{|c:c:c:c|lc}
					\multicolumn{1}{c}{$\qquad\qquad\qquad\qquad\qquad$}     & \multicolumn{1}{c}{$\qquad\qquad\qquad\qquad\qquad$}     &
					\multicolumn{1}{c}{$\qquad\qquad\qquad\qquad\qquad$}     &
					\multicolumn{1}{c}{$\qquad\qquad\qquad\qquad\qquad$} &
					\\
					\cline{4-4}		
					\multicolumn{1}{c}{$S$}     & \multicolumn{2}{c}{}     & \multicolumn{1}{|c|}{\tabincell{c}{$a^4 + a^3b+ ab^2+ c^2 - 4$}} & $S_4^D$ & \multicolumn{1}{c}{\tabincell{c}{$\,$\\$\,$}}\\
					\cline{3-3}\cdashline{4-4}[0.5pt/0.5pt]
					\multicolumn{2}{c}{}     &
					\multicolumn{1}{|c}{\tabincell{c}{$4a^3+ 3a^2b+b^2-1$\\$a^3+2ab$}}    &
					\multicolumn{1}{:c|}{} & $S_3^D$ & \multicolumn{1}{c}{\tabincell{c}{$\,$\\$\,$}}\\
					\cline{2-2}
					\cdashline{3-4}[0.5pt/0.5pt]
					\multicolumn{1}{c}{}     & \multicolumn{1}{|c}{\tabincell{c}{$6a^2+ 3ab$\\$3a^2+2b-3$}}     &
					\multicolumn{1}{:c}{} &
					\multicolumn{1}{:c|}{} & $S_2^D$ & \multicolumn{1}{c}{\tabincell{c}{$\,$\\$\,$}}\\
					\cline{1-1}
					\cdashline{2-4}[0.5pt/0.5pt]
					\multicolumn{1}{|c}{\tabincell{c}{$4a+b-1$\\$3a$}}     & \multicolumn{1}{:c}{$a$}     & \multicolumn{1}{:c}{$2c-4$}
					&
					\multicolumn{1}{:c|}{} &
					$S_1^D$ & \multicolumn{1}{c}{\tabincell{c}{$\,$\\$\,$}}\\
					\cdashline{1-4}[0.5pt/0.5pt]
					\multicolumn{1}{|c}{}     & \multicolumn{1}{:c}{}     &
					\multicolumn{1}{:c}{} &
					\multicolumn{1}{:c|}{$\,$} &
					$S_0^D$ & \multicolumn{1}{c}{\tabincell{c}{$\,$\\$\,$}}\\
					\cline{1-4}
					\multicolumn{1}{c}{$S_1^H$}     & \multicolumn{1}{c}{$S_2^H$}     &
					\multicolumn{1}{c}{$S_3^H$}     &
					\multicolumn{1}{c}{$S_4^H$} & \multicolumn{1}{c}{\tabincell{c}{$\,$}} &
					\\
				\end{tabular}
			}
		\end{center}
		So we can get $F_{p_2,q_2}=\varnothing$ if we use $\vx$-homogeneous cover and calculate $S_2^H$, while by $\va$-degree cover, we have to solve $2c-4=0$ needlessly.
	\end{enumerate}	
\end{example}

	\section{Isotropy groups and orbital decompositions}\label{sec:add decomp}
	In this section, we first recall the notion of isotropy groups under shifts, which plays a central role in the summability criteria and existence criteria of telescopers. Then we present different types of partial fraction decompositions of $\bF(\vx)$ with respect to different orbital factorizations as in~\cite{Chen2021}. These decompositions can be computed via algorithms for the SET problem over integers and will be used in the next sections for reducing the rational summability problem and the existence problem of telescopers to simpler cases.
	\subsection{Isotropy groups}\label{sec:isotropy}
	Let $G=\<\si_{x_1}, \ldots, \si_{x_n}>$ be the free abelian group generated by shift operators $\si_{x_1}, \ldots, \si_{x_n}$ and $A$ be a subgroup of $G$. Let $p$ be a multivariate polynomial in $\bF[\vx]$. The set
	\[[p]_A:= \left\{\si(p)\mid \si\in A\right\}\]
	is called the {\em $A$-orbit} of $p$. Two polynomials $p,q\in\bF[\vx]$ are said to be {\em $A$-shift equivalent} or {\em $A$-equivalent} if $[p]_A=[q]_A$, denoted by $p\sim_A q$. The relation $\sim_A$ is an equivalence relation.

	\begin{defn}[Sato's Isotropy Group~\cite{Sato1990}]\label{def:isotropy}
		Let $A$ and $p$ be defined as above. The set
		\begin{equation*}%\label{EQ:Gp}
			A_p:=\{\si\in A\mid \si(p)=p\}.
		\end{equation*}
		is a subgroup of $A$, called the {\em isotropy} group of $p$ in $A$.
	\end{defn}
	
	If two polynomials $p,q$ in $\bF[\vx]$ are $A$-shift equivalent, then $A_p=A_q$. The following remark says that we can test the $A$-equivalence of polynomials and compute a basis of $A_p$ by algorithms for the SET problem over integers in Section~\ref{sec:SET}.
	
	\begin{rem}\label{REM:isotropy}
		\begin{enumerate}
			\item\label{it:iso1} Two polynomials $p,q\in \bF[\vx]$ are $G$-equivalent if and only if there exists a $\si\in G$ such that $\si(p)=q$. Therefore, the $G$-equivalence relation of $p, q$ can be obtained via the computation of $Z_{p,q}$ in Section~\ref{sec:SET}. When $p=q$, the group $G_p$ {is isomorphic to} $Z_{p,p}$. Both of them are free abelian groups and a basis of $G_p$ can be given by a basis of $Z_{p,p}$.
			\item~\label{it:iso2} When $A = \<\si_{x_1}, \ldots, \si_{x_r}>$ with $1\leq r\leq n$, we can view $p,q$ as polynomials in $x_1, \ldots, x_r$ and the other variables are parameters. Then the $A$-equivalence relation of $p,q$ and a basis of the isotropy group $A_p$ are also available by algorithms in Section~\ref{sec:SET}.
			\item In general, let $A = \<\tau_{1}, \ldots, \tau_{r}>$, where $\{\tau_1, \ldots, \tau_r\} (1\leq r\leq n )$ are $\Z$-linearly independent. We will utilize Proposition~\ref{PROP:transformaiton} below to construct a difference isomorphism between $(\F(\vx), \tau_i)$ and $(\F(\vx), \si_{x_i})$ such that $\phi\circ\tau_i = \si_{x_i} \circ \phi$ for $1\leq i \leq r$. Let $B = \<\si_{x_1}, \ldots, \si_{x_r}>$. Then $p$ and $q$ are $A$-equivalent if and only if $\phi(p)$ and $\phi(q)$ are $B$-equivalent. Furthermore, we have $\tau_1^{a_1} \cdots \tau_r^{a_r} \in A_p$ if and only if $\si_{x_1}^{a_r} \cdots \si_{x_r}^{a_r} \in B_{\phi(p)}$ for any $a_1, \ldots, a_r \in \Z$.
		\end{enumerate}
	\end{rem}
	
	A structure property of the quotient group $G/G_p$ is given by Sato~\cite[Lemma A-3]{Sato1990} as follows.
	
	\begin{lem}\label{LEM:G/Gp free}
		$G/G_p$ is a free abelian group.
	\end{lem}
	
	If $p\in \F[\vx]\setminus\bF$ is a non-constant polynomial, then $G_p$ is a proper subgroup of $G$. By Lemma ~\ref{LEM:G/Gp free}, we have $\rank(G_p) < \rank (G)$, where $\rank(G)$ denotes the rank of the free abelian group $G$. This property about the rank of isotropy groups plays a key role in the reduction method of solving rational summability problem and the existence problem of telescopers.
	
	If $n>1$, let $H=\la\si_{x_1},\ldots,\si_{x_{n-1}}\ra$ be the subgroup of $G$ generated by $\si_{x_1},\ldots,\si_{x_{n-1}}$. The isotropy group of $p$ in $H$ is $H_p=\{\tau\in H\mid \tau(p)=p\}$. By Lemma \ref{LEM:G/Gp free}, both $G/G_p$ and $H/H_p$ are free abelian groups. So the ranks of $G_p$ and $H_p$ are strictly less than those of $G$ and $H$ respectively if $p$ has positive degree in $x_1$.
	
	\begin{lem}\label{LEM:Gd/Hd free}
		$G_p/H_p$ is a free abelian group of $\rank (G_p/H_p)\leq1$.
	\end{lem}
	\begin{proof}
		Define a group homomorphism $\varphi: G_p/H_p\to\Z$ by \[\si_{x_1}^{k_1}\cdots\si_{x_n}^{k_n}H_p\mapsto k_n.\]
		It can be verified that $\varphi$ is well-defined. For any $\tau_1,\tau_2\in G_p$, if they are in the same coset of $H_p$ in $G_p$, then $\tau_1\tau_2^{-1}\in H_p$. This implies $\tau_1\tau_2^{-1}\in H$ and hence $\varphi(\tau_1H_p)=\varphi(\tau_2H_p)$. Moreover, the converse is true since $G_p\cap H=H_p$. So $\varphi$ is injective. Then we have $G_p/H_p\cong \im \varphi=k\Z$ for some integer $k\in \Z$.
		So $G_p/H_p$ is a free abelian group generated by $\varphi^{-1}(k)$.
	\end{proof}
        \begin{rem}\label{REM:isotropy_normalization}
            Let $p$ be a polynomial in $\bF[\vx]$. By Remark~\ref{REM:isotropy}.\eqref{it:iso1}, one can compute a basis $\{\tau_1,\tau_2,\ldots,\tau_r\}$ of $G_p$. If $\tau_i\in H$ for all $i=1,\ldots, r$, then $G_p=H_p$ and $G_p/H_p =\bf\{1\bf\}$. So $\rank(G_p/H_p)= 0$ and $\{\tau_1,\tau_2,\ldots,\tau_r\}$ is a basis of $H_p$. If $\tau_\ell\notin H$ for some $\ell\in\{1,\ldots,r\}$, then $\rank(G_p/H_p)=1$. Write $\tau_i =\sigma_{x_1}^{b_{i,1}}\cdots\sigma_{x_n}^{b_{i,n}}$ with $b_{i,j}\in \set Z$ for each $i=1,\ldots,r$. Let $B=(b_{i,j})\in \set Z^{r\times n}$. Since $\tau_\ell\notin H$, we have $b_{\ell,n}\neq 0$. Using unimodular row reduction, one can compute a unimodular matrix $U\in \set Z^{r\times r}$ such that $C=UB$, where $C=(c_{i,j})\in \set Z^{r\times n}$ satisfies $c_{1,n}=\gcd(b_{1,n},b_{2,n},\ldots, b_{r,n})\neq 0$ and $c_{i,n} =0$ for all $i=2,\ldots,r$. Let $\sigma_i=\sigma_{x_1}^{c_{i,1}}\cdots\sigma_{x_n}^{c_{i,n}}$ for each $i=1,\ldots,r$. Then $\{\sigma_1,\ldots,\sigma_r\}$ is another basis of $G_p$ because $U$ is an invertible matrix over $\set Z$. Moreover, $G_p/H_p = \<\bar \sigma_1>$ and $\{\sigma_2,\ldots,\sigma_r\}$ is a basis of $H_p$.
        \end{rem}

	\begin{example}\label{Eg:isotropy}Consider polynomials in $\bQ[x,y,z]$. Let $G = \<\si_x, \si_y, \si_z>$ and $H = \<\si_x, \si_y>$.
		\begin{enumerate}%[label = (\arabic*)]
			\item\label{it:isotropy1} For $p = x^2 + 2xy + z^2$, we have $G_p = H_p = \{\bf 1\}$.
			\item\label{it:isotropy2} For $p = (x - 3y)^2(y + z) + 1$, we have $G_p = \<\tau>$ and $H_p = \{\bf 1\}$, where $\tau = \si_x^3\si_y\si_z^{-1}$. So $G_p/H_p = \<\bar \tau>$, where $\bar \tau = \tau H_p$ denotes the coset in $G_p/H_p$ represented by $\tau\in G_p$.
			\item\label{it:isotropy3} Let $p = x + 2y + z$, we have $G_p = \<\tau_1,\tau_2>$ and $H_p = \<\tau_2>$, where $\tau_1 = \si_x \si_y^{- 1}\si_z$ and $\tau_2 = \si_x^2\si_y^{-1}$. So $G_p/H_p =\<\bar \tau_1>$.
		\end{enumerate}
	\end{example}
	\subsection{Orbital decompositions}\label{subsec:orbital decomp}
	A polynomial $p \in \F[\vx]$ is said to be {\em monic} if the leading coefficient of $p$ is $1$ under a fixed monomial order. Let $\hx_1$ denote the $m-1$ variables $x_2, \ldots, x_m$. For any subgroup $A$ of $G = \<\si_{x_1},\ldots,\si_{x_n}>$ and any polynomial $Q$ in $\bF(\hx_1) [x_1]$, one can classify all of the monic irreducible factors in $x_1$ of $Q$ into distinct $A$-orbits which leads to a factorization
	\begin{equation*}%\label{EQ: orbital_factorization}
		Q = c \cdot \prod_{i = 1}^I \prod_{j = 1}^{J_i} \tau_{i, j} (d_i)^{e_{i, j}},
	\end{equation*}
	where $c \in \bF(\hx_1)$, $I, J_i, e_{i, j} \in \N$, $\tau_{i, j} \in A$, $d_i\in \bF[\vx]$ being monic irreducible polynomials in distinct $A$-orbits, and for each $i$, $\tau_{i, j} (d_i) \neq\tau_{i, j^\prime} (d_i)$ if $1\leq j \neq j^\prime \leq J_i$. With respect to this fixed representation, we have the unique irreducible partial fraction decomposition for a rational function $f = P/Q \in \F(\vx)$ of the form
	\begin{equation}\label{EQ:orbital_parf}
		f =p + \sum_{i = 1}^I \sum_{j = 1}^{J_i} \sum_{\ell = 1}^{e_{i, j}} \frac{a_{i, j, \ell}}{\tau_{i, j} (d_i)^\ell},
	\end{equation}
	where $p, a_{i, j, \ell} \in \F(\hx_1)[x_1]$ with $\deg_{x_1} (a_{i, j, \ell}) < \deg_{x_1} (d_i)$ for all $i, j, \ell$. Note that the representation in~\eqref{EQ:orbital_parf} depends on the choice of representatives $d_i$ in distinct $A$-orbits. However, the sum $\sum_{j = 1}^{J_i}\frac{a_{i, j, \ell}}{\tau_{i, j} (d_i)^\ell}$ only depends on the multiplicity $\ell$ and the orbit $[d_i]_A$ instead of its representative $d_i$. Based on this fact, we shall formulate a unique decomposition of a rational function with respect to the group $A$. In this sense, we can decompose $\bF(\vx)$ as a vector space over $\bE=\bF(\hx_1)$.
	
	Given an irreducible polynomial $d\in\bF[\vx]$ with $\deg_{x_1}(d)>0$ and $j\in \N^+$, we define a subspace of $\bF(\vx)$
	
	\begin{equation}\label{EQ: subspace V}
		V_{[d]_A,j}=\Span_{\bE}\left\{\ \frac{a}{\tau (d)^j}\ \middle \vert a\in \bE[x_1],
		\tau\in A, \deg_{x_1}(a)<\deg_{x_1}(d)\right\}.
	\end{equation}
	For any fraction in $V_{[d]_A,j}$, the irreducible factors of its denominator are in the same $A$-orbit as $d$. Let $V_0=\bE[x_1]$ denote the set of all polynomials in $x_1$. By the irreducible partial fraction decomposition, any rational function $f\in\bF(\vx)$ can be uniquely written in the form
	\begin{equation}\label{EQ:f add decomp A}
		f=f_0+ \sum_{j}\sum_{[d]_A}f_{[d]_A,j},
	\end{equation}
	where $f_0\in V_0$ and $f_{[d]_A,j}$ are in distinct $ V_{[d]_A,j}$ spaces. Let $T_A$ be the set of all distinct $A$-orbits of monic irreducible polynomials in $\bF[\vx]$, whose degrees with respect to $x_1$ are positive. Then $\bF(\vx)$ has the following direct sum decomposition
	\begin{equation}\label{EQ:add decomp}
		\bF(\vx)=V_0\bigoplus\left(\bigoplus_{j\in \N^+}\bigoplus_{[d]_A\in T_A}V_{[d]_A,j}\right).
	\end{equation}
	%where $j$ runs over all positive integer and $[d]_A$ runs over all elements in $T_A$.
	\begin{defn}
		The decomposition~\eqref{EQ:add decomp} of $\F(\vx)$ is called the {\em orbital decomposition} of $\bF(\vx)$ with respect to the variable $x_1$ and the group $A$. If $f$ is written in the form~\eqref{EQ:f add decomp A}, we call $f_0$ and $f_{[d]_A, j}$ {\em orbital components} of $f$ with respect to $x_1$ and $A$.
	\end{defn}
	A key feature of subspaces $V_{[d]_A,j}$ is the $A$-invariant property. In the field of univariate rational functions, the orbital decomposition of $\bF(x_1)$ with respect to the group $A=\la\si_{x_1}\ra$ was first given in \cite{Karr1981} by Karr.

	\begin{lem}\label{LEM:G-invariant}
		If $f\in V_{[d]_A,j}$ and $P\in \bF(\hx_1)[A]$, then $P(f)\in V_{[d]_A,j}$.
	\end{lem}
	\begin{proof}
		Let $f=\sum a_i/\tau_i(d)^j$ and $P=\sum p_\si \si$ with $p_\si\in\bF(\hx_1)$ and $\si\in A$. For any $\si\in A$, we have that $\si\tau_i$ is still in $A$, because $A$ is a group. Since the shift operators do not change the degree and multiplicity of a polynomial, we have $\deg_{x_1}(\si(a_i))<\deg_{x_1}(d)$ and then $\frac{p_\si \si(a_i)}{\si(\tau_i(d))^j}$ is in $V_{[d]_A,j}$. So $P(f)\in V_{[d]_A,j}$ by the linearity of the space.
	\end{proof}
	
	\begin{example}\label{Eg: f1_sum}
		Let $\bF = \bQ$, $\bE = \bQ(y, z)$ and $G = \<\si_x, \si_y, \si_z>$. Consider the rational function $f_1$ in $\bQ(x, y, z)$ of the form
		\[f_1 = \frac{x - z^2}{\underbrace{x^2 + 2xy + z^2}_{d_1 := d_{1,1}}} + \frac{x - y - 2z}{\underbrace{x^2 + 2xy + 2x + z^2}_{d_{1,2}}} + \frac{y + z^2}{\underbrace{x^2 + 2xy + 8x + 2y + z^2 - 2z + 8}_{d_{1,3}}}.\]	
		If $A = \<\si_x>$, then the orbital partial fraction decomposition of $f_1$ is
		\[f_1 = f_{1, 1} + f_{1, 2} + f_{1,3}\text{ with } f_{1,1}=\frac{x - z^2}{d_{1, 1}},~f_{1, 2} = \frac{x - y - 2z}{d_{1, 2}}\text{ and } f_{1, 3} = \frac{y + z^2}{d_{1, 3}},\]
		where $f_{1, i} \in V_{[d_{1, i}]_A, 1}$ for $i = 1, 2, 3$ and $d_1, d_{1, 2}, d_{1, 3}$ are in distinct $\<\si_x>$-orbits. If $A = \<\si_x, \si_y>$, then the orbital partial fraction decomposition of $f_1$ is
		\[f_1 = f_{1, 1} + f_{1, 2}\text{ with } f_{1,1}=\frac{x - z^2}{d_1} + \frac{x - y - 2z}{\si_y(d_1)}\text{ and } f_{1, 2} = \frac{y + z^2}{d_{1, 3}},\]
		where $f_{1, 1}\in V_{[d_1]_A, 1}$, $f_{1, 2}\in V_{[d_{1, 3}]_A, 1}$ and $d_1, d_{1,3}$ are in distinct $\<\si_x, \si_y>$-orbits. If $A = \<\si_x, \si_y, \si_z>$, then $f_1\in V_{[d_1]_A, 1}$ is one component in the orbital decomposition because
		\[f_1 = \frac{x - z^2}{d_1} + \frac{x - y - 2z}{\si_y(d_1)} + \frac{y + z^2}{\si_x\si_y^3\si_z^{-1}(d_1)}.\]
	\end{example}
	
	\begin{example}\label{Eg: f_sum}
		Let $\bF = \bQ$ and $G = \<\si_x, \si_y, \si_z>$. Consider the rational function $f = f_1 + f_2 + f_3$ in $\bQ(x, y, z)$ with $f_1$ given in Example~\ref{Eg: f1_sum},
		\[f_2 = \frac{x + z}{\underbrace{(x - 3y)^2(y + z) + 1}_{d_2}}~\text{and}~f_3 = \left(y + \frac{z}{y^2 + z - 1} - \frac{1}{y^2 + z}\right)\frac{1}{(\underbrace{x + 2y + z}_{d_3})^2}.\]
		If $A = G$, then the orbital partial fraction decomposition of $f$ is
		\[ f = f_1 + f_2 + f_3\quad\text{with}\quad f_i \in V_{[d_i]_G, 1} \text{ for } i = 1, 2\text{ and } f_3 \in V_{[d_3]_G, 2},\]
		where $d_1, d_2, d_3$ are in distinct $\<\si_x, \si_y, \si_z>$-orbits.
	\end{example}

	\section{The rational summability problem}\label{sec:summability}
	In this section, we solve the rational summability problem for multivariate rational functions and design an algorithm for rational summability testing. In Section~\ref{subsec:reduction_sum} we use a special orbital decomposition in Section~\ref{subsec:orbital decomp} to reduce the summability problem of a general rational function to its orbital components and then further to simple fractions by Abramov's reduction. In Section~\ref{subsec:sum criterion}, we use the structure of isotropy groups to reduce the number of variables in the summability problem inductively.

	\subsection{Orbital reduction for summability}\label{subsec:reduction_sum}
	Let $f$ be a rational function in $\bF(\vx)$, where $\vx = \{x_1, \ldots, x_m\}$. Recall that $\hx_1=\{x_2,\ldots,x_m\}$. Let $n$ be an integer such that $1\leq n\leq m$. We consider the $(\si_{x_1},\ldots,\si_{x_n})$-summability problem of $f$ in~$\bF(\vx)$. Let $G=\langle \si_{x_1}, \ldots, \si_{x_n}\rangle$. Taking $\bE=\bF(\hx_1)$ and $A=G$ in equality~\eqref{EQ: subspace V}, we get the subspace $V_{[d]_G,j}$ of $\bF(\vx)$
	
	\begin{equation*}%\label{EQ: subspace VG}
		V_{[d]_G,j}=\Span_{\bE}\left\{\ \frac{a}{\tau (d)^j}\ \middle \vert a\in \bE[x_1],
		\tau\in G, \deg_{x_1}(a)<\deg_{x_1}(d)\right\}.
	\end{equation*}
    where $j\in \bN^+$ and $d\in \bE[\vx]$ is irreducible with $\deg_{x_1}(d)>0$.
	According to Equation~\eqref{EQ:f add decomp A}, $f$ can be decomposed into the form
	\begin{equation}\label{EQ:f add decomp G}
		f=f_0+ \sum_{j}\sum_{[d]_G}f_{[d]_G,j},
	\end{equation}
	where $f_0\in V_0=\bE[x_1]$ and $f_{[d]_G,j}$ are in distinct $ V_{[d]_G,j}$ spaces. Then we obtain the orbital decomposition~\eqref{EQ:add decomp} of $\bF(\vx)$ with respect to the group $A=G$.
	
	\begin{lem}\label{LEM:red1}
		Let $f\in\bF(\vx)$. Then $f$ is $(\si_{x_1},\ldots,\si_{x_n})$-summable if and only if $f_0$ and each $f_{[d]_G,j}$ are $(\si_{x_1},\ldots,\si_{x_n})$-summable for all $[d]_G\in T_G$ and $j\in \N^+$.
	\end{lem}
	\begin{proof}
		The sufficiency is due to the linearity of difference operators $\Delta_{x_i}$. For the necessity, suppose $f=\sum_{i=1}^n\Delta_{x_i}(g^{(i)})$ with $g^{(i)}\in\bF(\vx)$. By the orbital decomposition of rational functions \eqref{EQ:f add decomp G}, we can write $f, g^{(i)}$ in the form
		\begin{equation*}
			f=f_0+\sum_{j}\sum_{[d]_G} f_{{[d]_G,j}}\quad\text{and}\quad g^{(i)}=g^{(i)}_{0}+\sum_{j}\sum_{[d]_G} g^{(i)}_{{[d]_G,j}} \ \text{ for } 1\leq i\leq n.
		\end{equation*}
		By the linearity of $\Delta_{x_i}$, we see
		\begin{equation*}%label{EQ: decomp1}
			f=\sum_{i=1}^n\Delta_{x_i} \left(g^{(i)}_{0}\right)+\sum_{j}\sum_{[d]_G}\left(\sum_{i=1}^n\Delta_{x_i}\left(g^{(i)}_{[d]_G,j}\right)\right).
		\end{equation*}
		By Lemma \ref{LEM:G-invariant}, it is another expression of $f$ with respect to $V_{[d]_G,j}$. Such a decomposition is unique, so $f_0=\sum_{i=1}^n\Delta_{x_i} (g^{(i)}_{0})$ and
		$ f_{[d]_G,j}=\sum_{i=1}^n\Delta_{x_i}(g^{(i)}_{[d]_G,j})$,
		which are $(\si_{x_1},\ldots,\si_{x_n})$-summable.
	\end{proof}
	
	Using Lemma \ref{LEM:red1}, we can reduce the summability problem of a rational function to its orbital components. Note that polynomials in $x_1$ are always $(\si_{x_1})$-summable. Thus Problem \ref{PROB:summability problem} can be reduced to that for rational functions in $V_{[d]_G,j}$, which are of the form
	\begin{equation}\label{EQ:red1}
		f = \sum_\tau\frac{a_{\tau}}{\tau(d)^j},
	\end{equation}
	where $\tau\in G$, $a_\tau\in \bF(\hx_1)[x_1]$, $d\in \bF[\vx]$ with $\deg_{x_1}(a_\tau)<\deg_{x_1}(d)$ and $d$
	is irreducible in $x_1$ over $\bF(\hx_1)$.
	
	Let $\si$ be an automorphism on $\bF(\vx)$ and $a,b\in \bF(\vx)$. Then for any integer $k\in \Z$, we have the reduction formula
	\begin{equation}\label{EQ:red formula}
		\frac{a}{\si^k(b)} = \si(h) - h + \frac{\si^{-k}(a)}{b},
	\end{equation}
	where $h = 0$ if $k = 0$, $h = \sum_{i=0}^{k-1} \frac{\si^{i-k}(a)}{\si^i(b)}$ if $k>0$ and $h = -\sum_{i=0}^{-k-1} \frac{\si^{i}(a)}{\si^{i+k}(b)}$ if $k<0$.
	For any $\tau=\si_{x_1}^{k_1}\cdots\si_{x_n}^{k_n}\in G$, applying the reduction formula~\eqref{EQ:red formula} with $\si=\si_{x_i}$ for $i=1,\ldots,n$, we get
	\begin{equation}\label{EQ:red formula2}
		\frac{a}{\si_{x_1}^{k_1}\cdots\si_{x_n}^{k_n}(b)}=\sum_{i=1}^n\left(\si_{x_i}(h_i)-h_i\right)+\frac{\si_{x_1}^{-k_1}\cdots\si_{x_n}^{-k_n}(a)}{b},
	\end{equation}
	where
	\[h_i = \left\{
	\begin{aligned}
		& \quad \quad 0, &\text{~if~} k_i = 0,\\
		&\sum_{\ell = 0}^{k_i - 1} \frac{\si_{x_i}^{\ell-k_i} \si_{x_{i - 1}}^{- k_{i - 1}} \cdots \si_{x_1}^{- k_1} (a)}{\si_{x_i}^\ell \si_{x_{i + 1}}^{k_{ i + 1}} \cdots \si_{x_n }^{k_n}(b)},&\text{~if~} k_i > 0,\\
		&-\sum_{\ell = 0}^{- k_i -1} \frac{\si_{x_i}^{\ell} \si_{x_{i - 1}}^{- k_{i - 1}} \cdots \si_{x_1}^{- k_1} (a)}{\si_{x_i}^{\ell + k_i} \si_{x_{i + 1}}^{k_{ i + 1}} \cdots \si_{x_n}^{k_n}(b)}, &\text{~if~} k_i < 0.
	\end{aligned}
	\right.\]
	for $i =1,\ldots,n$. The equation~\eqref{EQ:red formula2} is called the $(\si_{x_1}, \ldots, \si_{x_n})$-{\em reduction} formula. Rewriting every fraction of $f$ in~\eqref{EQ:red1} by the reduction formula~\eqref{EQ:red formula2}, we get the following lemma.
	\begin{lem}\label{LEM:red2}
		Let $f\in V_{[d]_G,j}$ be in the form~\eqref{EQ:red1}. Then we can decompose it into the form
		\begin{equation}\label{EQ:red2}
			f=\sum_{i=1}^n\Delta_{x_i}(g_i)+r \text{ with } r=\frac{a}{d^j},
		\end{equation}
		where $g_i\in\bF(\vx)$, $a=\sum_\tau\tau^{-1}(a_\tau)$ with $\deg_{x_1}(a)<\deg_{x_1}(d)$. In particular, $f$ is $(\si_{x_1},\ldots,\si_{x_n})$-summable if and only if $r$ is $(\si_{x_1},\ldots,\si_{x_n})$-summable.
	\end{lem}
	
	\begin{example}\label{Eg: f1_red}
		Consider the rational function $f_1\in \bQ(x, y, z)$ given in Example~\ref{Eg: f1_sum}. Then $f_1\in V_{[d_1]_G, 1}$ and it can be written as
		\[f_1 = \frac{x - z^2}{d_1} + \frac{x - y - 2z}{\si_y(d_1)} + \frac{y + z^2}{\si_x\si_y^3\si_z^{-1}(d_1)},\]
		where $d_1 = x^2 + 2xy + z^2.$ By applying the $(\si_x, \si_y, \si_z)$-reduction formula, we have
		\begin{equation*}
			f_1 = \Delta_x (u_1) + \Delta_y (v_1) + \Delta_z (w_1) + r_1 \text{ with } r_1 = \frac{2x - 1}{d_1},
		\end{equation*}
		where
		\[u_1 = \frac{y + z^2}{\si_y^3\si_z^{-1} (d_1)},~v_1 =  \frac{x-y + 1 - 2z}{d_1} + \sum_{\ell = 0}^2\frac{y + \ell - 3 + z^2}{\si_y^\ell\si_z^{-1}(d_1)}, ~w_1 = -  \frac{y   -3 + z^2}{\si_z^{-1} (d_1)}.\]
		Then $f_1$ is $(\si_x, \si_y, \si_z)$-summable if and only if $r_1$ is $(\si_x, \si_y, \si_z)$-summable.
	\end{example}
	
	The results in Lemmas \ref{LEM:red1} and~\ref{LEM:red2} are summarized as follows.
	The following lemma  reduces the rational summability problem from general rational functions to simple fractions.
	\begin{cor}\label{COR:summable}
		Let $f\in\bF(\vx)$. Then we can decompose $f$ into the form
		\begin{equation}\label{EQ:red mm}
			f=\sum_{i=1}^n\Delta_{x_i}(g_i)+r \text{ with } r=\sum_{i=1}^I\sum_{j=1}^{J_i}\frac{a_{i,j}}{d_i^j},
		\end{equation}
		where $g_i\in\bF(\vx)$, $a_{i,j}\in\bF(\hx_1)[x_1]$, $d_i\in\bF[\vx]$ with $\deg_{x_1}(a_{i, j})<\deg_{x_1}(d_i)$ and the $d_i$'s are monic irreducible polynomials in
		distinct $G$-orbits. Furthermore, $f$ is $(\si_{x_1},\ldots,\si_{x_n})$-summable if and only if each $a_{i, j}/d_i^j$ is $(\si_{x_1},\ldots,\si_{x_n})$-summable for all $i,j$ with $1\leq i\leq I$ and $1\leq j\leq J_i$.
	\end{cor}
	\subsection{Summability criteria}\label{subsec:sum criterion}
	By Corollary~\ref{COR:summable}, we reduce the rational summability problem to that for simple fractions
	\begin{equation}\label{EQ:simple fraction}
		f = \frac{a}{d^j},
	\end{equation}
	where $j\in\N^+$, $a\in \bF(\hx_1)[x_1]$ and $d\in \bF[\vx]$ is irreducible with $\deg_{x_1}(a)<\deg_{x_1}(d)$.
	In this section, we shall present a criterion on the summability for such simple fractions.
	
	For the univariate summability problem, we recall the following well-known result in~\cite{Abramov1975, Abramov1995,Matusevich2000,Paule1995b,Pirastu1995b,AshCatoiu2005}. Since the univariate case is the base of our induction method, we give a proof for the sake of completeness.
	\begin{lem}\label{Lem:n=1}
		Let $f\in \F(\vx)$ be of the form~\eqref{EQ:simple fraction}. Then $f$ is $ (\si_{x_1})$-summable in $\bF(\vx)$ if and only if $a=0$.
	\end{lem}	
	\begin{proof}
		The sufficiency is trivial since $f=\Delta_{x_1}(0)$. To show the necessity, suppose $f$ is $ (\si_{x_1})$-summable but $a\neq 0$. Since $f=a/d^j\in V_{[d]_G,j}$, by the proof of Lemma~\ref{LEM:red1} we can further assume $f=\Delta_{x_1}(g)$ for some $g\in V_{[d]_G,j}$. Write $g$ in the form $g=\sum_{i=\ell_0}^{\ell_1}a_i/\si_{x_1}^i(d)^j$ with $a_{\ell_0}a_{\ell_1}\neq0$. Then
		\begin{equation*}
			f=\Delta_{x_1}(g)=\sum_{i=\ell_0}^{\ell_1+1}\frac{\tilde{a}_i}{\si_{x_1}^i(d)^j},
		\end{equation*}
		where $\tilde{a}_i=\si_{x_1}(a_{i-1})-a_i$ for $\ell_0+1\leq i\leq\ell_1$, $\tilde{a}_{\ell_0}=-a_{\ell_0}$ and $\tilde{a}_{\ell_1+1}=\si_{x_1}(a_{\ell_1})$. Note that $\tilde a_{\ell_0}$ and $\tilde a_{\ell_{1}+1}$ are nonzero. For any integer $i\in \Z$, $\si_{x_1}^i(d)$ is still an irreducible polynomial. However, there is only one irreducible factor in the denominator of $f=a/d^j$. So we must have $\si_{x_1}^i(d)=d$ for some nonzero integer $i$. It implies that $d$ is free of $x_1$. This is a contradiction because $d$ has positive degree in $x_1$.
	\end{proof}
	For the multivariate summability problem with $n>1$, let $G=\<\si_{x_1},\ldots,\si_{x_n}>$ and $H=\<\si_{x_1},\ldots,\si_{x_{n-1}}>$. The isotropy groups of the polynomial $d$ in $G$ and $H$ are denoted by $G_d$ and $H_d$, respectively, i.e.,
	\[G_d=\{\tau\in G\mid \tau(d) = d\}\quad \text{and} \quad H_d=\{\tau\in H\mid \tau(d)=d\}.\]
	By Lemma \ref{LEM:Gd/Hd free}, we know either $\rank (G_d/H_d)=0$ or $\rank (G_d/H_d)=1$.
	
	When $\rank(G_d/H_d)=0$, the summability problem in $n$ variables can be reduced to that in $n-1$ variables.

	\begin{lem}\label{THM:rank=0}
		Let $f = a/d^j\in \F(\vx)$ be of the form~\eqref{EQ:simple fraction}. If $n>1$ and $\rank(G_d/H_d)=0$, then $f$ is $(\si_{x_1},\ldots,\si_{x_n})$-summable in $\bF(\vx)$ if and only if $f$ is $(\si_{x_1},\ldots,\si_{x_{n-1}})$-summable in $\bF(\vx)$.
	\end{lem}
	\begin{proof}
		The sufficiency is obvious by definition. For the necessity, suppose $f$ is $(\si_{x_1},\ldots,\si_{x_n})$-summable but not $(\si_{x_1},\ldots,\si_{x_{n-1}})$-summable. By the orbital decomposition of $f$ in \eqref{EQ:f add decomp G} and Lemma~\ref{LEM:red1}, we get
		\begin{equation}\label{EQ:rk=0_summable}
			f=\Delta_{x_1}(g_1)+\cdots+\Delta_{x_n}(g_n)
		\end{equation}
		with $g_1,\ldots,g_n$ in the same subspace $V_{[d]_G,j}$ as $f$. As an analogue to \eqref{EQ:red2} in $n-1$ variables $x_1,\ldots,x_{n-1}$, we can decompose $g_n$ as
		\begin{equation}\label{EQ:rk=0_decomp}
			g_n=\sum_{i=1}^{n-1}\Delta_{x_i}(u_i)+\sum_{\ell=0}^{\rho}\frac{\lam_\ell}{\si_{x_n}^{\ell}(\mu)^j},
		\end{equation}
		where $u_i\in \F(\vx)$, $\rho\in\N$, $\lam_\ell\in\F(\hx_1)[x_1]$, $\mu\in \F[\vx]$ with $\deg_{x_1}(\lam_{\ell})<\deg_{x_1}(d)$ and $\mu$ is in the same $G$-orbit as $d$.
		
		Furthermore, we can assume $\lam_{0}\lam_{\rho}\neq0$ and each nonzero $\lam_\ell/\si_{x_n}^{\ell}(\mu)^j$ is not $(\si_{x_1},\ldots,\si_{x_{n-1}})$-summable. Substituting $g_n$ in \eqref{EQ:rk=0_decomp} into \eqref{EQ:rk=0_summable}, we get
		
		\begin{equation}\label{EQ:rk=0_red}
			f+\sum_{\ell=0}^{\rho+1}\frac{\tilde{\lam}_\ell}{\si_{x_n}^{\ell}(\mu)^j}=\sum_{i=1}^{n-1}\Delta_{x_i}(h_i),
		\end{equation}
		where $\tilde{\lam}_{0}=\lam_0$, $\tilde{\lam}_{\rho+1}=-\si_{x_n}(\lam_\rho)$, $\tilde{\lam}_\ell=\lam_\ell-\si_{x_n}(\lam_{\ell-1})$ for all $1\leq\ell\leq\rho$ and $h_i=g_i+\Delta_{x_n}(u_i)$ for all $1\leq i\leq n-1$.

		Since $\rank(G_d/H_d)=0$ and $G_d=G_\mu$, it follows that all $\si_{x_n}^\ell(\mu)$ with $\ell\in \Z$ are in distinct $H$-orbits.
		In particular, $[\mu]_H,[\si_{x_n}(\mu)]_H\ldots,[\si_{x_n}^{\rho+1}(\mu)]_H$ are distinct $H$-orbits. On the other hand, the left hand side of \eqref{EQ:rk=0_red} is $(\si_{x_1},\ldots,\si_{x_{n-1}})$-summable, but $\tilde{\lam}_{0}/\mu^j$ is not $(\si_{x_1},\ldots,\si_{x_{n-1}})$-summable according to the assumption. By Lemma \ref{LEM:red1} (in $n-1$ variables), the only choice is that $\mu\sim_H d$. Similarly, $\si_{x_n}^{\rho+1}(\mu)\sim_H d$ and hence $\mu\sim_H \si_{x_n}^{\rho+1}(\mu)$. This leads to a contradiction since $\rho$ is a non-negative integer.
	\end{proof}

	\begin{lem}\label{LEM:exchange bases}
		Let $f\in \F(\vx)$ and $K$ be a subgroup of $G=\<\si_{x_1},\ldots,\si_{x_{ n}}>$ with rank $r~(1\leq r\leq n)$. If $\{\si_{i}\}_{i=1}^{r}$ and $\{\tau_i\}_{i=1}^{r}$ are two bases of $K$, then $f$ is $(\si_1,\ldots,\si_r)$-summable if and only if $f$ is $(\tau_1,\ldots,\tau_r)$-summable.
	\end{lem}
	To prove the basis change property of the summability problem in Lemma~\ref{LEM:exchange bases}, we first show the following lemma. It can be seen as a variant of the reduction formula~\eqref{EQ:red formula2}. Since it is useful in computation, we give a detailed proof by induction.
	\begin{lem}\label{Formula}
		Let $\si_1, \ldots, \si_r$ be elements in $G$ and $K=\la \si_1,\ldots,\si_r\ra$ be the subgroup of $G$ generated by $\si_1,\ldots, \si_r$. Then for every $\tau\in K$,
		\begin{equation*}
			{\tau}-{\bf 1}=(\si_1-{\bf 1})\tilde{\si}_1+\cdots+({\si_r}-{\bf 1})\tilde{\si}_r,
		\end{equation*}
		for some $\tilde{\si}_i\in \bF[K]$.
	\end{lem}
	
	\begin{proof}
		We prove this lemma by induction on the number of $\si_i$.
		If $r=1$, then $\tau=\si_1^{k_1}$ for some $k_1\in\Z$. We have $\si_1^{k_1}-{\bf1}=(\si_1-{\bf 1})\mu$, where $\mu = 0$ if $k_1 = 0$, $\mu=\sum_{i=0}^{k_1-1}\si_1^{i}$ if $k_1>0$ and $\mu=-\sum_{i=0}^{-k_1-1}\si_1^{i+k_1}$ if $k_1<0$. If $r\geq 2$, assume that the conclusion holds for $r-1$. Write $\tau=\si_1^{k_1}\cdots\si_r^{k_r}$ for some $k_1,\ldots,k_r\in\Z$. Then
		
		\[\tau-{ \bf 1}=\left(\si_1^{k_1}-{ \bf 1}\right)\si_2^{k_2}\cdots\si_r^{k_r}+\left(\si_2^{k_2}\cdots\si_r^{k_r}-{ \bf 1}\right).\]
		If $\si_2^{k_2}\cdots\si_r^{k_r}={ \bf 1}$, then we are done. Otherwise, by the inductive hypothesis, we get ${\tau}-{\bf 1}=(\si_1-{\bf 1})\tilde{\si}_1+\cdots+({\si_r}-{\bf 1})\tilde{\si}_r$ for some $\tilde{\si}_1,\ldots,\tilde{\si}_r\in \bF[K]$. In fact, the above argument gives the following explicit expression
		\[\tilde\si_i =\left\{ \begin{aligned}
			&\quad \quad 0 & \text{if } k_i = 0,\\
			&\sum_{\ell = 0}^{k_i -1} \si_{i}^{\ell} \si_{i + 1}^{k_{i + 1}}\cdots\si_{r}^{k_r} & \text{if } k_i >0,\\
			&-\sum_{\ell = 0}^{-k_i -1} \si_{i}^{\ell + k_i} \si_{i + 1}^{k_{i + 1}}\cdots\si_{r}^{k_r} & \text{if } k_i< 0,
		\end{aligned}\right.\]
		for $i = 1, \ldots, r$.
	\end{proof}
	
	\begin{proof}[Proof of Lemma~\ref{LEM:exchange bases}]
		Suppose $f$ is $(\tau_1, \ldots, \tau_r)$-summable. This means
		\begin{equation}\label{EQ: tau-sum}
			f = \Delta_{\tau_1} (h_1) + \cdots + \Delta_{\tau_r} (h_r)
		\end{equation}
		for some $h_1, \ldots, h_r\in \bF(\vx)$. For each $i = 1, \ldots, r$, since $\tau_i \in \<\si_1, \ldots, \si_r>$, it follows from Lemma~\ref{Formula} that ${\tau_i}-{\bf 1}=(\si_1-{\bf 1})\tilde{\si}_{i, 1}+\cdots+({\si_r}-{\bf 1})\tilde{\si}_{i, r}$ for some
		$\tilde{\si}_{i, j}\in \bF[K]$ with $K$ being the subgroup generated by $\si_1, \ldots, \si_r$. Applying this operator to $h_i$ yields that
		\begin{equation}\label{EQ: tau->si}
			\Delta_{\tau_i} (h_i) = \Delta_{\si_1} (h_{i, 1}) + \cdots + \Delta_{\si_r} (h_{i, r}),
		\end{equation}
		where $h_{i, j} = \tilde\si_{i, j} (h_i)$ for $j = 1, \ldots, r$. Combining Equations~\eqref{EQ: tau-sum} and \eqref{EQ: tau->si}, we have
		\begin{equation*}%\label{EQ: tau->si2}
			f = \sum_{i = 1}^r \Delta_{\tau_i} (h_i) = \sum_{i = 1}^r \sum_{j = 1}^r \Delta_{\si_j} (h_{i, j}) = \sum_{j = 1}^r \Delta_{\si_j}\left(\sum_{i = 1}^r h_{i, j}\right),
		\end{equation*}
		where the last equality follows from the linearity of $\Delta_{\si_j}$. Thus $f$ is $(\si_1, \ldots, \si_r)$-summable. Similarly, the other direction is also true.
	\end{proof}

	\begin{thm} \label{THM:rank=1}
		Let $f = a/d^j\in \F(\vx)$ be of the form~\eqref{EQ:simple fraction}. Let $\{\tau_i\}_{i=1}^r(1\leq r< n)$ be a basis of $G_d$ (take $\tau_1={\bf 1}$, if $G_d=\{\bf1\}$). Then
		$f$ is $(\si_{x_1},\ldots,\si_{x_n})$-summable if and only if
		\[a=\Delta_{\tau_1}(b_1)+\cdots+\Delta_{\tau_r}(b_r)\]
		for some $b_i\in \F(\hx_1)[x_1]$ with $\deg_{x_1}(b_i)<\deg_{x_1}(d)$ for all $1\leq i\leq r$.
	\end{thm}
	\begin{proof}
		The sufficiency follows from the fact that $f=\sum_{i=1}^{r}\Delta_{\tau_i}(b_i/d^j)$ and Lemma~\ref{Formula}. For the necessity, we proceed by induction on $n$.
		If $n=1$, then $G_d$ is a trivial group and the univariate case follows from Lemma~\ref{Lem:n=1}. If $n>1$, suppose the inductive hypothesis is true for $n-1$ as follows.

		{\em If $\{\theta_i\}_{i=1}^{s}$ is a basis of $H_d$, then $f$ is $(\si_{x_1},\ldots,\si_{x_{n-1}})$-summable if and only if $a=\sum_{i=1}^{s}\Delta_{\theta_i}(b_i)$ for some $b_i\in \F(\hx_1)[x_1]$ with $\deg_{x_1}(b_i)<\deg_{x_1}(d)$ for all $1\leq i\leq s$.}
		
		Now we proceed by a case distinction according to the rank of $ G_d/H_d$ which is either $0$ or $1$ by Lemma~\ref{LEM:Gd/Hd free}.
		If $\rank (G_d/H_d)=0$, then $H_d=G_d$. The conclusion follows from Lemma \ref{THM:rank=0} and the inductive hypothesis.
		If $\rank (G_d/H_d)=1$, by Lemma \ref{LEM:exchange bases}, we may assume that $\{\tau_i\}_{i=1}^{r}$ is a basis of $G_d$ such that
		$H_d=\la\tau_1,\ldots,\tau_{r-1}\ra$ and $G_d/H_d=\la\bar{\tau}_r\ra$. Here $\bar\tau_r$ represents the element $\tau_r H_d$ with $\tau_r\in G_d$. Then we can choose $\tau_r=\si_{x_1}^{-k_1}\cdots\si_{x_{n-1}}^{-k_{n-1}}\si_{x_n}^{k_n}$ such that $k_n$ is a positive integer. Otherwise, replace $\tau_r$ by $\tau_r^{-1}$. Since $\bar\tau_r$ is a generator of $G_d/H_d$, we have that $k_n$ is the smallest positive integer such that $\si_{x_n}^{k_n}(d)\sim_H d$.
		
		By the decomposition~\eqref{EQ:add decomp}, we can assume $f=\Delta_{x_1}(g_1)+\cdots+\Delta_{x_n}(g_n)$ with $g_i\in V_{[d]_G,j}$. In here, using Lemma~\ref{LEM:red2}, $g_n$ can be decomposed as
		\begin{equation*}%\label{EQ:rk=0_decomp}
			g_n=\sum_{i=1}^{n-1}\Delta_{x_i}(u_i)+\sum_{\ell=0}^{k_n-1}\frac{\lam_\ell}{\si_{x_n}^{\ell}(d)^j},
		\end{equation*}
		where $u_i\in \F(\vx)$ and $\lam_\ell\in\F(\hx_1)[x_1]$ with $\deg_{x_1}(\lam_{\ell})<\deg_{x_1}(d)$. Then we have
		\begin{equation}\label{EQ:rk=1_red}
			f-\Delta_{x_n}\left(\sum_{\ell=0}^{k_n-1}\frac{\lam_\ell}{\si_{x_n}^{\ell}(d)^j}\right)=\sum_{i=1}^{n-1}\Delta_{x_i}(h_i),
		\end{equation}
		where $h_i=g_i+\Delta_{x_n}(u_i)$.
		Note that $\si_{x_n}^{k_n}(d)=\si_{x_1}^{k_1}\cdots\si_{x_{n-1}}^{k_{n-1}}(d)$
		%\[\frac{\tilde{\lam}_{k_1}}{\si_{x_1}^{k_1}(d)^m}=-\frac{\si_{x_1}(\lam_{k_1-1})}{\si_{x_2}^{k_2}\si_{x_3}^{k_3}(d)^m}\]
		and apply the reduction formula~\eqref{EQ:red formula2} to simplify \eqref{EQ:rk=1_red}.
		We get
		\begin{equation}\label{EQ:rk=1_red2}
			\tilde{f}:=\sum_{\ell=0}^{k_n-1}\frac{\tilde{\lam}_\ell}{\si_{x_n}^{\ell}(d)^j}=\sum_{i=1}^{n-1}\Delta_{x_i}(\tilde{h}_i),
		\end{equation}
		where $\tilde{h}_i\in\F(\vx)$, $\tilde{\lam}_0=a+\lam_{0}-\si_{x_1}^{-k_1}\cdots\si_{x_{n-1}}^{-k_{n-1}}\si_{x_n}(\lam_{k_n-1})$ and
		$\tilde{\lam}_\ell=\lam_\ell-\si_{x_n}(\lam_{\ell-1})$ for $1\leq\ell\leq k_n-1$.
		
		Note that $[d]_H,[\si_{x_n}(d)]_H,\ldots,[\si_{x_n}^{k_n-1}(d)]_H$ are distinct $H$-orbits due to the minimality of $k_n$. From the equation~\eqref{EQ:rk=1_red2}, $\tilde{f}$ is $(\si_{x_1},\ldots,\si_{x_{n-1}})$-summable. So by Lemma~\ref{LEM:red1}, each $\frac{\tilde{\lambda}_\ell}{\si_{x_n}^{\ell}(d)^j}$
		is $(\si_{x_1},\ldots,\si_{x_{n-1}})$-summable for $0\leq \ell\leq k_n-1$. Let $W$ denote the vector subspace of $\bF(\vx)$ over $\bF$ consisting of all elements in the form of $\sum_{i=1}^{r-1}\Delta_{\tau_i}(b_i)$ with $b_i\in \bF(\hx_1)[x_1]$ and $\deg_{x_1}(b_i)<\deg_{x_1}(d)$. (If $r=1$, take $W=\{0\}$.)
		If two rational functions $g,h\in\bF(\vx)$ satisfy the property that $g-h\in W$, we say $g, h$ are congruent modulo $W$, denoted by $g\equiv h$ (mod $W$). Since $H_{d}=H_{\si_{x_n}^\ell(d)}$, we apply the inductive hypothesis to conclude that
		
		\begin{equation*}%\label{EQ:rk=1_relation}
			\left\{ \begin{array}{ll}
				0 \equiv a + \lam_{0}-\si_{x_1}^{-k_1}\cdots\si_{x_{n-1}}^{-k_{n-1}}\si_{x_n}(\lam_{k_n-1})&(\text{mod } W)\\
				0 \equiv \lambda_1 - \sigma_{x_n}(\lambda_{0})&(\text{mod } W)\\
				\qquad\qquad\vdots&\\
				0 \equiv \lambda_{k_n-1} - \sigma_{x_n}(\lambda_{k_n-2})&(\text{mod } W).
			\end{array} \right.
		\end{equation*}
		Since $W$ is $G$-invariant, it follows from the equations that
		\[a\equiv\si_{x_1}^{-k_1}\cdots\si_{x_{n-1}}^{-k_{n-1}}\si_{x_n}^{k_n}(\lam_{0})-\lambda_0 \equiv \Delta_{\tau_r}(\lambda_0)\quad(\text{mod } W).\]
		This completes the proof.
	\end{proof}
	\begin{rem}
		For the bivariate case with $n=2$, Theorem~\ref{THM:rank=1} coincides with the known criterion in~\cite[Theorem 3.3]{HouWang2015} and~\cite[Theorem 3.7]{ChenSinger2014}. In this case, $\rank(G_d)\leq 1$ and $H_d=\{\bf 1\}$. If $\rank { (G_d)}=0$, then $a/d^j$ is $(\si_{x_1},\si_{x_2})$-summable in $\bF(\vx)$ if and only if $a=0$. If $\rank{ (G_d)}=1$ and $G_d$ is generated by $\tau=\si_{x_1}^{\ell_1}\si_{x_2}^{-\ell_2}\in G$ for some $\ell_2\neq 0$, then $a/d^j$ is $(\si_{x_1},\si_{x_2})$-summable if and only if $a=\si_{x_1}^{\ell_1}\si_{x_2}^{-\ell_2}(b)-b$ for some $b\in\bF(\hx_1)[x_1]$ with $\deg_{x_1}(b)<\deg_{x_1}(d)$.
	\end{rem}
	\begin{example}\label{Eg:shift}
		Let $f = 1/(x_1^s + \cdots+x_n^s)\in \Q(x_1,\ldots,x_n)$ with $s,n \in \N \setminus\{0\}$. Let $G_d$ be the isotropy group of $d=x_1^s + \cdots+x_n^s$ in $G=\la \si_{x_1},\ldots,\si_{x_n}\ra$. Then we can decide for all cases the $(\si_{x_1},\ldots, \si_{x_n})$-summability of $f$ in $\bQ(x_1, \ldots, x_n)$.
		\begin{enumerate}
			\item If $s=1$ and $n>1$, then $d$ is irreducible. The rank of $G_d$ is $n-1$ and one basis is given by $\tau_1,\ldots,\tau_{n-1}$ with $\tau_i=\si_{x_i}\si_{x_{i+1}}^{-1}$ for $i=1,\ldots,n-1$. Since $1=\tau_1(x_1)-x_1$, it follows that $f$ is $(\si_{x_1},\ldots,\si_{x_n})$-summable. In fact, we have
			\[\frac{1}{x_1+\cdots+x_n}=\Delta_{x_1}\left(
			\frac{x_1}{x_1+\cdots+x_n}\right)+\Delta_{x_2}\left(
			\frac{-x_1-1}{x_1+\cdots+x_n}\right).\]
			This means $f$ is $(\si_{x_1},\si_{x_2})$-summable, so is $(\si_{x_1},\ldots, \si_{x_n})$-summable.
			\item If $s\geq1$ and $n=1$, then $f = 1/x_1^s$. Since the isotropy group of $x_1$ in $\<\si_{x_1}>$ is $\{\bf 1\}$, by Theorem~\ref{THM:rank=1}, we get that $f$ is not $ (\si_{x_1})$-summable.
			\item If $s>1$ and $n=2$, then $f= 1/(x_1^s + x_2^s) = \sum_{j = 1}^s a_j/(x_1-\beta_jx_2)$, where $\beta_j$'s are distinct roots of $z^s = -1$ and $a_j = {1}/{(s (\beta_jx_2)^{s - 1})}$. There exists $j\in\{1, \ldots, s\}$ such that $\beta_j\notin \Z$. Then for $d_j= x_1 - \beta_jx_2$, we have $G_{d_j}=\{\bf 1\}$. So $a_j/d_j$ is not $(\si_{x_1},\si_{x_2})$-summable in $\bC(x_1, x_2)$ and by Lemma~\ref{LEM:red1}, neither is $f$. Hence $f$ is not $(\si_{x_1},\si_{x_2})$-summable in $\bQ(x_1, x_2)$. This result has appeared in~\cite[Example 3.8]{ChenSinger2014}.
			\item If $s>1$ and $n>2$, then $d$ is irreducible. Since $G_d=\{\bf 1\}$, by Theorem~\ref{THM:rank=1}, it follows that $f$ is not $(\si_{x_1},\ldots,\si_{x_n})$-summable.
		\end{enumerate}
	\end{example}
	Now we transfer the $(\tau_1, \ldots, \tau_r)$-summability problem to the $(\si_{x_1}, \ldots, \si_{x_r})$-summability problem.
	\begin{prop}\label{PROP:transformaiton}
		Let $\{\tau_i\}_{i=1}^{r}(1\leq r\leq n)$ be a family of linearly independent elements in $G=\<\si_{x_1},\ldots,\si_{x_n}>$. Then there exists an $\bF$-automorphism $\phi$ of $\bF(\vx)$ such that $\phi$ is a difference isomorphism between the difference fields $(\bF(\vx),\tau_i)$ and $(\bF(\vx),\si_{x_i})$ for all $i=1,\ldots,r$. Therefore, for any $f\in\bF(\vx)$, $f$ is $(\tau_1,\ldots,\tau_r)$-summable in $\bF(\vx)$ if and only if $\phi(f)$ is $(\si_{x_1},\ldots,\si_{x_r})$-summable in $\bF(\vx)$.
	\end{prop}
	\begin{proof}
		Assume $\tau_i=\si_{x_1}^{\a_{i,1}}\cdots\si_{x_m}^{\a_{i,m}}$ with $\a_{i, j} = 0$ if $j>n$ and write $\val_i=(\a_{i,1},\ldots,\a_{i,m})\in \Z^m$ viewed as a vector in $\Q^m$ for $i=1,\ldots,r$.
		Then $\val_1,\ldots,\val_r$ are linearly independent over $\Q$. So we can find the other vectors $\val_{r+1},\ldots,\val_m$ such that $\{\val_1,\ldots,\val_m\}$ forms a basis of $\Q^m$. Let $\val_i=(\a_{i,1},\ldots,\a_{i,m})$ for $i=r+1,\ldots,m$ and $A={(\a_{i,j})}\in\Q^{m\times m}$. Then $A$ is an invertible matrix. Thus we define an $\F$-automorphism
		$\phi:\bF(\vx) \to\bF(\vx)$ by
		\[(\phi(x_1),\ldots,\phi(x_m)):=(x_1,\ldots,x_m)A.\]
		Let $u_j:=\phi(x_j)=\sum_{i=1}^m \a_{i,j}x_i$ for all $1\leq j\leq m$. Then $\phi$ satisfies the relation $\phi\circ\tau_i=\si_{x_i}\circ\phi$ for all $i=1,\ldots,r$, which means the following diagrams
		\begin{center}
			\vspace{-0.48cm}
			\begin{minipage}[t]{0.38\textwidth}
				\[\xymatrix{
					\bF(\vx) \ar[d]_{\tau_1} \ar[r]^{\phi}
					&\bF(\vx) \ar[d]^{\si_1}
					\\
					\bF(\vx)
					\ar[r]^{\phi}
					&\bF(\vx)}\]
			\end{minipage}
			\begin{minipage}[t]{0.1\textwidth}
				\centering
				\vspace{1.15cm}
				$\cdots$
			\end{minipage}
			\centering
			\begin{minipage}[t]{0.38\textwidth}
				\[\xymatrix{
					\bF(\vx) \ar[d]_{\tau_r} \ar[r]^{\phi}
					&\bF(\vx) \ar[d]^{\si_r}
					\\
					\bF(\vx)
					\ar[r]^{\phi}
					&\bF(\vx)}\]
			\end{minipage}
		\end{center}
		are commutative.
		This is true since for any $f\in\bF(x_1,\ldots,x_m)$, we have
		\begin{align*}
			\phi\left(\tau_i (f(x_1,\ldots,x_m)\right) &= \phi\left(f(x_1+\a_{i,1},\ldots,x_m+\a_{i,m})\right) \\
			& =f(u_1+\a_{i,1},\ldots,u_m+\a_{i,m})
		\end{align*}
		and
		\begin{align*}
			\si_{x_i}\left(\phi(f(x_1,\ldots,x_m)\right) &=\si_{x_i}\left(f\left(u_1,\ldots,u_m\right)\right) \\
			& =f\left(u_1+\a_{i,1},\ldots,u_m+\a_{i,m}\right).\quad\quad
		\end{align*}
		It follows that
		\begin{equation*}
			f=\sum_{i=1}^{r}\Delta_{\tau_i}(g_i)\quad\text{if and only if}\quad \phi(f)=\sum_{i=1}^r\Delta_{{x_i}}(\phi(g_i))
		\end{equation*}
		whenever $f, g_1,\ldots,g_r\in\bF(\vx)$. This proves our assertion.
	\end{proof}
	Combining Theorem~\ref{THM:rank=1} and Proposition~\ref{PROP:transformaiton}, the summability problem~\ref{PROB:summability problem} in $n$ variables can be
	reduced to that in fewer variables. So we can design the following recursive algorithm for testing $(\si_{x_1}, \ldots, \si_{x_n})$-summability of multivariate rational functions. Furthermore, the $(\tau_1,\ldots,\tau_n)$-summability problem can also be solved via the transformation in Proposition~\ref{PROP:transformaiton}.
	
	Recall that a map $\phi:\bF(\vx) \to \bF(\vx)$ is called a $\bQ$-\emph{affine} map if $\phi(f(\vx)) = f(\vx\cdot A + \vb)$, where $A$ is an invertible matrix in $\text{GL}_m(\bQ)$ and $\vb$ is a vector in $\bQ^m$. Note that the identity map, all shift operations and all difference isomorphisms constructed in Proposition~\ref{PROP:transformaiton} are $\bQ$-affine maps. The composition of two $\bQ$-affine maps is still a $\bQ$-affine map. If $f$ is $(\si_{x_1},\ldots,\si_{x_n})$-summable in $\bF(\vx)$, Algorithm~\ref{Alg:summable} will output \textit{unnormalised certificates} for $f$ in the form
	\begin{equation*}
	g=\sum_{\ell=1}^\rho\prod_{k=1}^{K_\ell}\psi_{\ell,k}(u_{\ell,k}),
	\end{equation*}
	where $u_{\ell,k}\in\bF(\vx)$ and the $\psi_{\ell,k}$'s are $\bQ$-affine maps.

	\begin{algorithm}[Constructive Testing of the Rational Summability]\label{Alg:summable}\quad\\{\em\bf IsSummable($f$, $[x_1, \ldots, x_n]$)}.\\
		INPUT: a multivariate rational function $f\in \bF(\vx)$ and a list $[x_1, \ldots, x_n]$ of variable names;\\
		OUTPUT: unnormalised certificates $g_1, \ldots, g_n$ for $f$ if $f$ is $(\si_{x_1},\ldots,\si_{x_n})$-summable in $\bF(\vx)$; {\em false} otherwise.
		
		\step 12 using shift equivalence testing and partial fraction decomposition, decompose $f$ into $f = f_0 + \sum_{j\in\N^+ }\sum_{[d]_G}f_{[d]_G,j}$ as in Equation~\eqref{EQ:f add decomp G}.
		\step 22 apply the reduction to $f_0$ and each nonzero component $f_{[d]_G,j}$ such that
		\begin{equation*}%\label{EQ:red}
			f = \Delta_{x_1} (g_1) + \cdots + \Delta_{x_n} (g_n) + r \text{ with } r = \sum_{i = 1}^I \sum_{j = 1}^{J_i} \frac{a_{i, j}}{d_i^j},
		\end{equation*}
		where $a_{i, j}/d_i^j$ is the remainder of $f_{[d_i]_G, j}$ described in Lemma~\ref{LEM:red2}.
		\step 32 if $r = 0$, then {\em \bf return} $g_1, \ldots, g_n$.
		\step 42 {{\em\bf for}} $i = 1,\ldots,I$ {{\em\bf do}}
		\step 53 by Remark~\ref{REM:isotropy}, one can compute a basis $\tau_{i, 1}, \ldots, \tau_{i, r_i}$ for the isotropy group $G_{d_i}$ of $d_i$.
		\step 63 {{\em\bf for}} $j = 1, \ldots, J_i$ {{\em\bf do}}
		\step 74 if $n=1$ or $G_{d_i} = \{\bf 1\}$ then
		\step 85 {{\em\bf return}} {\em false} if $a_{i,j} \neq 0$.
		\step 94 else
		\step {10}5 find an $\bF$-automorphism $\phi_i$ of $\bF(\vx)$ given in Proposition~\ref{PROP:transformaiton} such that $\phi_i\circ\tau_{i, \ell} = \si_{x_\ell} \circ \phi_i$ for $\ell = 1, \ldots, r_i$.
		\step {11}5 set $\tilde a_{i, j} = \phi_i(a_{i, j})$.
		\step {12}5 execute \text{\em\bf IsSummable($\tilde a_{i, j}$, $[x_1, \ldots, x_{r_i}]$)}.
		\step {13}5 if $\tilde a_{i, j}$ is $(\si_{x_1}, \ldots, \si_{x_{r_i}})$-summable in $\bF(\vx)$, let
		\begin{equation*}
			\tilde a_{i, j} = \Delta_{x_1}\left(\tilde b_{i, j}^{(1)}\right) + \cdots + \Delta_{x_{r_i}} \left(\tilde b_{i, j}^{(r_i)}\right);
		\end{equation*}
		{\em\bf return} {\em false} otherwise.
		\step {14}5 applying $\phi_i^{-1}$ to the previous equation yields that
		\begin{equation*}
			a_{i, j} = \Delta_{\tau_{i, 1}} \left(b_{i, j}^{(1)}\right) + \cdots + \Delta_{\tau_{i, r_i}} \left(b_{i, j}^{(r_i)}\right),
		\end{equation*}
		where $(b_{i, j}^{(1)}, \ldots, b_{i, j}^{(r_i)}) = (\phi_i^{-1}(\tilde b_{i, j}^{(1)}), \ldots, \phi_i^{-1}(\tilde b_{i, j}^{(r_i)} ))$.
		\step {15}5 using Lemma~\ref{Formula}, compute $h_{i, j}^{(1)}, \ldots, h_{i, j}^{(n)}\in \bF(\vx)$ such that
		\begin{equation*}
			\frac{a_{i, j}}{d_i^j} = \sum_{\ell = 1}^{r_i} \Delta_{\tau_{i, \ell}} \left(\frac{b_{i, j}^{(\ell)}}{d_i^j}\right) =\sum_{\ell = 1}^n \Delta_{x_\ell} \left(h_{i, j}^{(\ell)}\right).
		\end{equation*}
		\step {16}{5} update $g_\ell = g_\ell + h_{i, j}^{(\ell)}$ for $\ell = 1, \ldots, n$.
		\step {17}{2} {{\em\bf return}} $g_1, \ldots, g_n$.
	\end{algorithm}

We now analyse the complexity of Algorithm~\ref{Alg:summable} for $\bF=\bQ$. The following theorem shows that the rational summability problem can be solved in polynomial time.

\begin{thm}\label{THM:com-summable}
    Let $\delta$ and $M$ be two positive integers and $f(\vx)$ be a multivariate rational function in $\bQ(\vx)_{\delta}$ with $\|f\|=M$. If $m$ and $M$ are fixed, i.e., $m,M\in O(1)$, then the total runtime of Algorithm~\ref{Alg:summable} is $O(\delta^{O(1)})$ ops in $\bQ$.
\end{thm}
\begin{proof}
In each recursion, the input rational function $\tilde{a}_{i,j}$ in Step~12 is in $\bQ(\vx)_{O(\delta^{O(1)})}$ with max-norm $O((M\delta)^{\delta^{O(1)}})$ by Theorem~\ref{THM:com-SET} and the timing of the first three steps is dominated by that of the irreducible factorization of the denominator of the input, which is $O(\delta^{O(1)})$ ops in $\bQ$ according to Fact~\ref{FACT:com-factor}. Since there are at most $ m$ recursive calls and $m \in O(1)$, the total runtime is $O(\delta^{O(1)})$ ops in $\bQ$.
\end{proof}

	\begin{example}\label{Eg:f}
		Let $G = \<\si_x, \si_y, \si_z>$ and $f = f_1 + f_2 + f_3\in \bQ(x, y, z)$ be the same as in Example~\ref{Eg: f_sum}.
		\begin{enumerate}
			\item\label{it:f1} After the $(\si_x, \si_y, \si_z)$-reduction for $f_1$, see Example~\ref{Eg: f1_red}, we get
			\begin{equation}\label{EQ:f1}
				f_1 = \Delta_{x} (u_1) + \Delta_y (v_1) + \Delta_{z} (w_1) + r_1 \text{ with } r_1 = \frac{2x - 1}{d_1},
			\end{equation}
			where $u_1, v_1, w_1\in \Q(x, y, z)$ and $d_1 = x^2 + 2xy + z^2$. By Example~\ref{Eg:isotropy}~(\ref{it:isotropy1}), the isotropy group $G_{d_1} = \{ \bf 1\}$ is trivial. By Theorem~\ref{THM:rank=1}, we see that $r_1$ is not $(\si_x, \si_y, \si_z)$-summable because its numerator $a_1 = 2x - 1$ is not zero. Hence $f_1$ is not $(\si_x, \si_y, \si_z)$-summable.
			
			\item\label{it:f2} For $f_2 = a_2 / d_2$ with $a_2 = x + z$ and $d_2 = (x - 3y)^2(y + z) + 1$, we know from Example~\ref{Eg:isotropy}~(\ref{it:isotropy2}) that a basis of $G_{d_2}$ is $\{\si_x^3\si_y\si_z^{-1}\}$. For any $\{\mu,\nu\} \subseteq \{x, y, z\}$, since the isotropy group of $d_2$ in $\<\si_\mu, \si_\nu>$ is trivial, we get that $f_2$ is not $(\si_\mu, \si_\nu)$-summable in $\bQ(x, y, z)$. By Theorem~\ref{THM:rank=1}, we see that $f_2$ is $(\si_x, \si_y, \si_z)$-summable in $\bQ(x, y, z)$ if and only if $a_2$ is $ (\tau)$-summable in $\bQ(x, y, z)$ with $\tau = \si_x^3\si_y\si_z^{-1}$. Choose one $\bQ$-automorphism $\phi_2$ of $\bQ(x, y, z)$ given in Proposition~\ref{PROP:transformaiton} as follows
			\[\phi_2(h(x, y, z)) = h(3x, x + y, - x + z),\]
			for any $h\in \bQ(x, y, z)$. Then $\phi_2\circ \tau = \si_x \circ \phi_2$. Hence $a_2$ is $ (\tau)$-summable in $\bQ(x, y, z)$ if and only if $\phi_2 (a_2)$ is $ (\si_x)$-summable in $\bQ(x, y, z)$. Since
			\begin{equation}\label{EQ:phi2a2}
				\phi_2(a_2) = 2x + z = \Delta_x ((x - 1)(x + z))
			\end{equation}
			is $ (\si_x)$-summable, it follows that $f_2$ is $(\si_x, \si_y, \si_z)$-summable. In fact, applying $\phi_2^{-1}$ to Equation~\eqref{EQ:phi2a2} yields that
			\[a_2 = x + z = \Delta_{\tau} (b) \text{ with } b = \frac{1}{9}(x - 3)(2x + 3z).\]
			By Lemma~\ref{Formula}, we have
			\begin{equation}\label{EQ:f2}
				f_2 = \Delta_{\tau}\left(\frac{b}{d_2}\right) = \Delta_{x} (u_2) + \Delta_{y} (v_2) + \Delta_{z} (w_2),
			\end{equation}
			where $u_2 = \sum_{\ell =0 }^2\si_x^\ell \si_y\si_z^{-1} \left(\frac{b}{d_2}\right)$, $v_2 = \si_z^{-1}\left(\frac{b}{d_2}\right)$ and $w_2 = - \si_z^{-1}\left(\frac{b}{d_2}\right)$.
			
			\item\label{it:f3} For $f_3 = a_3/d_3^2$ with $a_3 = y + z/(y^2 + z - 1) - 1/(y^2 + z)$ and $d_3 = x + 2y + z$, we know from Example~\ref{Eg:isotropy}~(\ref{it:isotropy2}) that a basis of $G_{d_3}$ is $\{\tau_{1}, \tau_{2}\}$, where $\tau_{1} = \si_x^{2}\si_y^{-1}$ and $\tau_{2} = \si_x\si_z^{-1}$. To decide the $(\si_x, \si_y, \si_z)$-summability of $f_3$, we construct a $\bQ$-automorphism $\phi_3$ of $\bQ(x, y, z)$ such that $\phi_3\circ \tau_{1} = \si_x \circ \phi_3$ and $\phi_3\circ \tau_{2} = \si_y \circ \phi_3$ as follows
			\[\phi_3(h(x, y, z)) = h(2x + y, - x, - y + z),\]
			for any $h\in \bQ(x, y, z)$. Then it remains to decide the $(\si_x, \si_y)$-summability of
			\[\phi_3(a_3) = - x + \frac{z - y}{\underbrace{x^2 - y + z - 1}_{\si_y(\tilde d)}} - \frac{1}{\underbrace{x^2 - y + z}_{\tilde d}}\]
			in $\bQ(x, y, z)$.
			So we use the $(\si_x, \si_y)$-reduction to reduce $\phi_3(a_3)$ and obtain
			\begin{equation}\label{EQ:phi3a3}
				\phi_3(a_3)= \Delta_x \left(\tilde b_1\right) + \Delta_y \left(\tilde b_2\right) + \frac{z - y}{x^2 - y + z},
			\end{equation}
			where $\tilde b_1= -\frac{1}{2}x (x - 1)$ and $\tilde b_2= \frac{z - y + 1}{x^2 - y + z}$. Since the isotropy group of $\tilde d$ in $\<\si_x,\si_y>$ is trivial, $\phi_3(a_3)$ is not $(\si_x, \si_y)$-summable. Hence $f_3$ is not $(\si_x, \si_y, \si_z)$-summable. Even so, in this case, using the above calculation, we can further decompose $f_3$ into a summable part and a remainder. We now show how to do this. Starting from the decomposition~\eqref{EQ:phi3a3} of $\phi_3(a_3)$ with respect to the $(\si_x, \si_y)$-summability problem, we apply $\phi_3^{-1}$ to both sides of this decomposition to obtain that
			\[a_3= \Delta_{\tau_{1}} \left(b_1\right) + \Delta_{\tau_{2}}\left(b_2\right) + \frac{z}{y^2 + z},\]
			where $b_1 = \phi_3^{-1}(\tilde b_2) = -\frac{1}{2}y(y + 1)$ and $b_2 = \phi_3^{-1}(\tilde b_2) = \frac{z + 1}{y^2 + z}$. By Lemma~\ref{Formula} with $\tau = \tau_1, \tau_2$, we have
			\begin{align}\label{EQ:f3}
				f_3 =\frac{a_3}{d_3^2}&= \Delta_{\tau_{1}} \left(\frac{b_1}{d_3^2}\right) + \Delta_{\tau_{2}} \left(\frac{b_2}{d_3^2}\right) + \underbrace{\frac{z}{(y^2 + z)d_3^2}}_{r_3}\nonumber\\
				&= \Delta_{x} (u_3) + \Delta_{y} (v_3) + \Delta_{z} (w_3) + r_3,
			\end{align}
			where $u_3 = \sum_{\ell = 0}^1 \si_x^\ell \si_y^{-1} \left(\frac{b_1}{d_3^2}\right) + \si_z^{-1}\left(\frac{b_2}{d_3^2}\right)$, $v_3 = -\si_y^{-1}\left(\frac{b_1}{d_3^2}\right)$ and $w_3 = - \si_z^{-1}\left(\frac{b_2}{d_3^2}\right)$.
			\item For $f = f_1 + f_2 + f_3$, from Example~\ref{Eg: f_sum} we know that $f_1, f_2, f_3$ are in distinct $V_{[d]_G, j}$ spaces. Since $f_1$ is not $(\si_x, \si_y, \si_z)$-summable, it follows from Lemma~\ref{LEM:red1} that $f$ is not $(\si_x, \si_y, \si_z)$-summable. Moreover, combining Equations \eqref{EQ:f1}, \eqref{EQ:f2} and~\eqref{EQ:f3}, we decompose $f$ into
			\begin{equation*}
				f = \Delta_x (u) + \Delta_y (v) + \Delta_z (w) + r \text{ with } r = \frac{2x - 1}{d_1} + \frac{z}{(y^2 + z)d_3^2},
			\end{equation*}
			where $u=\sum_{i = 1}^2u_i$, $v=\sum_{i = 1}^2v_i$ and $w=\sum_{i = 1}^2w_i$ are rational functions in $\bQ(x, y, z)$.
		\end{enumerate}
	\end{example}
	As we discussed in the above example, given a rational function $f\in \bF(\vx)$, we can compute rational functions $g_1, \ldots, g_n, r\in\bF(\vx)$ such that
    \begin{equation}\label{EQ:sumadddecomp}
        f=\Delta_{x_1} (g_1) + \cdots + \Delta_{x_n} (g_n) + r
    \end{equation}
	satisfying the property that $f$ is $(\si_{x_1}, \ldots, \si_{x_n})$-summable if and only if $r = 0$. This process can be achieved by induction on $n$. However, this remainder $r$ is not unique, which depends on the choice of the difference isomorphisms $\phi_i$'s. It remains an open problem to make the remainder $r$ minimal in terms of degrees and arithmetic sizes. Moreover, better choices of the isomorphism might also lead to more efficient reductions.

		\section{The existence problem of telescopers}\label{sec:telescopers}
		Similar to the summability problem, there are mainly two steps of solving the existence problem~\ref{PROB:telescopers} of telescopers. First we use the orbital decomposition and Abramov's reduction to simplify the existence problem in Section~\ref{subsec:reduction_tele}. Then in Section~\ref{subsec:tele criterion}, we use the exponent separation introduced in~\cite{Chen2016} to further reduce the existence problem to simple fractions and use the summability criteria in Section~\ref{subsec:sum criterion} to derive the existence criteria.
		\subsection{Orbital reduction for existence of telescopers}\label{subsec:reduction_tele}
		Let $f$ be a rational function in $\bK(t,\vx)$, where $\vx = \{x_1, \ldots, x_m\}$. Let $n$ be an integer such that $1\leq n\leq m$. We consider the existence problem of telescopers of type $(\si_t;\si_{x_1},\ldots,\si_{x_n})$ for the rational function $f$ in $\bK(t,\vx)$. Let $G_t=\<\si_t,\si_{x_1},\ldots,\si_{x_n}>$ be the free abelian group generated by the shift operators $\si_t,\si_{x_1},\ldots,\si_{x_n}$. Taking $\bE=\bK(t,\hx_1)$ and $A=G_t$ in {Equality}~\eqref{EQ: subspace V},
		we get	
		\begin{equation*}%\label{EQ: subspace V G_t}
			V_{[d]_{G_t},j}=\Span_{\bE}\left\{\ \frac{a}{\tau (d)^j}\ \middle \vert a\in \bE[x_1],
			\tau\in G_t, \deg_{x_1}(a)<\deg_{x_1}(d)\right\},
		\end{equation*}
	     where $j\in \bN^+$ and $d\in \bE[\vx]$ is irreducible with $\deg_{x_1}(d)>0$.
		Then $f$ can be decomposed as
		\begin{equation}\label{EQ:f add decomp Gt}
			f = f_0 + \sum_{j}\sum_{[d]_{G_t}}f_{[d]_{G_t},j},
		\end{equation}
		where $f_0\in V_0=\bE[x_1]$ and $f_{[d]_{G_t},j}$ are in distinct $ V_{[d]_{G_t},j}$ spaces.
		It induces the following orbital decomposition of $\bK(t,\vx)$ with respect to the group $G_t$
		\begin{equation*}%\label{EQ:decomp2}
			\bK(t,\vx)=V_0\bigoplus\left(\bigoplus_{j\in\N^+}\bigoplus_{[d]_{G_t}\in T_{G_t}}V_{[d]_{G_t},j}\right)
		\end{equation*}
		as a vector space over $\bK(t,\hx_1)$. This orbital decomposition is $G_{t}$-invariant. Moreover, for any $L$ in $\bK(t)\<S_t>$, if $f\in V_{[d]_{G_t},j}$, then $L(f)\in V_{[d]_{G_t},j}$. Note that such an operator $L$ commutes with the difference operator $\Delta_{x_i}$ for $i=1,\ldots,n$. So by Remark~\ref{REM: LCLM} and the similar argument as in the proof of Lemma~\ref{LEM:red1}, we arrive at the following lemma.

		\begin{lem}\label{LEM:red-t1}
			Let $f\in \bK(t,\vx)$. Then $f$ has a telescoper of type $(\si_t;\si_{x_1},\ldots,\si_{x_n})$ if and only if $f_0$ and each $f_{[d]_{G_t},j}$ have a telescoper of the same type for all $[d]_{G_t}\in T_{G_t}$ and $j\in \N^+$.
		\end{lem}

		Since $f_0\in V_0=\bK(t,\hx_1)[x_1]$ is always $ (\si_{x_1})$-summable, it follows that $L=1$ is a telescoper for~$f_0$. For $f\in V_{[d]_{G_t},j}$, it can be written as
		
		\begin{equation}\label{EQ:red-t1}
			f = \sum_\tau\frac{a_{\tau}}{\tau(d)^j},
		\end{equation}
		where $\tau\in G_t$, $a_\tau\in \bK(t,\hx_1)[x_1]$, $d\in \bK[t,\vx]$ with $\deg_{x_1}(a_\tau)<\deg_{x_1}(d)$ and $d$
		is irreducible in $x_1$ over $\bK(t,\hx_1)$. Each $\tau\in G_t$ is in the form of $\tau=\si_t^{k_0}\si_{x_1}^{k_1}\cdots\si_{x_n}^{k_n}$ for some $k_0,k_1,\ldots,k_n\in \Z$. Using the $(\si_{x_1}, \ldots, \si_{x_n})$-reduction formula~\eqref{EQ:red formula2}, we get the following decomposition.
		\begin{lem}\label{LEM:red-t2}
			Let $f\in V_{[d]_{G_t},j}$ be in the form~\eqref{EQ:red-t1}. Then we can decompose it into the form
			\begin{equation*}%\label{EQ:red-t2}
				f=\sum_{i=1}^n\Delta_{x_i}(g_i)+r \text{ with } r=\sum_{\ell=0}^\rho\frac{a_\ell}{\si_t^\ell(\mu)^j},
			\end{equation*}
			where $\rho\in \N$, $g_i\in\bK(t, \vx)$, $a_\ell \in\bK(t,\hx_1)[x_1]$, $\mu\in\bK[t,\vx]$, $\deg_{x_1}(a_\ell)<\deg_{x_1}(d)$, $\mu$ is in the same $G_t$-orbit as $d$, and $\si_{t}^\ell(\mu)$, $\si_{t}^{\ell^\prime}(\mu)$ are not $G$-equivalent for $0\leq\ell\neq\ell^\prime\leq\rho$. Therefore $f$ has a telescoper of type $(\si_t;\si_{x_1},\ldots,\si_{x_n})$ if and only if $r$ has a telescoper of the same type.
		\end{lem}
		
		\begin{example}\label{Eg:f_red_t}
			Let $\bK = \bQ$, $G_t = \<\si_t, \si_x, \si_y, \si_z>$ and $G = \<\si_x, \si_y, \si_z>$.
			\begin{enumerate}
				\item\label{it:f1-t} Consider the rational function $f$ in $\bQ(t, x, y, z)$ of the form
				\[f = \frac{2x - 1}{d} + \frac{y}{\si_t(d)} + \frac{1}{\si_t^3\si_x\si_y\si_z(d)}\]
				where $d = x^2 + 2xy + z^2 + t$. Then $f\in V_{[d]_{G_t}, 1}$ and applying the $(\si_x, \si_y, \si_z)$-reduction formula to $f$ yields
				\begin{equation}\label{EQ:red-t2-eg1}
					f = \Delta_x (u_0) + \Delta_y (v_0) + \Delta_z (w_0) + \frac{2x - 1}{d} + \frac{y}{\si_t(d)} + \frac{1}{\si_t^3(d)},
				\end{equation}
				where
				\[u_0 = \frac{1}{\si_t^3\si_y\si_z(d)}, ~v_0 = \frac{1}{\si_t^3\si_z(d)} \text{ and } w_0 = \frac{1}{\si_t^3(d)}.\]
				Since there is no nonzero integer $s$ such that $\si_t^s(d)$ and $d$ are $G$-equivalent, the equation~\eqref{EQ:red-t2-eg1} gives a required decomposition for $f$ in Lemma~\ref{LEM:red-t2}.
				\item\label{it:f2-t} Consider the rational function $f$ in $\bQ(t, x, y, z)$ of the form
				\[f = \frac{1}{t (t + y + 2z) d} + \frac{y + z - 1}{(t + 3z) \si_t(d)} - \frac{y + z}{(t + 3z) \si_t\si_x^3\si_y^2(d)},\]
				where $d = 3y + (x + z)^2 + t$. Then $f\in V_{[d]_{G_t}, 1}$ and applying the $(\si_x, \si_y, \si_z)$-reduction formula to $f$ yields that
				\begin{equation}\label{EQ:red-t2-eg2}
					f = \Delta_x(u_0) + \Delta_y (v_0) + \Delta_z (w_0) + \frac{1}{t (t + y + 2z) d} + \frac{1}{(t + 3z) \si_t(d)},
				\end{equation}
				where \[u_0 = - \sum_{\ell = 0}^2 \frac{y + z}{(t + 3z) \si_t\si_x^\ell\si_y^2(d)}, ~v_0 = -\sum_{\ell = 0}^1\frac{y + \ell - 2 + z}{(t + 3z) \si_t\si_y^\ell(d)} \text{ and } ~w_0 = 0.\] Since the isotropy group of $d$ in $G_t$ is $G_{t, d} = \<\si_t^3\si_y^{-1}, \si_x\si_z^{-1}>$, the minimal positive integer $s$ such that $\si_t^{s}(d)$ and $d$ are $G$-equivalent is $s = 3$. So $d$ and $\si_t(d)$ are not $G$-equivalent. Thus the equation~\eqref{EQ:red-t2-eg2} gives a required decomposition for $f$ in Lemma~\ref{LEM:red-t2}.
			\end{enumerate}
		\end{example}

		\subsection{Criteria on the existence of telescopers}\label{subsec:tele criterion}
		Combining Lemmas \ref{LEM:red-t1} and~\ref{LEM:red-t2}, we reduce the existence problem~\eqref{PROB:telescopers} to that for rational functions in the form
		\begin{equation}\label{EQ:red-t3}
			f=\sum_{i=0}^I\frac{a_i}{\si_t^i(d)^j},
		\end{equation}
		where $j\in\N^+$, $a_i\in\bK(t,\hx_1)[x_1]$, $d\in\bK[t,\vx]$, $\deg_{x_1}(a_i)<\deg_{x_1}(d)$ and $d$ is irreducible such that $\si_t^i(d)$ and $\si_t^{i^\prime}(d)$ are not $G$-equivalent for $0\leq i\neq i^\prime\leq I$.
		
		Let $G_t=\<\si_t,\si_{x_1},\ldots,\si_{x_n}>$ and $G=\<\si_{x_1},\ldots,\si_{x_n}>$ be a subgroup of $G_t$.
		Let $G_d$ and $G_{t,d}$ be the isotropy groups of the polynomial $d$ in $G$ and $G_t$, respectively.
		By Lemma~\ref{LEM:Gd/Hd free}, the quotient group $G_{t,d}/G_d$ is free and of rank $0$ or $1$.
		
		In the case of $\rank (G_{t,d}/G_d)=0$, the existence problem of telescopers is equivalent to the summability problem.
		\begin{lem}\label{LEM:red-ts1}
			Let $f\in\bK(t,\vx)$ be in the form~\eqref{EQ:red-t3}. If $\rank (G_{t,d}/G_d) =0$, then $f$ has a telescoper of type $(\si_t;\si_{x_1},\ldots,\si_{x_n})$ if and only if each $a_i/\si_t^i(d)^j$ is $(\si_{x_1},\ldots,\si_{x_n})$-summable in $\bK(t, \vx)$ for $0\leq i\leq I$.
		\end{lem}
		\begin{proof}
			Suppose that each $a_i/\si_t^i(d)^j$ is $(\si_{x_1},\ldots,\si_{x_n})$-summable for $0\leq i\leq I$. By the linearity of the difference operators $\Delta_{x_i}$, we see that $L={ \bf 1}$ is a telescoper for $f$. Conversely, assume that $L=\sum_{\ell=0}^\rho e_\ell S_t^\ell$ with $e_\ell\in\bK(t)$ is a telescoper of type $(\si_t;\si_{x_1},\ldots,\si_{x_n})$ for $f$. Without loss of generality, we may suppose that $e_0\neq0$. Then we have
			\begin{equation*}
				L(f)=\sum_{\ell=0}^\rho \sum_{i=0}^Ie_{\ell}\si_t^\ell\left(\frac{a_i}{\si_t^i(d)^j}\right)=\sum_{\ell=0}^{I+\rho}\left(\frac{\sum_{i=0}^\ell e_i \si_t^i(a_{\ell-i})}{\si_t^\ell(d)^j}\right)
			\end{equation*}
			is $(\si_{x_1},\ldots,\si_{x_n})$-summable where $e_\ell=0$ if $\ell>\rho$ and $a_i=0$ if $i>I$.
			Since $\rank (G_{t,d}/G_d)=0$,
			all $\si_t^\ell(d)$ with $\ell\in \Z$ are in distinct $G$-orbits. By Lemma~\ref{LEM:red1}, for any $\ell$ with $0\leq\ell\leq\rho$, there exist $g_{\ell,1},\ldots,g_{\ell,n}\in\bK(t,\vx)$ such that
			\begin{equation}\label{EQ:ind}
				\frac{\sum_{i=0}^\ell e_i \si_t^i(a_{\ell-i})}{\si_t^\ell(d)^j}=\Delta_{x_1}(g_{\ell,1})+\cdots+\Delta_{x_n}(g_{\ell,n}).
			\end{equation}
			
			To show that each $a_i/\si^i_{t}(d)^j$ is $(\si_{x_1},\ldots,\si_{x_n})$-summable for $0\leq i\leq I$, we proceed by induction. For $i=0$, substituting $\ell=0$ into \eqref{EQ:ind}, we get $a_0/d^j=\Delta_{x_1}(g_{0,1}/{e_0})+\cdots+\Delta_{x_n}(g_{0,n}/{e_0})$. Suppose that $a_i/\si^i_{t}(d)^j$ is $(\si_{x_1},\ldots,\si_{x_n})$-summable for $i=0,\ldots,s-1$ with $s\leq I$. Taking $\ell=s$ in Equation~\eqref{EQ:ind} yields that
			\[\frac{a_s}{\si_t^s(d)^j}=\Delta_{x_1}\left(\frac{g_{s,1}}{e_0}\right)+\cdots+\Delta_{x_n}\left(\frac{g_{s,n}}{e_0}\right)-\frac{1}{e_0}\sum_{i=1}^se_i\si_t^i\left(\frac{a_{s-i}}{\si_t^{s-i}(d)^j}\right).\]
			By the inductive hypothesis, we have ${a_{s-i}}/{\si_t^{s-i}(d)^j}$ is $(\si_{x_1},\ldots,\si_{x_n})$-summable for $1\leq i\leq s$. Note that $e_i\in\bK(t)$ is free of $\vx$. Due to the commutativity between $\si_t$ and $\si_{x_i}$ for $i=1,\ldots,n$, we get $\frac{1}{e_0}\sum_{i=1}^se_i\si_t^i\left(\frac{a_{s-i}}{\si_t^{s-i}(d)^j}\right)$ is $(\si_{x_1},\ldots,\si_{x_n})$-summable. Hence $a_s/\si^s_t(d)^j$ is also $(\si_{x_1},\ldots,\si_{x_n})$-summable.
		\end{proof}

		\begin{example}
			We continue Example~\ref{Eg:f_red_t}~(\ref{it:f1-t}) and write $f\in\bQ(t, x, y, z)$ as
			\[	f = \Delta_x (u_0) + \Delta_y (v_0) + \Delta_z (w_0) + r \text{ with } r =\frac{2x - 1}{d} + \frac{y}{\si_t(d)} + \frac{1}{\si_t^3(d)},\]
			where $u_0, v_0, w_0\in\bQ(t, x, y, z)$ and $d= x^2 + 2xy + z^2 +t $. Note that the isotropy groups $G_{t, d}$ and $G_{d}$ are trivial. The first term $(2x-1)/d$ of $r$ is not $(\si_x, \si_y, \si_z)$-summable in $\bQ(t, x, y, z)$ by the similar reason as in Example~\ref{Eg:f}~(\ref{it:f1}). Since $\rank(G_{t, d}/G_d) = 0$, we know from Lemma~\ref{LEM:red-ts1} that $r$ does not have any telescoper of type $(\si_t; \si_x, \si_y, \si_z)$ and neither does $f$.
		\end{example}
		
		\begin{lem}\label{LEM:red-ts2}
			Let $f = \sum_{i = 0}^I a_i/\si_t^i(d)^j\in\bK(t,\vx)$ be in the form~\eqref{EQ:red-t3}. If $\rank (G_{t,d}/G_d)=1$, then $f$ has a telescoper of type $(\si_t;\si_{x_1},\ldots,\si_{x_n})$ if and only if each $a_i/\si_t^i(d)^j$ has a telescoper of the same type for $0\leq i\leq I$.
		\end{lem}
		\begin{proof}
			Sufficiency follows from Remark~\ref{REM: LCLM}. The proof of necessity is a natural generalization from the trivariate case~\cite[lemma 5.3]{Chen2016} to the multivariate case. Suppose $L = \sum_{i = 0}^\ell e_i S_t^i\in \bK(t)\<S_t>$ is a telescoper for $f$. Since $\rank(G_{t,d}/G_d) = 1$, there is a minimal positive integer $k_0$ such that $\si_t^{k_0} (d) = \si_{x_1}^{k_1}\cdots\si_ {x_n}^{k_n}(d)$ for some integers $k_1,\ldots, k_n$. In the expression~\eqref{EQ:red-t3}, we require that $\si_t^i(d)$ and $\si_t^{i^\prime}(d)$ are not $G$-equivalent for any $0\leq i\neq i^\prime\leq I$. By the minimality of $k_0$, we may assume $f = \sum_{i = 0}^{k_0 - 1}a_i/\si_t^i(d)^j$. The $k_0$-exponent separation of $L$ (see~\cite[Section 4]{Chen2016}) is defined as follows
			\[L = L_0 + L_1 + \cdots + L_{k_0 - 1},\]
			where $L_i = \sum_{j = 0}^{\ell} e_{jk_0 + i} S_t^{jk_0 + i}$ and $e_i = 0$ if $i > \ell$. Since $L(f)$ is $(\si_{x_1}, \ldots, \si_{x_n})$-summable, by Lemma~\ref{LEM:red1} each orbital component of $L(f)$ is summable. So we have
			\begin{equation}\label{EQ:exp_sep_eqs}
				\left\{\begin{aligned}
					L_0\frac{a_0}{d^j}&+ L_{k_0-1}\frac{a_1}{\si_t(d)^j} + \cdots +L_{1}\frac{a_{k_0 - 1} }{\si_t^{k_0 - 1}(d)^j}&\equiv 0\\
					L_1\frac{a_0}{d^j}&+\quad~L_0\frac{a_1}{\si_t(d)^j} + \cdots + L_{2}\frac{a_{k_0-1}}{ \si_x^{k_0 - 1}(d)^j} &\equiv 0\\
					& \hspace{2.0cm}\cdots&\\
					L_{k_0 - 1}\frac{a_0}{d^j}&+ L_{k_0-2}\frac{a_1}{\si_t(d)^j}+\cdots + L_{0}\frac{a_ {k_0-1}}{\si_t^{k_0-1}(d)^j}&\equiv 0,
				\end{aligned}\right.
			\end{equation}
			where $f \equiv 0$ means that $f$ is $(\si_{x_1}, \ldots, \si_{x_n})$-summable in $\bK(t, \vx)$. Taking
			\[\mathcal{V} = \left[\frac{a_0}{d^j}, \frac{a_1}{\si_t(d)^j},\ldots,\frac{a_{k_0-1}}{\si_t ^{k_0-1}(d)^j}\right]^T,\]
			then Equation~\eqref{EQ:exp_sep_eqs} can be written as
			\begin{equation*}
				\mathcal{L}_{k_0} \cdot \mathcal{V} \equiv 0,
			\end{equation*}
			where
			\[\mathcal{L}_{k_0}= \left[\begin{matrix}
				L_0 & L_{k_0 - 1}&L_{k_0 - 2}&\cdots& L_{1}\\
				L_1 & L_0 &L_{k_0 - 1}&\cdots&L_2\\
				L_2 & L_1 &L_0 &\cdots&L_3\\
				\vdots&\vdots&\vdots&&\vdots\\
				L_{k_0 - 1}&L_{k_0 - 2}& L_{k_0 - 3}&\cdots& L_0
			\end{matrix}\right].\]
			According to~\cite[Proposition 4.3]{Chen2016}, there exist nonzero operators $T_0, \ldots, T_{k_0 - 1}\in \bK(t)\<S_t>$ and a matrix $\mathcal{M}$ over $\bK(t)\<S_t>$ such that
			\begin{equation*}
				\mathcal{M}\cdot\mathcal{L}_{k_0} = \text{diag}(T_0, \ldots, T_{k_0 - 1}).
			\end{equation*}
			For each $0 \leq i \leq k_0 - 1$, we know that $T_i$ is a telescoper of type $(\si_t; \si_{x_1}, \ldots,\si_{x_n})$ for $a_i/\si_t^i(d)^j$, because the operators in $T_i \in \bK(t)\<S_t>$ commute with the difference operators $\Delta_{x_1}, \ldots, \Delta_{x_n}$.
		\end{proof}
		Now we consider the existence problem of telescopers for simple fractions in the form
		\begin{equation}\label{EQ:red-ts}
			f=\frac{a}{d^j}
		\end{equation}
		where $j\in\N^+$, $a\in\bK(t,\hx_1)[x_1]$, $d\in\bK[t,\vx]$, $\deg_{x_1}(a)<\deg_{x_1}(d)$ and $d$ is irreducible such that $\rank (G_{t,d}/G_d)=1$.

		\begin{thm}\label{THM:telescopers}
			Let $f\in\bK(t,\vx)$ be as in~\eqref{EQ:red-ts}. Let $\{\tau_0,\tau_1,\ldots,\tau_r\}(1\leq r<n)$ be a basis of $G_{t,d}$ such that $G_{t,d}/G_d=\<\bar\tau_0>$ and $\{\tau_1,\ldots,\tau_r\}$ is a basis of $G_d$ (take $\tau_1={\bf1}$, if $G_d=\{{\bf 1}\}$). Let $\tau_0=\si_t^{k_0}\si_{x_1}^{-k_1}\cdots\si_{x_n}^{-k_n}$ for some $k_i\in\set Z$ and set $T_0=S_t^{k_0}S_{x_1}^{-k_1}\cdots S_{x_n}^{-k_n}$. Then $f$ has a telescoper of type $(\si_t;\si_{x_1},\ldots,\si_{x_n})$ if and only if there exists a nonzero operator $L\in\bK(t)\<T_0>$ such that \[L(a)=\Delta_{\tau_1}(b_1)+\cdots\Delta_{\tau_r}(b_r)\] for some $b_i\in\bK(t,\hx_1)[x_1]$ with $\deg_{x_1}(b_i)<\deg_{x_1}(d)$ for $1\leq i\leq r$.
		\end{thm}
		\begin{proof}
			Firstly, suppose that $L_0=\sum_{\ell=0}^\rho e_\ell T_0^\ell\in\bK(t)\<T_0>$ is a nonzero operator such that $L_0(a)=\sum_{i=1}^r\Delta_{\tau_i}(b_i)$ for some $b_i\in\bK(t,\hx_1)[x_1]$ with $\deg_{x_1}(b_i)<\deg_{x_1}(d)$. Set $L=\sum_{\ell=0}^\rho e_{\ell}S_t^{\ell k_0}$. Then
			\begin{align}
				L(f)&=\sum_{\ell=0}^\rho\frac{e_\ell \si_t^{\ell k_0}(a)}{\si_t^{\ell k_0}(d)^j}=\sum_{\ell=0}^\rho\frac{e_\ell \si_t^{\ell k_0}(a)}{\si_{x_1}^{\ell k_1}\cdots\si_{x_n}^{\ell k_n}(d)^j}\nonumber\\
				&=\sum_{i=1}^n\Delta_{x_i}(g_i)+\frac{\sum_{\ell=0}^\rho e_{\ell}\si_{t}^{\ell k_0}\si_{x_1}^{-\ell k_1}\cdots\si_{x_n}^{-\ell k_n}(a)}{d^j}\quad\text{for some }g_i\in\bK(t,\vx)\nonumber\\
				\label{EQ:trans1_sep}&=\sum_{i=1}^n\Delta_{x_i}(g_i)+\frac{L_0(a)}{d^j}\\
				&=\sum_{i=1}^n\Delta_{x_i}(g_i)+\frac{1}{d^j}\sum_{i=1}^r\left(\tau_i(b_i)-b_i)\right)\nonumber\\
				&=\sum_{i=1}^n\Delta_{x_i}(g_i)+\sum_{i=1}^r\left(\tau_i\left(\frac{b_i}{d^j}\right)-\frac{b_i}{d^j}\right)\nonumber\\
				\label{EQ:trans2_tel}&=\sum_{i=1}^n\Delta_{x_i}(g_i+h_i)\quad\text{for some } h_i\in\bK(t,\vx).
			\end{align}
			The last equality follows from Lemma~\ref{Formula}.
			
			Conversely, let $L$ be a telescoper of type $(\si_t;\si_{x_1},\ldots,\si_{x_n})
			$ for $f$. By the $k_0$-exponent separation (see~\cite[Section 4]{Chen2016}) of $L$ and Lemma~\ref{LEM:red1}, without loss of generality, we may assume $L=\sum_{\ell=0}^\rho e_\ell S_t^{\ell k_0}\in\bK(t)\<S_t>$ is a telescoper for $f$. Then
			\begin{equation*}
				L\left(\frac{a}{d^j}\right)=\sum_{\ell=0}^\rho\frac{e_\ell\si_{t}^{\ell k_0}(a)}{\si_{x_1}^{\ell k_1}\cdots\si_{x_n}^{\ell k_n}(d)^j}=\sum_{i=1}^n\Delta_{x_i}(h_i)+\frac{1}{d^j}h
			\end{equation*}
			for some $h_1,\ldots,h_n,h\in\bK(t,\vx)$ with
			\begin{equation}\label{EQ:t1}
				h=\sum_{\ell=0}^\rho e_{\ell}\si_t^{\ell k_0}\si_{x_1}^{-\ell k_1}\cdots\si_{x_n}^{-\ell k_n}(a)=\sum_{\ell=0}^{\rho}e_{\ell}\tau_0^{\ell}(a).
			\end{equation}
			Since $L(a/d^j)$ is $(\si_{x_1},\ldots,\si_{x_n})$-summable and $\{\tau_1,\ldots,\tau_r\}$ is a basis of $G_d$, by Theorem~\ref{THM:rank=1} with $\bF=\bK(t)$ we get
			\begin{equation}\label{EQ:t2}
				h=\Delta_{\tau_1}(b_1)+\cdots+\Delta_{\tau_r}(b_r)
			\end{equation}
			for some $b_i\in\bK(t,\hx_1)[x_1]$ with $\deg_{x_1}(b_i)<\deg_{x_1}(d)$ for $1\leq i\leq r$. Combining Equations~\eqref{EQ:t1} and~\eqref{EQ:t2} yields that $a$ has a telescoper $L_0=\sum_{\ell=0}^\rho e_\ell T_{0}^{\ell}$ of type $(\tau_0;\tau_1,\ldots,\tau_r)$.
		\end{proof}
		\begin{prop}\label{PROP:separation}
			Let $\tau\in G_t\setminus G$ and $f=a/b$ with $a,b\in\bK[t,\vx]$ and $\gcd(a,b)=1$. Then there exist $e_0,\ldots,e_r\in\bK(t)$, not all zero, such that $\sum_{i=0}^re_i\tau^i(f)=0$ if and only if $b=b_1b_2$ with $b_1\in\bK[t]$ and $b_2\in\bK[t,\vx]$ satisfying that $\tau(b_2)=b_2$.
		\end{prop}
		\begin{proof}
			First we suppose $b=b_1b_2$ with $b_1,b_2$ satisfying the above conditions. Then for any $i\in\N$,
			\begin{equation}
				\tau^i(f)=\frac{\tau^i(a)}{\tau^i(b_1b_2)}=\frac{\tau^i(a)}{\tau^i(b_1)b_2}=\frac{\tau^i(a/b_1)}{b_2}.
			\end{equation}
			Note that $b_1\in\bK[t]$ and the total degrees of the polynomials $\tau^i(a)$ in $\vx$ are the same as that of $a$. Thus all shifts of $a/b_1$ lie in a finite dimensional linear space over $\bK(t)$. So there exist $e_0,e_1,\ldots,e_r\in\bK(t)$, not all zero, such that $\sum_{i=0}^re_i\tau^i(a/b_1)=0$. This implies $\sum_{i=0}^r e_i \tau^i(f)=0$.
			
			Conversely, suppose $\sum_{i=0}^re_i\tau^i(f)=0$. Let $b_1$ and $b_2$ be the content and primitive part of $b$ as a polynomial in $\vx$ over $\bK(t)$. If $b_2\in\bK$, then we are done. Now we assume that $b_2\notin\bK$. Then all of its irreducible factors have positive total degrees in $\vx$. Assume that there exists an irreducible polynomial $p$ such that $\tau(p)\neq p$. By Lemma~\ref{LEM:G/Gp free}, the quotient group $G_t/G_{t,p}$ is free, so is torsion free. So for any integer $i\neq 0$, $\tau^{i}(p)\neq p$. Among all of such irreducible factors of $b_2$, we can find one factor $p$ such that $\tau^{i}(p)\nmid b_2$ for any integer $i<0$. Let $s$ be the largest integer such that $\tau^{s}(p)\mid b_2$. Then the irreducible polynomial $\tau^{r+s}(p)$ divides $\tau^r(b_2)$, but $\tau^{r+s}{ (p)}\nmid\tau^i(b_2)$ for any $0\leq i\leq r-1$. Otherwise $\tau^{r+s-i}(p)\mid b_2$, which contradicts the choice of $s$. Therefore we have $\sum_{i=0}^re_i\tau^i(f)\neq 0$, since $p$ depends on $\vx$ and the coefficients $e_i$ are in $\bK(t)$. This leads to a contradiction. So every irreducible factor $p$ of $b_2$ satisfies the property that $\tau(p)=p$. This implies that $\tau(b_2)=b_2$.
		\end{proof}

              \begin{lem}\label{LEM:annihilator}
            Let $\tau\in G_t\setminus G$ and $f=a/(b_1b_2)$ with $b_1\in\bK[t]$, $a,b_2\in\bK[t,\vx]$ and $\tau(b_2)=b_2$.
               Then we can compute $e_0,\ldots,e_r\in\bK[t]$, not all zero, such that $\sum_{i=0}^re_i\tau^i(f)=0$.
            \end{lem}
            \begin{proof}
             By Proposition~\ref{PROP:telescopers} below, we can construct a difference isomorphism between $(\bK(t,\vx), \tau)$ and $(\bK(t,\vx), \si_t)$ such that $\varphi\circ\tau = \si_{t} \circ \varphi$ and $\varphi(\bK[t])\subseteq \bK[t]$. Then $\sigma_t(\varphi(b_2)) =\varphi(\tau(b_2)) = \varphi(b_2)$ and for all $e_i(t)\in \bK[t]$,
             \[\sum_{i=0}^r e_i(t)\tau^i(f)=0\quad\text{if and only if}\quad\sum_{i=0}^r e_i(\varphi(t))\si_t^i(\varphi(f))=0.\] So we only need to consider the case $\tau = \sigma_t$. Now suppose that $f=a/(b_1b_2)$ with $b_1\in\bK[t]$, $a,b_2\in\bK[t,\vx]$ and $\sigma_t(b_2)=b_2$. It suffices to find a nonzero operator $L\in \bK[t]\<S_t>$ such that $L(f)=0$. We write $a=\sum_{i=0}^sa_it^i$ with $a_i\in \bK[\vx]$. For each $0\leq i\leq s$, let $L_i=t^i\sigma_t(b_1)S_t - (t+1)^ib_1$. Then $L_i$ is an operator in $\bK[t]\<S_t>$ such that $L_i(t^i/b_1) = 0$.  Since $\sigma_t(b_2) = b_2$ and $\sigma_t(a_i)=a_i$, we have $L_i(a_it^i/(b_1b_2)) = (a_i/b_2)L_i(t^i/b_1) = 0$. Let $L\in \bK[t]\<S_t>$ be the LCLM of $L_i$ for all $i=0,\ldots,s$ and write $L = R_iL_i$ with $R_i\in\bK[t]\<S_t>$. Then $L$ is a nonzero operator and
            \[L(f) = L\left(\sum_{i=0}^s\frac{a_it^i}{b_1b_2}\right) = \sum_{i=0}^sL\left(\frac{a_it^i}{b_1b_2}\right)=\sum_{i=0}^sR_iL_i\left(\frac{a_i t^i}{b_1b_2}\right)=0.\]
            This completes the proof.
            \end{proof}

		\begin{rem}\label{REM:tele-basis}
	For the bivariate case with $m=n=1$, our existence criterion coincides with the known result in \cite[Theorem 1]{AbramovLe2002} and~\cite[Theorem 4.11]{ChenSinger2012}. Let $f=a/d^j\in\bK(t,\vx)$, where $j\in\N^+$, $a\in\bK(t,\hx_1)[x_1]$, $d\in\bK[t,\vx]$, $\deg_{x_1}(a)<\deg_{x_1}(d)$ and $d$ is irreducible. Let $G_t=\<\si_t;\si_{x_1}>$ and $G=\<\si_{x_1}>$. Since the degree of $d$ in $x_1$ is positive, we have $G_d=\{{\bf1}\}$. If $\rank(G_{t,d}/G_d)=0$, then by Lemma~\ref{LEM:red-ts1} and Theorem~\ref{THM:rank=1}, $f$ has a telescoper of type $(\si_t,\si_{x_1})$ if and only if $a=0$. If $\rank(G_{t,d}/G_d)=1$, then there exists $\tau=\si_t^s\si_{x_1}^k\in G_{t,d}$ with $s>0$ such that $G_{t,d}/G_d=\<\bar\tau>$.  By Theorem~\ref{THM:telescopers} and Proposition~\ref{PROP:separation}, $f$ has a telescoper of type $(\si_t;\si_{x_1})$ if and only if $a=c/b$ with $c\in\bK[t,\vx]$, $b\in\bK[t,\hx_1]$, $\gcd(b,c)=1$, where $b$ can be written as $b=b_1b_2$ with $b_1\in\bK[t]$ and $b_2\in\bK[t,\hx_1]$ such that $\tau(b_2)=b_2$. Since $m=n=1$, we have $b\in\bK[t]$ and hence $f$ always has a telescoper of type $(\si_t;\si_{x_1})$.
		\end{rem}

		\begin{example}\label{ex:shift}
			Let $f = 1/(t^s + x_1^s + \cdots+x_n^s)\in \Q(t, x_1,\ldots,x_n)$ with $s,n \in \N \setminus\{0\}$. Then $d=t^s + x_1^s + \cdots+x_n^s$ is irreducible over $\Q$ if $n >1$. Let $G_{t,d}$ and $G_d$ be the isotropy groups of $d$ in $G_t=\la\si_t, \si_{x_1},\ldots,\si_{x_n}\ra$ and $G=\la \si_{x_1},\ldots,\si_{x_n}\ra$, respectively. Then we can decide the existence of telescopers of type $(\si_t; \si_{x_1},\ldots, \si_{x_n})$ for all cases of $f$.
			\begin{enumerate}
				\item If $s=1$, then $d$ is irreducible. Since $G_{t, d} =\<\tau>$ with $\tau = \si_t\si_{x_1}^{-1}$ and $G_d =\{\bf 1\}$, we have $G_{t,d}/G_d =\<\bar\tau>$ and $\rank(G_{t,d}/G_d) = 1$. Observing that $\tau - 1$ is an annihilator of the numerator of $f$, by Theorem~\ref{THM:telescopers} we get $f$ has a telescoper. In fact, $L(f) = \Delta_{x_1}(f) + \Delta_{x_2}(0) +\cdots +\Delta_{x_n} (0)$, where $L = S_t - 1$.
				\item If $s>1$ and $n=1$, then $f = 1/(t^s+x_1^s) = \sum_{j = 1}^s a_j/(t-\beta_jx_1) $, where $\beta_j$'s are distinct roots of $z^s = -1$ and $a_j = {1}/{(s (\beta_jx_2)^{s - 1})}$. There exists $j\in\{1, \ldots, s\}$ such that $\beta^j \notin \Z$. Then for $d_j = t - \beta_jx_1$, we have $G_{t, d_j} = G_{d_j} = \{\bf 1\}$. So $a_j/d_j$ is not $\si_{x_1}$-summable in $\bC(t, x_1)$ and neither is $f$. By Lemma~\ref{LEM:red-ts1}, we get that $f$ does not have any telescoper of type $(\si_t, \si_{x_1})$ in $\bC(t)\<S_t>$. Hence $f$ does not have any telescoper of the same type in $\bQ(t)\<S_t>$.
				\item If $s>1$ and $n>1$, then $d$ is irreducible. Since $G_d=\{\bf 1\}$, $f$ is not $(\si_{x_1},\ldots,\si_{x_n})$-summable by Theorem~\ref{THM:rank=1}. Since $G_{t,d} = \{\bf 1\}$ and $\rank(G_{t,d}/G_d) = 0$, by Lemma~\ref{LEM:red-ts1}, we conclude that $f$ does not have any telescoper.
			\end{enumerate}
		\end{example}
		
		\begin{prop}\label{PROP:telescopers}
			Let $\{\tau_0,\tau_1,\ldots,\tau_r\}(1\leq r\leq n)$ be a family of $\Z$-linearly independent elements in $G_t$ such that $\tau_0\in G_t\setminus G$ and $\{\tau_1,\ldots,\tau_r\}\subseteq G$. Then there exists a $\bK$-automorphism $\varphi$ of $\bK(t,\vx)$ such that $\varphi$ is a difference isomorphism between the difference fields $(\bK(t,\vx),\tau_0)$ and $(\bK(t,\vx),\si_t)$, and simultaneously a difference isomorphism between $(\bK(t,\vx),\tau_i)$ and $(\bK(t,\vx),\si_{x_i})$ for all $i=1,\ldots,r$. Furthermore, $\varphi(\bK(t))\subseteq \bK(t)$ and hence for any $f\in\bK(t,\vx)$, $f$ has a telescoper of type $(\tau_0;\tau_1,\ldots,\tau_r)$ if and only if $\varphi(f)$ has a telescoper of type $(\si_t;\si_{x_1},\ldots,\si_{x_r})$.
		\end{prop}
		\begin{proof}
			Let
			$\tau_{i}=\si_t^{\a_{i,0}}\si_{x_1}^{\a_{i,1}}\cdots\si_{x_m}^{\a_{i,m}}$, where $\a_{i, j} = 0$ if $j>n$. Define $\val_i=(\a_{i,0},\a_{i,1},\ldots,\a_{i,m})\in\Z^{m+1}$ for $i=0,1,\ldots,r$. Since $\val_0,\val_1,\ldots,\val_r$ are linearly independent over $\Q$, we can find vectors $\val_{r+1},\ldots,\val_m\in\Q^{m+1}$ such that $\{\val_0,\val_1,\ldots,\val_m\}$ is a basis of $\Q^{m+1}$ over $\Q$. Write $\val_i=(\a_{i,0},\a_{i,1},\ldots ,\a_{i,m})$ for $i=r+1,\ldots,m$. Since $\tau_0\in G_t\setminus G$ and $\{\tau_1,\ldots,\tau_r\}\subseteq G$, we have $\a_{0,0}\neq0$ and $\a_{i,0}=0$ for $i=1,\ldots,r$. So we can further assume that $\a_{i,0}=0$ for $i=r+1,\ldots,m$. Let $A=(\a_{i,j})\in\Q^{(m+1)\times(m+1)}$ which is invertible. Let $\varphi$ be a $\bK$-automorphism of $\bK(t,\vx)$ defined by
			\[(\varphi(t),\varphi(x_1),\ldots,\varphi(x_m)):=(t,x_1,\ldots,x_m)A.\]
			Then $\varphi(t)=\a_{0,0}\cdot t$ and $\varphi(x_j)=\a_{0,j}\cdot t+\sum_{i=1}^m\a_{i,j}\cdot x_i$ for $j=1,\ldots,m$. It can be checked that
			$\varphi\circ\tau_0=\si_t\circ\varphi\text{ and } \varphi\circ\tau_i=\si_{x_i}\circ\varphi\text{ for } 1\leq i\leq r$. This means the following diagrams are commutative.
			\begin{center}
				\centering
				\vspace{-0.40cm}
				\begin{minipage}[t]{0.38\textwidth}
					$$\xymatrix{
						\bK(t,\vx) \ar[d]_{\tau_0} \ar[r]^{\varphi}
						&\bK(t,\vx) \ar[d]^{\si_t}
						\\
						\bK(t,\vx)
						\ar[r]^{\varphi}
						&\bK(t,\vx)}$$
				\end{minipage}
				\begin{minipage}[t]{0.1\textwidth}
					\centering
					\vspace{1cm}
					$\cdots$
				\end{minipage}
				\begin{minipage}[t]{0.38\textwidth}
					$$\xymatrix{
						\bK(t,\vx) \ar[d]_{\tau_i} \ar[r]^{\varphi}
						&\bK(t,\vx) \ar[d]^{\si_{x_i}}
						\\
						\bK(t,\vx)
						\ar[r]^{\varphi}
						&\bK(t,\vx)}$$
				\end{minipage}
			\end{center}
			Note that $\varphi(\bK(t))\subseteq\bK(t)$. It follows that
			\begin{equation*}
				\sum_{\ell=0}^\rho{e_\ell(t)\tau_{0}^\ell(f)=\sum_{i=1}^r\Delta_{\tau_i}(g_i)}\quad\text{if and only if}\quad\sum_{\ell=0}^\rho e_\ell(\a_{0,0}t)\si_t^\ell(\varphi(f))=\sum_{i=1}^r\Delta_{{x_i}}(\varphi(g_i)),
			\end{equation*}
			whenever $e_\ell(t)\in\bK(t)$ and $f,g_i\in\bK(t,\vx)$. This completes our proof.
		\end{proof}
        The way of constructing a difference isomorphism in Proposition~\ref{PROP:telescopers} is almost the same as that in Proposition~\ref{PROP:transformaiton}. The only difference is that in Proposition~\ref{PROP:telescopers}, we require $\varphi(\bK(t))\subseteq \bK(t)$.

		Let $f=a/d^j$ be in the form~\eqref{EQ:red-ts} with $\rank (G_{t,d}/G_d)=1$. By Theorem~\ref{THM:telescopers}, there are two cases according to whether $G_d$ is trivial or not. If $G_d=\{{\bf1}\}$, then $a/d^j$ has a telescoper of type $(\si_t;\si_{x_1},\ldots,\si_{x_n})$ if and only if there exists a nonzero operator $L\in\bK(t)\<T_0>$ such that $L(a)=0$. This problem is solved by Proposition~\ref{PROP:separation}. If $G_d$ is nontrivial, we can apply the transformation in Proposition~\ref{PROP:telescopers} to reduce the existence problem of telescopers to that of fewer variables. Moreover, the general existence of telescopers of type $(\tau_0;\tau_1,\ldots,\tau_n)$ for rational functions has also been solved.
		
		Similar to the treatment in Algorithm~\ref{Alg:summable}, if $f$ has a telescoper of type $(\si_t;\si_{x_1},\ldots,\si_{x_n})$, Algorithm~\ref{alg:telescoper} will output \textit{unnormalised certificates} for $f$ in the form
		\begin{equation*}
		g=\sum_{\ell=1}^\rho\prod_{k=1}^{K_\ell}\psi_{\ell,k}(v_{\ell,k}),
		\end{equation*}
		where $v_{\ell,k}\in\bF(t,\vx)$ and the $\psi_{\ell,k}$'s are $\bQ$-affine maps.

		\begin{algorithm}[Constructive Testing of the Existence of Telescopers]\label{alg:telescoper}
            \quad\\{\em\bf IsTelescoperable($f$, $[x_1,\ldots, x_n]$, $t$)}.\\
			INPUT: a multivariate rational function $f\in \bK(t, \vx)$, a set $\{x_1, \ldots, x_n\}$ of variable names and a variable name $t$ for telescoping;\\
			OUTPUT: a telescoper $L$ and its unnormalised certificates $g_1, \ldots, g_n$ if $f$ has a telescoper of type $(\si_t;\si_{x_1},\ldots,\si_{x_n})$; {\em false} otherwise.
			
			\step 12 using shift equivalence testing and irreducible partial fraction decomposition, decompose $f$ into $f = f_0 + \sum_{j\in\N^+ }\sum_{[d]_{G_t}}f_{[d]_{G_t},j}$ as in Equation~\eqref{EQ:f add decomp Gt}.
			\step 22 apply the reduction to $f_0$ and each nonzero component $f_{[d]_{G_t},j}$ such that
			\begin{equation*}
				f = \Delta_{x_1} (g_1) + \cdots + \Delta_{x_n} (g_n) + r \text{ with } r = \sum_{i = 1}^I \sum_{j = 1}^{J_i} \sum_{\ell = 0}^{s_{i, j}}\frac{a_{i, j, \ell}}{\si_t^{\ell}(d_i)^j},
			\end{equation*}
			where $\sum_{\ell = 0}^{s_{i, j}}\frac{a_{i, j, \ell}}{\si_t^{\ell}(d_i)^j}$ is the remainder of $f_{[d_i]_{G_t}, j}$ described in Lemma~\ref{LEM:red-t2}.
			\step 32 if $r = 0$, then {{\em\bf return}} $L = { \bf 1}$ and $g_1, \ldots, g_n$.
			\step 42 {{\em\bf for}} $i = 1,\ldots,I$ {{\em\bf do}}
			\step 53 using Remark~\ref{REM:isotropy_normalization}, one can find elements $\tau_{i, 0}, \tau_{i, 1}, \ldots, \tau_{i, r_i}\in G_{t,d_i}$ such that $G_{t,d_i}/G_{d_i} = \<\bar \tau_{i, 0}>$ and	
			\{$\tau_{i, 1}, \ldots, \tau_{i, r_i}\}$ forms a basis for $G_{d_i}$.
			\step 63 {{\em\bf for}} $j = 1, \ldots, J_i$, $0 = 1, \ldots, s_{i, j}$ {{\em\bf do}}
			\step 74 if $\rank (G_{t,d_i}/G_{d_i}) = 0$ then
			\step 85 execute Algorithm~\ref{Alg:summable} with \text{\em\bf IsSummable($r_{i, j, \ell}$, $[x_1,\ldots,x_n]$)}, where $r_{i, j, \ell} := \frac{a_{i, j, \ell}}{\si_t^{\ell}(d_i)^j}$.
			\step 95 if $r_{i, j, \ell}$ is $(\si_{x_1}, \ldots, \si_{x_n})$-summable in $\bF(\vx)$, let
			\begin{equation*}
				r_{i, j, \ell} = \Delta_{x_1} \left(h_{i, j, \ell}^{(1)}\right) + \cdots + \Delta_{x_n} \left(h_{i, j, \ell}^{(n)}\right)
			\end{equation*}
			and set $L_{i, j, \ell} = { \bf 1}$; {\em\bf return} {\em false} otherwise.
			\step {10}4 if $\rank (G_{t,d_i}/G_{d_i}) = 1$ then
			\step {11}5 choose $\tau_{i, 0} = \si_{t}^{k_{i, 0}} \si_{x_1}^{-k_{i, 1}}\cdots \si_{x_n}^{-k_{i, n}}$ with $k_{i, 0} > 0$.
			\step {12}5 set $T_{i,0} = S_t^{k_{i, 0}}S_{x_1}^{-k_{i, 1}}\cdots S_{x_n}^{k_{i, n}}$.
			\step {13}5 if $n=1$ or $G_{d_i} = \{\bf 1\}$ then
			\step {14}6 using Proposition~\ref{PROP:separation}, check whether there exists a nonzero operator $~\bar L_{i, j, \ell} (t, T_{i, 0}) \in \bK(t)\<T_{i, 0}>$ such that $\bar L_{i, j, \ell} (t, T_{i, 0}) (a_{i, j, \ell}) = 0$. If so, use Lemma~\ref{LEM:annihilator} to find such an operator $\bar L_{i, j, \ell} (t, T_{i, 0})$ and set $L_{i, j, \ell} (t, S_t) = \bar L_{i, j, \ell} (t, S_t^{k_{i, 0}})$. By Equation~\eqref{EQ:trans1_sep} we obtain
			\begin{equation*}
				L_{i, j, \ell} \left(\frac{a_{i, j, \ell}}{\si_t^\ell(d_i)^j}\right) =\sum_{\lambda = 1}^n \Delta_{x_\lambda} \left(h_{i, j, \ell}^{(\lambda)}\right) + \underbrace{\frac{\bar L_{i, j, \ell} (a_{i, j, \ell})}{\si_t^\ell(d_i)^j} }_{= 0};
			\end{equation*}
			{\em\bf return} {\em false} otherwise.
			
			\step {15}5 else
			\step {16}6 find a $\bK$-automorphism $\varphi_i$ of $\bK(t, \vx)$ given in Proposition~\ref{PROP:telescopers} such that $\varphi_i\circ\tau_{i,0}=\si_t\circ\varphi_i\text{ and } \varphi_i\circ\tau_{i,\lambda}=\si_{x_{\lambda}}\circ\varphi_i\text{ for } \lambda = 1, \ldots, r_i$.
			
			\step {17}6 set $\tilde a_{i, j, \ell} = \varphi_i (a_{i, j, \ell})$.
			\step {18}6 execute \text{\em\bf IsTelescoperable($\tilde a_{i, j, \ell}$, $[x_1, \ldots, x_{r_i}], t$)}.
			\step {19}6 if $\tilde a_{i, j, \ell}$ has a telescoper of type $(\si_t; \si_{x_1}, \ldots, \si_{x_{r_i}})$, let
			\begin{equation*}
				\tilde L_{i, j, \ell} (t, S_t) (\tilde a_{i, j, \ell})= \sum_{\lambda = 1}^{r_i} \Delta_{x_\lambda} \left(\tilde b_{i, j, \ell}^{(\lambda)}\right);
			\end{equation*}
			{\em \bf return} {\em false} otherwise.
			\step {20}6 apply $\varphi_{i}^{-1}$ to both sides of the previous equation to get
			\begin{equation*}
				\bar L_{i, j, \ell} (t, T_{i, 0}) (a_{i, j, \ell}) = \sum_{\lambda = 1}^{r_i} \Delta_{\tau_{i, \lambda}} \left(b_{i, j, \ell}^{(\lambda)}\right),
			\end{equation*}
			where $\bar L_{i, j, \ell} (t, T_{i, 0}) = \tilde L_{i, j, \ell} ({t}/{k_{i, 0}}, T_{i, 0})$ and $ b_{i, j, \ell}^{(\lambda)} = \varphi_i^{-1}(\tilde b_{i, j, \ell}^{(\lambda)})$ for all $\lambda = 1, \ldots, r_i$.
			\step {21}{6} set $L_{i, j, \ell} (t, S_t) = \bar L_{i, j, \ell} (t, S_t^{k_{i, 0}})$ and by Equations \eqref{EQ:trans1_sep} and~\eqref{EQ:trans2_tel} we obtain
			\begin{align*}
				L_{i, j, \ell} \left(\frac{a_{i, j, \ell}}{\si_t^\ell(d_i)^j}\right) & = \sum_{\lambda = 1}^n \Delta_{x_\lambda} \left(u_{i, j, \ell}^{(\lambda)}\right) + \frac{\bar L_{i, j, \ell} (a_{i, j, \ell})}{\si_t^\ell(d_i)^j}\nonumber\\
				& = \sum_{\lambda = 1}^n \Delta_{x_\lambda} \left(h_{i, j, \ell}^{(\lambda)}\right)
			\end{align*}
			for some $u_{i, j, \ell}^{(\lambda)}, h_{i, j, \ell}^{(\lambda)}\in \bK(t, \vx)$.
			\step {22}2 let $L\in \bK(t)\<S_t>$ be the LCLM of $L_{i, j, \ell}$ for all $i, j, \ell$ and write
			\[ L = R_{i, j, \ell} L_{i, j, \ell}\]
            for some $R_{i,j,\ell} \in \bK(t)\<S_t>$.
			\step {23}2 update $g_\lambda = L(g_\lambda) + \sum_{i = 1}^I \sum_{j = 1}^{J_i} \sum_{\ell = 0}^{s_{i, j}} R_{i, j, \ell} (h_{i, j, \ell}^{(\lambda)})$ for all $\lambda = 1, \ldots, n$.
			\step {24}2 {{\em\bf return}} $L$ and $g_1, \ldots, g_n$.
		\end{algorithm}

		\begin{example}\label{Eg:f_t}
			Let $\bK = \bQ$ and $f\in \bQ(t, x, y, z)$. We will decide constructively the existence of telescopers of type $(\si_t; \si_x, \si_y, \si_z)$ for various cases of $f$. Let $G_t = \<\si_t, \si_x, \si_y, \si_z>$ and $G = \<\si_x, \si_y, \si_z>$.
			\begin{enumerate}
				\item\label{it:f1_t} For $f = \frac{1}{(t + 1) (t + 2z) d}$ with $d = (t - 3y + x)^2(t + y)(t + z) + 1$, a basis of the isotropy group $G_{t, d}$ is $\{\tau_0\}$, where $\tau_0 = \si_t\si_x^{-4}\si_y^{-1}\si_z^{-1}$. Then $G_{t, d}/G_d = \<\bar\tau_0>$. Since $\rank (G_{t, d}/G_d) = 1$ and $G_d = \{\bf 1\}$ is a trivial group, we know from Theorem~\ref{THM:telescopers} that $f$ has a telescoper of type $(\si_t; \si_x, \si_y, \si_z)$ if and only if there exists a nonzero operator $L_0\in \bQ(t)\<T_0>$ with $T_0 = S_tS_x^{-4}S_y^{-1}S_z^{-1}$ such that
				\[L_0 (a) = 0, \text{ where } a = fd = \frac{1}{(t + 1) (t + 2 z)}.\]
				Note that the prime part of the denominator $b = (t + 1) (t + 2z)$ of $a$ with respect to the variables $\{x, y, z\}$ is $b_2 = t + 2z$ and $\tau_0(b_2) \neq b_2$. By Proposition~\ref{PROP:separation}, there does not exist any operator $L_0\in \bQ(t)\<T_0>$ such that $L_0 (a) = 0$. So $f$ does not have any telescoper of type $(\si_t; \si_x, \si_y, \si_z)$.
				\item\label{it:f2_t} For $f = \frac{1}{(t + 1)d}$ with $d$ being the same as in Example~\ref{Eg:f_t}~(\ref{it:f1_t}), it is easy to check that for $a = \frac{1}{t + 1}$,
				\begin{equation*}%\label{EQ:L0a}
					L_0(a) = 0\text{ with } L_0 = T_0 - \frac{t + 1}{t + 2},
				\end{equation*}
				where $T_0 = S_tS_x^{-4}S_y^{-1}S_z^{-1}$. So by Theorem~\ref{THM:telescopers}, $f$ has a telescoper $L$ of type $(\si_t; \si_x, \si_y, \si_z)$. In fact, we can take $L = S_t - \frac{t + 1}{t + 2}$. Then
				\begin{align*}
					L(f) &= \frac{\si_t(a)}{\si_t(d)} - \frac{t + 1}{t + 2}\cdot \frac{a}{d} = \frac{\si_t(a)}{\si_x^4\si_y\si_z(d)} - \frac{t + 1}{t + 2}\cdot \frac{a}{d}\\
					& = \Delta_x (u) + \Delta_y (v) + \Delta_z (w) + \underbrace{\frac{\tau_0(a) - ((t + 1)/(t + 2)) a}{d}}_{=L_0(a)/d = 0}\\
					& = \Delta_x (u) + \Delta_y (v) + \Delta_z (w),
				\end{align*}
				where $u = \sum_{\ell = 0}^3\frac{\si_t(a)}{\si_x^{\ell}\si_y\si_z(d)}$, $v = \frac{ \si_t(a)}{\si_z(d)}$, and $w = \frac{\si_t(a)}{d}$. Additionally, this is a non-trivial example in two respects. Firstly, since $G_d = \{\bf 1\}$, this rational function $f$ is not $(\si_x, \si_y, \si_z)$-summable in $\bQ(t, x, y, z)$. Secondly, for any $\{\mu, \nu\}\subseteq\{x, y, z\}$, since the isotropy group of $d$ in $\<\si_t, \si_\mu, \si_\nu>$ is trivial and $f$ is not $(\si_\mu, \si_\nu)$-summable, by Lemma~\ref{LEM:red-ts1}, $f$ does not have any telescoper in $\bQ(t)\<S_t>$ of type $(\si_t; \si_\mu, \si_\nu)$.
				
				\item\label{it:r1_t} We continue Example~\ref{Eg:f_red_t}~(\ref{it:f2-t}) and write $f$ in the form
				\[f = \Delta_x (u_0) + \Delta_y (v_0) + \Delta_z (w_0) + r_1 + r_2,\]
				where $u_0, v_0, w_0\in\bQ(t, x, y, z)$ and $r_1=\frac{1}{t (t + y + 2z) d}$, $r_2 = \frac{1}{ (t + 3z) \si_t(d)}$ with $d = 3y + (x + z)^2 + t$.
				\begin{enumerate}[label = (\alph*)]
					\item For $r_1 = a_1/d$ with $a_1 = 1/(t (t + y + 2z))$, we find that a basis of $G_{t, d}$ is $\{\tau_0, \tau_1\}$, where $\tau_0 = \si_t^3\si_y^{-1}$ and $\tau_1 = \si_x\si_z^{-1}$. Then by Theorem~\ref{THM:telescopers}, $r_1$ has a telescoper of type $(\si_t; \si_x, \si_y, \si_z)$ if and only if $a_1$ has a telescoper of type $(\tau_0; \tau_1)$. Choose one $\bQ$-automorphism $\phi_1$ of $\bQ(t, x, y, z)$ given in Proposition~\ref{PROP:telescopers} as follows
					\[\phi_1(h(t, x, y, z)) = h(3t, x, -t + y, - x + z),\]
					for any $h\in \bQ(x, y, z)$. Then $\phi_1\circ \tau_0 = \si_t \circ \phi_1$ and $\phi_1\circ \tau_1 = \si_x \circ \phi_1$. So $a_1$ has a telescoper of type $(\tau_0;\tau_1)$ if and only if $\phi_1(a_1)$ has a telescoper of type $(\si_t; \si_x)$. A direct calculation yields that
					\[\phi_1(a_1) = \frac{1}{3t(\underbrace{2t + y - 2x + 2z}_{\tilde d})}.\]
					Again consider the isotropy group of $\tilde d$ in $\<\si_t, \si_x>$, which is generated by $\tilde\tau_0 = \si_t\si_x^2$. Since $(\tilde\tau_0 - \frac{t}{t + 1}) (\frac{1}{3t}) = 0$, by similar arguments used in Example~\ref{Eg:f_t}~(\ref{it:f2_t}), we see that $\phi_1(a_1)$ has a telescoper $\tilde L_1\in \bQ(t)\<S_t>$ of type $(\si_t; \si_x)$ and in particular we find
					\begin{equation}\label{EQ:Lphi1a1}
						\tilde L_1 (t, S_t) (\phi_1 (a_1)) =  \Delta_x(\tilde b_1)
					\end{equation}
					with $\tilde L_1 = S_t - \frac{t}{t + 1}$ and $\tilde b_1 = - \frac{1}{3(t  + 1)(2t + y - 2x + 2 + 2z)}$.
					So by Theorem~\ref{THM:telescopers}, $r_1$ has a telescoper $L\in \bQ(t)\<S_t>$ of type $(\si_t; \si_x, \si_y, \si_z)$. In fact, we can find an explicit expression for $L$. Applying $\phi_1^{-1}$ to Equation~\eqref{EQ:Lphi1a1} yields that
					\[\bar L_1 (t, T_0) (a_1) = \Delta_{\tau_1}(b_1),\]
					where $T_0 = S_t^3S_y^{-1}$, $\bar L_1(t, T_0) = \tilde L_1 (\frac{t}{3}, T_0) = T_0 - \frac{t}{t + 3}$, $b_1 =\phi_1^{-1}(\tilde b_1) =-\frac{1}{(t + 3)(t + y + 2z + 2)}$. Let $L_1 (t, S_t) = \bar L_1 (t, S_t^3) = S_t^3 - \frac{t}{t + 3}$. Then we have
					\begin{align*} L_1 (r_1) &= \frac{\si_t^3(a_1)}{\si_t^3(d)} - \frac{t}{t + 3}\cdot\frac{a_1}{d} = \frac{\si_t^3(a_1)}{\si_y(d)} - \frac{t}{t + 3}\cdot\frac{a_1}{d}\\
						&= \Delta_{x} (0) + \Delta_y (v_1) + \Delta_z (0) + \frac{\bar L _1 (a_1)}{d}\text{ with } v_1 = \frac{\si_t^3\si_y^{-1}(a_1)}{d}
					\end{align*}
					and using Lemma~\ref{Formula} with $\tau = \tau_1$, we get
					\[\frac{\bar L_1(a_1)}{d} = \Delta_{\tau_1}\left(\frac{b_1}{d}\right) = \Delta_x (u_1) + \Delta_y (0) + \Delta_z (w_1)\]
					with $u_1 =\si_z^{-1}\left(\frac{b_1}{d}\right)$ and $w_1 = -\si_z^{-1}\left(\frac{b_1}{d}\right)$. Hence $L_1$ is a telescoper of type $(\si_t; \si_x, \si_y, \si_z)$ for $r_1$ and \begin{equation*}
						L_1 (r_1) = \Delta_x (u_1) + \Delta_y (v_1) + \Delta_z (w_1).
					\end{equation*}
					\item Similarly, for $r_2 = a_2/ \si_t(d)$ with $a_2 = 1/(t + 3z)$, we apply the algorithm {\em IsTelescoperable} to $r_2$. The result is true and we obtain
					\begin{equation*}
						L_2 (r_2) = \Delta_x (u_2) + \Delta_y (v_2) + \Delta_z (w_2),
					\end{equation*}
					where $L_2 = S_t^3 - 1$, $u_2 = \si_z^{-1}\left(\frac{b_2}{\si_t(d)}\right)$, $v_2 = \frac{\si_t^3\si_y^{-1}(a_2)}{\si_t(d)}$ and $w_2 = -\si_z^{-1}\left(\frac{b_2}{\si_t(d)}\right)$ with $b_2 = -\frac{1}{t + 3z + 3}$.
					\item For $r = r_1 + r_2$, using the LCLM algorithm to compute the least common multiple $L$ of $L_1$ and $L_2$ in $\bQ(t)\<S_t>$, we obtain
					\[ L =  R_1 L_1 = R_2 L_2 = S_t^6 - \frac{2(t + 3)}{t + 6} S_t^3 + \frac{t}{t + 6}\]
					with $R_1 =S_t^3 -\frac{t + 3}{t + 6}$ and $R_2 = S_t^3 -\frac{t}{t + 6}$. Then
					\[L(r) =  \Delta_x (u) + \Delta_y (v) + \Delta_z (w),\]
					where $u =\sum_{i = 1}^2 R_i(u_i)$, $v =\sum_{i = 1}^2 R_i(v_i)$ and $u =\sum_{i = 1}^2 R_i(w_i)$ are rational functions in $\bQ(t, x, y, z)$. Updating $u = u + L(u_0)$, $v = v + L(v_0)$ and $w = w + L(w_0)$, we get
					\[L(f) = \Delta_x (u) + \Delta_y (v) + \Delta_z (w).\]
					So $L$ is a telescoper of type $(\si_t; \si_x, \si_y, \si_z)$ for $f$.
				\end{enumerate}
			\end{enumerate}
		\end{example}

  \subsection{Examples and applications}

    Creative telescoping is a powerful tool for proving combinatorial identities algorithmically~\cite{PWZbook1996}. The following example shows an application of telescopers for multivariate rational functions.

    \begin{example}\label{Eg:application}
        We show that $$F(t)=\sum_{x=0}^t\sum_{y=0}^t\sum_{z=0}^tf(t,x,y,z)=0,$$
        where
        $$f(t,x,y,z)=\frac{(2y-t)(2x-t)(2z-t)}{(y+t+1)(-2t+y-1)(x+t+1)(-2t+x-1)(z+t+1)(-2t+z-1)}.$$
        Applying Algorithm~\ref{alg:telescoper} to $f$, we find that $f$ has a telescoper $L=S_t-1$ of type $(\si_t;\si_x,\si_y,\si_z)$ with certificates $(u,v,w)$, where
        $$u=\frac{(-2y + t + 1)(-2z + t + 1)(8t^2 - 2tx - x^2 + 19t - 2x + 11)}{(x + t + 1)(2t - x + 2)(2t - x + 3)(y + t + 2)(2t - y + 3)(z + t + 2)(2t - z+3)},$$
        $$v=\frac{(-2x + t)(-2z + t + 1)(8t^2 - 2ty - y^2 + 19t - 2y + 11)}{(x + t + 1)(2t - x + 1)(y+t + 1)(2t - y + 2)(2t  - y+3)(z + t + 2)(2t - z+ 3)}$$
        and
        $$w=\frac{(-2x + t)(-2y + t)(8t^2 - 2tz - z^2 + 19t - 2z + 11)}{(x + t + 1)(2t - x + 1)(y+t + 1)(2t - y + 1)(z + t + 1)(2t - z + 2)(2t - z + 3)}.$$
        Thus we have
        \begin{align*}
        &\sum_{x=0}^t\sum_{y=0}^t\sum_{z=0}^t\left(f(t+1,x,y,z)-f(t,x,y,z)\right)\\
        =&\sum_{x=0}^t\sum_{y=0}^t\sum_{z=0}^t\left(\Delta_x(u)+\Delta_y(v)+\Delta_z(w)\right)\\
        =&\sum_{y=0}^t\sum_{z=0}^t(u(t,t+1,y,z)-u(t,0,y,z))
        +\sum_{x=0}^t\sum_{z=0}^t(v(t,x,t+1,z)-v(t,x,0,z))\\
        &+\sum_{x=0}^t\sum_{y=0}^t(w(t,x,y,t+1)-w(t,x,y,0)).
        \end{align*}
        Then applying $L$ to $F(t)$ yields
        \begin{align*}
        &F(t+1)-F(t)\\
        =&\sum_{y=0}^t\sum_{z=0}^t\underbrace{(u(t,t+1,y,z)-u(t,0,y,z)+f(t+1,t+1,y,z))}_{g_1}\\
        &+\sum_{x=0}^t\sum_{z=0}^t\underbrace{(v(t,x,t+1,z)-v(t,x,0,z)+f(t+1,x,t+1,z))}_{g_2}\\
        &+\sum_{x=0}^t\sum_{y=0}^t\underbrace{(w(t,x,y,t+1)-w(t,x,y,0)+f(t,x,y,t+1))}_{g_3}\\        &+\sum_{x=0}^tf(t+1,x,t+1,t+1)+\sum_{y=0}^tf(t+1,t+1,y,t+1)+\sum_{z=0}^tf(t+1,t+1,t+1,z)\\
        &+f(t+1,t+1,t+1,t+1).
        \end{align*}
        So we reduce the triple sum to double sums and single sums. One can check that $g_1=0$. By Algorithm~\ref{Alg:summable}, one can find that $g_2$ is $\si_x$-summable and $g_3$ is $(\si_x,\si_y)$-summable. Similarly, we further reduce the double sums to the single sums.
        Applying the Algorithm~\ref{Alg:summable} (specialized to the univariate case) again, we simplify the single sums and finally obtain that $F(t+1)-F(t)=0$. Since the initial value of $F(t)$ is $F(0)=f(0,0,0,0)=0$, we conclude that $F(t) = 0$. This completes the proof.
    \end{example}

    Under some assumptions, there are several packages to compute the creative telescoping in more than two variables. The Mathematica package ``HolonomicFunctions" developed by Koutschan contains two functions ``CreativeTelescoping" and ``FindCreativeTelescoping" to construct telescopers for holonomic functions in different ways~\cite{koutschan2010holonomicfunctions}. Another Mathematica function ``FindRecurrence", the core of the Mathematica package ``MultiSum" by Wegschaider, is designed to find telescopers for proper hypergeometric functions~\cite{Wegschaider1997}. For rational functions in three variables, an effective algorithm has been presented in~\cite{chen2022AAM} to compute their minimal telescopers.

    Experiments show that Algorithm~\ref{alg:telescoper} is more efficient to test the existence of telescopers and construct one telescoper. For example, for the rational function $$f(t,x,y)=\frac{4t+2}{(45t + 5x + 10y + 47)(45t + 5x + 10y + 2)(63t-5x + 2y + 58)(63t-5x + 2y-5)},$$
    Algorithm~\ref{alg:telescoper} tells us that it has a telescoper of type $(\si_t;\si_x,\si_y)$ and outputs a telescoper and its corresponding certificates within two seconds in Maple, while the algorithm in~\cite{chen2022AAM} takes about three minutes and the other three functions of the two Mathematica packages use much more timings as shown in~\cite{chen2022AAM}. Given a rational function, one could use our algorithm to pre-check the existence of its telescopers and find a telescoper if such a telescoper exists. After that one may apply the other efficient methods to find a telescoper with lower degree in $S_t$.

\section{Implementations and timings}\label{sec:appendix}
		
		\lstset{basicstyle=\ttfamily,breaklines=true,columns=flexible}
		\DefineParaStyle{Maple Bullet Item}
		\DefineParaStyle{Maple Heading 1}
		\DefineParaStyle{Maple Warning}
		\DefineParaStyle{Maple Heading 4}
		\DefineParaStyle{Maple Heading 2}
		\DefineParaStyle{Maple Heading 3}
		\DefineParaStyle{Maple Dash Item}
		\DefineParaStyle{Maple Error}
		\DefineParaStyle{Maple Title}
		\DefineParaStyle{Maple Text Output}
		\DefineParaStyle{Maple Normal}
		\DefineCharStyle{Maple 2D Output}
		\DefineCharStyle{Maple 2D Input}
		\DefineCharStyle{Maple Maple Input}
		\DefineCharStyle{Maple 2D Math}
		\DefineCharStyle{Maple Hyperlink}
		
We have implemented Algorithms \ref{Alg:SET}, \ref{Alg:summable} and~\ref{alg:telescoper} in the computer algebra system Maple 2020. In this section, we compare the efficiency of the algorithms for solving the SET problem and illustrate the usage of our package ``RationalWZ" by several examples. Our maple code and more examples are available for download at
\begin{center}
	\href{http://www.mmrc.iss.ac.cn/~schen/RationalWZ-2022.html}{\url{http://www.mmrc.iss.ac.cn/~schen/RationalWZ-2022.html}}
\end{center}

 We have implemented the G algorithm, the KS algorithm, the DOS algorithm, and the algorithm applying $\va$-degree cover to Algorithm~\ref{Alg:SET} (ADC) in Maple 2020 with
		$\bF=\bQ$.
		
		Fixing one admissible cover, there are two methods to calculate it and then to implement Algorithm~\ref{Alg:SET}. A direct method is expanding $p(\vx+\va)-q(\vx)$ with $2n$ variables to get the set of its coefficients in $\vx$ and then the admissible cover, while another is obtaining the members of the admissible cover successively by computing partial derivatives dynamically. For efficiency, the DOS algorithm and the ADC algorithm are realized by partial derivatives and expansion respectively.
		
		The test suite was generated as follows.
		
		Let $n,d,\mu,d'\in\bN$ and $d'<d$. Let $\vx=\{x_1,x_2,\ldots,x_n\}$. We first generated randomly a $\mu$-term polynomial $p(\vx)$ of degree $d$, as well as a polynomial $dis(\vx)$ of degree $d'$. Then we generated randomly a vector $\vs=(s_1,s_2,\ldots,s_n)\in{\bZ}^n$ and let $q(\vx)=p(\vx+\vs)+dis(\vx)$. Therefore, the calculation is most likely to terminate after computing $\bV_{\bF}(\cup_{i=0}^{d-1-d'}S^H_{i})$. By setting $0\leq s_i \leq 99$, we can avoid the case where the memory is not enough to complete the computation.
		
		Note that, in all the tests, the algorithms take the expanded forms of examples given above as input. All timings are measured in seconds on a macOS Monterey (Version 12.0.1) MacBook Pro with 32GB Memory and Apple M1 Pro Chip.

		For a selection of random polynomials and vectors for different
		choices of $n,\mu,d,d'$ as above, we first tabulate the timings of the
		G algorithm, the KS algorithm, the DOS algorithm and the ADC algorithm. Note that $d'=-\infty$ means $dis=0$, implying that $p$ is shift equivalent to $q$.
		
		\begin{center}
			\begin{tabular}{cccc|cccc}
				\cline{1-8}
				$n$ & $\mu$  & $d$  & $d'$ & G            & KS    & DOS        & ADC     \\\cline{1-8}
				3 & 10  & 15 & 13     & 5.476                 & 2.090                  & 0.014     & 0.008              \\
				3 & 10  & 15 & 10     & 0.243                 & 1.124                  & 0.023     & 0.020                  \\
				3 & 10  & 15 & 5      & 21.719                & 1.809                  & 0.050     & 0.032                \\
				3 & 10  & 15 & 0      & 573.178               & 2.576                  & 0.068     & 0.039                \\
				3 & 10  & 15 & $-\infty$    & 18.491          & 0.714                  & 0.043     & 0.036                  \\
				3 & 100 & 15 & 13     & 0.205                 & 10.025                 & 0.044     & 0.028               \\
				3 & 100 & 15 & 10     & 0.482                 & 9.997                 & 0.046      & 0.046             \\
				3 & 100 & 15 & 5      & 22.114                & 11.317                 & 0.061     & 0.062  \\
				3 & 100 & 15 & 0      & 2152.378              & 19.470                 & 0.083     & 0.069       \\
				3 & 100 & 15 & $-\infty$   & 1200.473          & 13.640                 & 0.085    & 0.068   \\
				\cline{1-8}
			\end{tabular}
		\end{center}

		The experimental results illustrate that the DOS algorithm and the ADC algorithm outperform the other two algorithms. Furthermore, we conducted experiments in more complicated cases.
		
		\begin{center}
			\begin{tabular}{cccc|ccc}
				\cline{1-7}
				$n$ & $\mu$  & $d$  & $d'$ & DOS    & ADC \\ \cline{1-7}
				5 & 100   & 40 & 35  & 199.177          & 59.889                     \\
				5 & 100   & 40 & 30  & 24.684                      &  90.159                     \\
				5 & 100   & 40 & 20  & 379.835                     & 95.761                    \\
				5 & 100   & 40 & 10  & 681.189  & 665.885      \\
				5 & 100   & 40 & 0   & 182.671 & 67.261                    \\
				5 & 100   & 40 & $-\infty$ & 709.223 & 77.880    \\
				5 & 10000 & 20 & 18  & 2.724                       &  122.744                    \\
				5 & 10000 & 20 & 15  & 3.088                       & 163.258                    \\
				5 & 10000 & 20 & 10  & 5.290                       &  139.685                   \\
				5 & 10000 & 20 & 5   & 10.755                      &  125.359                   \\
				5 & 10000 & 20 & 0   & 23.949                      &  151.010 \\
				5 & 10000 & 20 & $-\infty$ & 24.562                      &  136.187                 \\
				\cline{1-7}
			\end{tabular}	
		\end{center}
		
		The experimental results indicate that the ADC algorithm outperforms the other for most of non-dense testing examples, while the DOS algorithm has a clear advantage for dense ones. It may be because the timing of expansion grows fast with the number of terms in the input polynomial. In conclusion, we present an algorithm, the ADC algorithm, which is complementary to the DOS algorithm for solving the SET problem.

From the runtime comparison, we implemented the ADC algorithm in the package RationalWZ. In our setting, the base field $\bF$ can be $\bQ$ or the field of rational functions $\bQ(u_1,\ldots,u_s)$.
The following instructions show how to load the modules.
\begin{lstlisting}
> read "RationalWZ.mm";
> with(ShiftEquivalenceTesting):
> with(OrbitalDecomposition):
> with(RationalReduction):
> with(RationalSummation):
> with(RationalTelescoping):
\end{lstlisting}
\begin{example} Compute the dispersion set (over $\bZ$) of two polynomials.
\begin{enumerate}
    \item For $p,q\in \bQ[x,y]$ in Example~\ref{EX:SET}, we type
    \begin{lstlisting}
> ShiftEquivalent(x^2 + 2*x*y + y^2 + 2*x + 6*y, x^2 + 2*x*y + y^2 + 4*x + 8*y + 11, [x, y])
    \end{lstlisting}
    \begin{maplelatex}
    \mapleinline{inert}{2d}{[-1, 2]}{\[[-1,2]\]}
    \end{maplelatex}
    which implies $Z_{p,q}=\{(-1,2)\}$. So $p(x-1, y+2)=q(x,y)$.
    \item For $p,q\in \bQ[x,y,z]$ in Example~\ref{EX:SET_compare}~(\ref{it:p1q1}) , we type
    \begin{lstlisting}
> ShiftEquivalent(x^4 + x^3*y + x*y^2 + z^2, x^4 + x^3*(y + 1) + x*(y + 1)^2 + (z + 2)^2 + x*y, [x, y, z])
    \end{lstlisting}
    \mapleresult
    \begin{maplelatex}
    \mapleinline{inert}{2d}{[]}{\[[\,]\]}
    \end{maplelatex}
    which implies $Z_{p,q}=\varnothing$. So $p,q$ are not shift equivalent.
\end{enumerate}
\end{example}
\begin{example} Decide the $(\si_x,\si_y,\si_z)$-summability of a rational function $f\in \bQ(x,y,z)$. Let $f_3$, $r_3$ be the same as in Example~\ref{Eg:f}~(\ref{it:f3}).
\begin{enumerate}
    \item Applying the function ``IsSummable" to $f=f_3$, we see that $f$ is not $(\si_x,\si_y,\si_z)$-summable.
    \begin{lstlisting}
> IsSummable((y + z/(y^2 + z - 1) - 1/(y^2 + z))/(x + 2*y + z)^2, [x, y, z])
    \end{lstlisting}
\mapleresult
\mapleresult
\begin{maplelatex}
\mapleinline{inert}{2d}{false}{\[{\it false}\]}
\end{maplelatex}
\item Applying the function ``IsSummable" to $f=f_3-r_3$, we see that $f$ is $(\si_x,\si_y,\si_z)$-summable and its certificates are as follows:
\begin{lstlisting}
> IsSummable((y + z/(y^2 + z - 1) - 1/(y^2 + z))/(x + 2*y + z)^2 - z/((y^2 + z)*(x + 2*y + z)^2), [x, y, z])
\end{lstlisting}
\begin{maplelatex}
\mapleinline{inert}{2d}{true, [-(1/2)*y*(y-1)/(x-2+2*y+z)^2-(1/2)*y*(y-1)/(x-1+2*y+z)^2+z/((y^2+z-1)*(x-1+2*y+z)^2), (1/2)*y*(y-1)/(x-2+2*y+z)^2, -z/((y^2+z-1)*(x-1+2*y+z)^2)]}{\[{\it true},\,\left[-{\frac {y \left( y-1 \right) }{2\, \left( x-2+2\,y+z \right) ^{2}}}-{\frac {y \left( y-1 \right) }{2\, \left( x-1+2\,y+z \right) ^{2}}}+{\frac {z}{ \left( {y}^{2}+z-1 \right)  \left( x-1+2\,y+z \right) ^{2}}}\\
\mbox{},{\frac {y \left( y-1 \right) }{2\, \left( x-2+2\,y+z \right) ^{2}}},-{\frac {z}{ \left( {y}^{2}+z-1 \right)  \left( x-1+2\,y+z \right) ^{2}}}\right]\]}
\end{maplelatex}
\end{enumerate}
\end{example}

\begin{example}
Decide the existence of telescopers of type $(\si_t;\si_x,\si_y,\si_z)$ for a rational function $f\in\bQ(t,x,y,z)$.
\begin{enumerate}
    \item Applying the function ``IsTelescoperable" to $f$ in Example~\ref{Eg:f_t}~(\ref{it:f1_t}), we see that $f$ does not have a telescoper in $\bQ(t)\<S_t>$ of type $(\si_t;\si_x,\si_y,\si_z)$.
    \begin{lstlisting}
> IsTelescoperable(1/((t + 1)*(t + 2*z)*((t - 3*y + x)^2*(t + y)*(t + z) + 1)), [x, y, z], t, 'St')
\end{lstlisting}
\mapleresult
\mapleresult
\begin{maplelatex}
\mapleinline{inert}{2d}{T1 := false}{\[\,{\it false}\]}
\end{maplelatex}
\item Applying the function ``IsTelescoperable" to $f=r_1$ in Example~\ref{Eg:f_t}~(\ref{it:r1_t}), we see that $f$ has a telescoper $L=-\tfrac{t}{t+3}+S_t^3$ of type $(\si_t;\si_x,\si_y,\si_z)$ and its certificates are as follows: \begin{lstlisting}
> IsTelescoperable(1/(t*(t + y + 2*z)*(3*y + (x + z)^2 + t)), [x, y, z], t, 'St')
\end{lstlisting}
\mapleresult
\mapleresult
\begin{maplelatex}
\mapleinline{inert}{2d}{T1 := true, -t/(t+3)+St^3, [-(1/6)/(((1/3)*t+1)*((1/2)*t+(1/2)*y+z)*(x^2+2*x*(z-1)+(z-1)^2+t+3*y)), 1/((t+3)*(t+2+y+2*z)*(x^2+2*x*z+z^2+t+3*y)), (1/6)/(((1/3)*t+1)*((1/2)*t+(1/2)*y+z)*(x^2+2*x*(z-1)+(z-1)^2+t+3*y))]}{\[{\it true},\,-{\frac {t}{t+3}}+{{\it St}}^{3},\,\left[-{\frac {1}{6\left( {\frac {t}{2}}+{\frac {y}{2}}+z \right)\left( {\frac {t}{3}}+1 \right) \left({x}^{2}+2\,x \left( z-1 \right) +\, \left( z-1 \right) ^{2}+\,t+3\,y\right)}},\right.\]}
\mapleinline{inert}{2d}{...}{\[\quad\quad\left.{\frac {1}{ \left( t+3 \right) \\
\mbox{} \left( t+2+y+2\,z \right)  \left( {x}^{2}+2\,xz+{z}^{2}+t+3\,y \right) }},\frac {1}{6\left( {\frac {t}{3}}+1 \right)\left( {\frac {t}{2}}+{\frac {y}{2}}+z \right)\left(\,{x}^{2}+2\,x \left( z-1 \right)\\
\mbox{}+\, \left( z-1 \right) ^{2}+\,t+3\,y\right)}\right] \]}
\end{maplelatex}

\end{enumerate}
\end{example}

\section{Conclusion and future work}	

In this paper, we constructively solve the summability problem and the existence problem of telescopers for multivariate rational functions, and present a new efficient algorithm for solving the shift equivalence testing (SET) problem of multivariate polynomials.

Our algorithm can compute a telescoper for a given multivariate rational function if the existence of telescopers is guaranteed, but the computed telescoper may not be of minimal order. So a natural question is how to compute the minimal telescoper (which is unique if it is monic) for multivariate rational function if it exists. Similar to the trivariate case~\cite{chen2022AAM}, we may first need an additive decomposition to decompose a rational function $f$ as a sum of a summable function and a remainder $r$, as shown in Example~\ref{Eg:f}, such that $f$ is summable if and only if the remainder $r$ is zero. Then we need to deal with the problem that the sum of two remainders in the additive decomposition may not be a remainder. The similar problem appears and has been solved in the case for trivariate rational functions~\cite[Section 4]{chen2022AAM} and the case for bivariate hypergeometric terms~\cite[Section 5]{Chen2015b}.

For the efficiency, we may need to consider how to choose a ``minimal'' remainder, which may depend on the choices of difference isomorphisms. Choosing a ``good'' admissible cover may help us to discover a more efficient SET algorithm. In theory, we use an irreducible partial fraction decomposition in summation algorithms, but in practice, an incomplete partial fraction decomposition would be enough, like in the univariate case~\cite{Paule1995b,Arreche24}.

In the future research, we hope to explore more summation algorithms for other classes of functions, like multivariate hypergeometric terms~\cite{Wegschaider1997}. This would be an extension of Gosper's algorithm~\cite{Gosper1978} which only works for the summation of univariate hypergeometric terms and has many applications in proving combinatorial identities~\cite{PWZbook1996}. Some algorithms have been developed for special bivariate hypergeometric terms~\cite{ChenHouMu2006} and for multiple binomial sums~\cite{BostanLairezSalvy2017}.

In the differential case, telescopers always exist for D-finite functions~\cite{Zeilberger1990}. One interesting problem is how to find a telescoper~\cite{BostanLairezSalvy2017,CKS2012}, especially the minimal one. Another problem is the integrability problem proposed by Picard~\cite{Picard1933,Picard1897,Picard1899}, which is a continuous analogue of the summability problem. Given a rational function $f\in \bF(x_1,\ldots, x_m)$, the integrability problem is deciding whether there exist rational functions $g_1,\ldots,g_m\in\bF(x_1,\ldots, x_m)$ such that
\[f = \partial_{x_1}(g_1) + \cdots + \partial_{x_m}(g_m),\]
where $\partial_{x_i}$ is the usual partial derivative with respect to $x_i$. When $m=1$, it can be solved by Ostrogradsky-Hermite reduction~\cite{Ostrogradsky1845,Hermite1872}. When $m=2$, it was solved by Picard~\cite[p. 475--479]{Picard1897}. In more than two variables, there is no complete algorithm for deciding the integrability of rational functions. Under a regularity assumption, some related results are listed in~\cite{BostanLairezSalvy2017}.

\section*{Acknowledgement}
We would like to thank the anonymous referees for their careful readings and many constructive suggestions and thank the editors for their nice
suggestions on the way of presenting our main results.  We would also like to express our gratitude to Yisen Wang for the insightful discussions in the complexity estimates.

\bibliographystyle{plain}

	\end{document}